\documentclass[11pt]{article}
\usepackage[usenames,dvipsnames,svgnames,table]{xcolor}
\usepackage{enumitem}
\usepackage{graphicx}
\usepackage{fancybox}
\usepackage{comment}
\usepackage{xcolor}
\usepackage{nameref}
\usepackage{bbm}

\definecolor{ForestGreen}{rgb}{0.1333,0.5451,0.1333}
\definecolor{DarkRed}{rgb}{0.8,0,0}
\definecolor{Red}{rgb}{1,0,0}
\usepackage[linktocpage=true,
pagebackref=true,colorlinks,
linkcolor=ForestGreen,citecolor=ForestGreen,
bookmarks,bookmarksopen,bookmarksnumbered]
{hyperref}
\usepackage[nottoc,numbib]{tocbibind}

\usepackage{amsmath}
\usepackage{amssymb}
\usepackage{amsthm}

\usepackage[ruled, noend, linesnumbered]{algorithm2e}

\usepackage{subfig}

\usepackage{thm-restate}

\usepackage{float}
\usepackage{cleveref}

\newtheorem{theorem}{Theorem}[section]
\newtheorem{informaltheorem}[theorem]{Informal Theorem}

\newtheorem{corollary}[theorem]{Corollary}
\newtheorem{lemma}[theorem]{Lemma}

\newtheorem{claim}[theorem]{Claim}

\newtheorem{definition}[theorem]{Definition}
\newtheorem{remark}[theorem]{Remark}

\newtheorem*{theorem*}{Theorem}
\newtheorem*{corollary*}{Corollary}
\newtheorem*{conjecture*}{Conjecture}
\newtheorem*{lemma*}{Lemma}
\newtheorem*{thm*}{Theorem}
\newtheorem*{prop*}{Proposition}
\newtheorem*{obs*}{Observation}
\newtheorem*{definition*}{Definition}

\newtheorem*{remark*}{Remark}
\newtheorem*{rec*}{Recommendation}

\newenvironment{fminipage}%
  {\begin{Sbox}\begin{minipage}}%
  {\end{minipage}\end{Sbox}\fbox{\TheSbox}}

\def\defeq{\stackrel{\mathrm{def}}{=}}
\def\setof#1{\left\{#1  \right\}}

\def\union{\cup}

\def\abs#1{\left|#1  \right|}

\newcommand\vecone{\boldsymbol{1}}

\renewcommand{\deg}{\operatorname{deg}}

\newcommand\Otil{\widetilde{O}}

\newcommand\R{\mathbb{R}}

\DeclareMathOperator*{\argmin}{arg\,min}

\DeclareMathOperator*{\im}{im}

\newcommand{\polylog}{\text{ polylog }} 
\newcommand{\concat}{\oplus}
\newcommand{\econg}{\text{econg}}

\newcommand{\diam}{diam}

\newcommand{\eps}{\varepsilon}
\renewcommand{\O}{\widetilde{O}}

\renewcommand{\forall}{\mathrm{\text{ for all }}}

\newcommand{\bx}{\boldsymbol{x}}

\renewcommand{\root}{\mathsf{root}}

\renewcommand{\hat}{\widehat}
\renewcommand{\tilde}{\widetilde}

\DeclareFontFamily{U}{mathb}{\hyphenchar\font45}
\DeclareFontShape{U}{mathb}{m}{n}{<5> <6> <7> <8> <9> <10> gen * mathb
<10.95> mathb10 <12> <14.4> <17.28> <20.74> <24.88> mathb12}{}
\DeclareSymbolFont{mathb}{U}{mathb}{m}{n}
\DeclareMathSymbol{\rcirclearrow}{\mathbin}{mathb}{'367}

\newcommand{\wt}{\widetilde}

\renewcommand{\bar}{\overline}

\newif\ifrandom
\randomtrue

\newcommand{\poly}{{\mathrm{poly}}}

\newcommand{\str}{{\mathsf{str}}}

\newcommand{\wstr}{{\wt{\str}}}

\newcommand{\todolater}[1]{}

\newcommand{\cC}{\mathcal{C}}

\newcommand{\cW}{\mathcal{W}}

\newcommand{\projP}{\tilde{\mathcal{P}}}

\DeclareUnicodeCharacter{2113}{$\ell$}

\renewcommand{\root}{\mathsf{root}}

\newcommand{\gammaVSApprox}{\gamma_{approxVS}}

\newcommand{\Apxball}{\textsc{ApxBall}}
\newcommand{\stretchSSSP}{\gamma}
\newcommand{\cover}{\textsc{Cover}}

\newcommand{\Core}{\textsc{Core}}

\newcommand{\Htil}{\tilde{H}}
\newcommand{\ball}{\bar{B}}
\usepackage{fullpage}

\DeclareMathOperator{\dist}{dist}

\title{\vspace{-3em}
A Dynamic Shortest Paths Toolbox:\\
Low-Congestion Vertex Sparsifiers and their Applications
}
\date{}
\newcommand*\samethanks[1][\value{footnote}]{\footnotemark[#1]}
\author{Rasmus Kyng\thanks{The research leading to these results has received funding from grant no. 200021 204787 of the Swiss National Science Foundation.} \\ ETH Zurich \\ kyng@inf.ethz.ch \\  \and Simon Meierhans\samethanks \\ ETH Zurich \\ mesimon@inf.ethz.ch \and Maximilian Probst Gutenberg\samethanks \\ ETH Zurich \\ maximilian.probst@inf.ethz.ch}

\begin{document}
\pagenumbering{gobble}

\maketitle

\vspace{-2 em} 
\begin{abstract}
We present a general toolbox, based on vertex sparsifiers, for designing new data structures to maintain shortest paths in graphs undergoing edge insertions and/or deletions.
In particular, we obtain the following results:
\begin{itemize}
    \item the first data structure to maintain $m^{o(1)}$-approximate all-pairs shortest paths (APSP) in an $m$-edge graph undergoing edge insertions and deletions with \emph{worst-case} update time $m^{o(1)}$ and query time $\tilde{O}(1)$, and %
    \item a data structure to maintain a tree $T$ that has diameter no larger than a subpolynomial factor than the underlying graph $G$ that is undergoing edge insertions and deletions where each update is handled in  amortized subpolynomial time, and
    \item a simpler and more efficient data structure to maintain a $(1+\eps)$-approximate single-source shortest paths (SSSP) tree $T$ in a graph undergoing edge deletions in amortized time $m^{o(1)}$ per update. 
\end{itemize}
All our data structures are deterministic. For the last two data structures, we further have that while the trees $T$ are not subgraphs of $G$, they do embed with small edge congestion into $G$. This is in stark contrast to previous approaches and is particularly useful for algorithms that use these data structures internally to route flow along shortest paths. 

To illustrate the power of our new toolbox, we show that our SSSP data structure can be used directly to give a deterministic implementation of the classic MWU algorithm for approximate undirected minimum-cost flow running in time $m^{1+o(1)}$.
Previously, Bernstein-Gutenberg-Saranurak [FOCS'21] had built a randomized data structure achieving $m^{1+o(1)}$ time whp. By using our SSSP data structure in the recent almost-linear time algorithm for computing Gomory-Hu trees by Abboud-Li-Panigrahi-Saranurak [FOCS'23], we simplify their algorithm significantly and slightly improve their runtime.

To obtain our toolbox, we give the first algorithm that, given a graph $G$ undergoing edge insertions and deletions and a dynamic terminal set $A$, maintains a vertex sparsifier $H$ that approximately preserves distances between terminals in $A$, consists of at most $|A|m^{o(1)}$ vertices and edges, and can be updated in worst-case time $m^{o(1)}$. 
Crucially, our vertex sparsifier construction allows us to maintain a low edge-congestion embedding of $H$ into $G$.
This low congestion embedding is needed when using our toolbox in data structures that are then in turn used to implement algorithms routing flows along shortest paths.

\vspace{2em}

\end{abstract}

\pagebreak

\tableofcontents

\pagebreak
\pagenumbering{arabic}

\section{Introduction}

Over the past two decades, vertex sparsifiers have emerged as a central tool in graph algorithms and played a crucial role in the development of efficient flow algorithms. A vertex sparsifier\footnote{Here we use the term in accordance with usage in the `fast graph algorithms' literature. A different notion of vertex sparsification was introduced in \cite{M09}.
}
$H$ of a graph $G=(V,E,l)$ with respect to a terminal set $A \subseteq V$ is a graph that contains roughly $|A|$ vertices and (approximately) preserves a certain graph property between the vertices in $A$ in $H$. 

In \cite{spielman2004nearly}, a framework of vertex sparsifiers that preserve electrical energy between terminals was used to derive the first nearly-linear time\footnote{We follow the convention where an algorithm that is inputted an $m$-edge graph is said to run in nearly-linear time if it runs in time $\tilde{O}(m)$ and to run in almost-linear time if it runs in time $m^{1+o(1)}$.} algorithm to compute electrical flows, a major breakthrough in graph algorithms. In \cite{sherman2013nearly, kelner2014almost, peng2016approximate}, different frameworks based on vertex sparsifiers that preserve cuts%
\footnote{All of these vertex sparsifiers were heavily inspired by the work in \cite{madry2010fast}. 
\cite{kelner2014almost} also ensures a notion of (low-congestion) flow preservation.
} were used to obtain the first almost-linear and then nearly-linear time algorithms to compute approximate, undirected maximum flow.

More recently, algorithms for \emph{dynamically maintaining} vertex sparsifiers have received considerable attention \cite{goranci_et_al:LIPIcs:2017:7846, goranci2018incremental, goranci2018dynamic, forster2019dynamic, durfee2019fully, goranci2020improved, chen2020fast, forster2021dynamic, goranci2021expander, gao2023fully, chen2022maximum, detMaxFlow, forster2023deterministic}. Besides their applications to a myriad of dynamic graph problems, they recently have been used to obtain faster algorithms to solve the static exact maximum flow problem and various of its generalizations. In \cite{gao2023fully}, dynamic vertex sparsifiers that preserve electrical energy between terminals were used to obtain the first exact maximum flow algorithm with runtime ${O}(m^{1.5-\delta})$ for some constant $\delta > 0$ for the important case where the number of edges $m$ is almost-linear in the number of vertices\footnote{For reasonably dense graphs, the algorithm in \cite{van2021minimum} achieves near-linear runtime. They obtain a runtime of $\tilde{O}(m + n^{1.5})$ where $n$ is the number of vertices in the graph $G$.}. Briefly thereafter, \cite{chen2022maximum} presented the first almost-linear time algorithm to compute maximum flows which crucially relies on fast dynamic distance-preserving vertex sparsifiers. But despite their pivotal role in the development of many dynamic algorithms and the recent almost-linear time exact maximum flow algorithm, distance-preserving vertex sparsifiers can only be maintained with polynomial update time \cite{chen2020fast} or by making strong assumptions on the adversary model either by requiring the adversary to be oblivious to the vertex sparsifier \cite{forster2021dynamic, chen2022maximum} or by placing extremely strong assumptions on the update sequence \cite{chen2022maximum, detMaxFlow}. 

In this article, we finally give an algorithm to maintain distance-preserving vertex sparsifiers that work against any adversary and have subpolynomial update time and approximation quality, and thus are essentially optimal. We summarize our main technical result in the following theorem.

\begin{theorem}\label{thm:mainVS}
Given an $m$-edge input graph $G=(V,E,l)$ with lengths in $[1,L]$ and an initially empty terminal set $A$. Then, for some $\gamma_{VS} = e^{O(\log^{20/21} m \log\log m})$, there is a deterministic algorithm that processes edge insertions and deletions to $G$ and vertex insertions and deletions to $A$ and maintains a graph $H$ such that for all vertices $u,v \in A$, we have 
\[
\dist_G(u,v) \leq \dist_H(u,v) \leq \gamma_{VS} \cdot \dist_G(u,v),
\]
and $H$ consists of at most $(|A|+1)\gamma_{VS}\log L$ vertices and edges at any time. The algorithm takes initial time $m\gamma_{VS}\log L$ and then processes each update in worst-case time $\gamma_{VS} \log L$.
\end{theorem}

Given this result, we obtain new algorithms for dynamic shortest paths problems: we give the first $m^{o(1)}$-approximate all-pairs shortest-paths (APSP) algorithm that runs in $m^{o(1)}$ \emph{worst-case} update and query time, the first algorithm to explicitly maintain a tree $T$ over the graph $G$ that has the same diameter as $G$ up to a $m^{o(1)}$ factors, and obtain a much simpler and faster algorithm to maintain a $(1+\epsilon)$-approximate single-source shortest-paths tree $T$ for graphs undergoing only edge deletions. Both trees $T$ here are hierarchical trees, that is, they span the vertex set of $G$ but might include vertices and edges not in $G$.

Our techniques differ starkly from previous techniques used in the area, are arguably simpler, and also yield stronger data structures. Our algorithms for maintaining a low-diameter tree $T$ in a fully-dynamic graph $G$, and for maintaining a $(1+\eps)$-approximate single-source shortest-path tree $T$ in a graph under deletions, both have the property that they can be embedded into $G$ with very low congestion. 
This turns out to be crucial in applications where the tree $T$ is used to route flow as it allows us to maintain the flow explicitly (and always have an $m^{o(1)}$-approximate estimate of the increase of flow on a single edge in $G$ since a given time). This in turn allows us to use the SSSP algorithm to derandomize the recent almost-linear time implementation of the MWU algorithm for undirected approximate min-cost flow in \cite{FOCSbernstein2022deterministic}, and to significantly simplify the recent almost-linear time algorithm to compute Gomory-Hu trees \cite{gomoryHu}. 

\subsection{Roadmap} 

In the next section, \Cref{sec:applicaitons}, we describe our applications of \Cref{thm:mainVS} in more detail and give more formal statements. We then give a brief overview of related work on dynamic distance-preserving vertex sparsifiers and their application to maximum flow in \Cref{subsec:relWork}. Finally, we give, in \Cref{subsec:overview}, an overview of our new techniques to obtain \Cref{thm:mainVS}.

\subsection{Applications}
\label{sec:applicaitons}

\paragraph{Application \#1: Approximate Dynamic APSP.} As an immediate Corollary of \Cref{thm:mainVS} one can obtain a data structure that maintains approximate all-pairs shortest paths (APSP) with worst-case subpolynomial update and query time. To obtain this Corollary, one can simply maintain the terminal set $A$ to be the empty set and upon query for the distance between two vertices $u,v \in V$, add these two vertices to the set $A$, compute static APSP on the vertex sparsifier $H$ (which is of size $m^{o(1)}$) and then output the distance between $u$ and $v$ in $H$ as a distance estimate.

We show that using a more careful, but still simple, approach, one can in fact obtain query times that are significantly better, and a slightly better approximation guarantees.

\begin{restatable}{theorem}{mainTheoremAPSP}\label{thm:mainTheoremAPSP}
Given an $m$-edge input graph $G = (V,E,l)$ with lengths in $[1,L]$, there is a data structure $\textsc{DynamicAPSP}$ that can process a polynomial\footnote{In this paper, the term polynomial always refers to a polynomial in $m$.} number of edge insertions and deletions to $G$ and at any point in time answers queries where inputted $u,v \in V$, it returns a distance estimate $\widehat{\dist}(u,v)$ such that $\dist_G(u,v) \leq \widehat{\dist}(u,v) \leq \gamma_{ApproxAPSP} \cdot \dist_G(u,v)$,  for some $\gamma_{ApproxAPSP} =  e^{O(\log^{6/7} m \log\log m)}$.

For some $\gamma_{timeAPSP} =  e^{O(\log^{20/21} m \log\log m)}$, the data structure can be initialized in time $m \cdot \gamma_{timeAPSP} \cdot \log L$, and thereafter processes each edge update in worst-case time $\gamma_{timeAPSP} \cdot \log L$ and each query in worst-case time $O(\log m \log L)$.
\end{restatable}

We also show that our data structure can be used to certify the diameter of a dynamic set $X \subseteq V$ by outputting two vertices from $X$ that are at a distance roughly equal to the diameter of the set $X$ in the graph $G$. Further, it can be extended to output an approximate shortest path $P$ in time $O(|P|\log m \log L)$. See \Cref{thm:mainTheoremAPSPFormal} and \Cref{rmk:pathReporting} for an extended version of \Cref{thm:mainTheoremAPSP}.

Recently, two different data structures with similar guarantees were obtained in \cite{chuzhoy2023new, ForsterGNS23}, however, both achieved only amortized update time guarantees (and both approximation and update times are significantly larger subpolynomial factors). In fact, even in the easier setting where only edge deletions (no insertions) are allowed, all state-of-the-art algorithms (see \cite{henzinger2018decremental, chechik2018near, forster2021dynamic, chuzhoy2021decremental, FOCSbernstein2022deterministic} obtain trivial worst-case update times. Further, among these, the algorithms that work against an adaptive adversary \cite{chuzhoy2021decremental, FOCSbernstein2022deterministic} all have subpolynomial approximation and update times slightly worse than the bounds we achieve. For the setting where only edge insertions are allowed (no deletions), \cite{forster2023deterministic} obtains polylogarithmic approximation and amortized update time guarantees. We believe that the framework can further be adapted to run with subpolynomial worst-case update time (and approximation). We point out that both Erdös' girth conjecture \cite{thorup2005approximate} and recent conditional hardness results \cite{abboud2022hardness, abboud2023stronger} strongly indicate that an $\omega(1)$ approximation factor is necessary to obtain subpolynomial update and query time (even amortized and against oblivious adversary). We refer the reader to \cite{probst2020new, brand2023deterministic} for a more in-depth discussion of the literature on the dynamic APSP problem.

\paragraph{Application \#2: Dynamic Low-Diameter Tree.} We further show how to maintain a dynamic forest $F$ along with a vertex map $\Pi_{V(G) \mapsto V(F)}$ such that any two vertices in $G$ are mapped to vertices in $F$ that are at distance at most $m^{o(1)} \cdot \diam(G)$ where $\diam(G)$ denotes the diameter of $G$. Further, $F$ embeds into $G$ where every edge in $F$ is mapped to an edge in $G$, thus $F$ is simply obtained from copying $G$ multiple times, and the map is of low congestion meaning that we only need few copies of $G$ to form $F$. For technical reasons, we maintain a forest $F$, however, the vertex map $\Pi_{V(G) \mapsto V(F)}$ maps all vertices in $G$ to the same tree in $F$ as can be extracted from the statement below.

\begin{theorem}\label{thm:mainTheoremLowDiamTreeOverview}
Given an $m$-edge input graph $G = (V,E,l)$ with lengths in $ [1,L]$ and a parameter $D \geq 1$. There is a data structure $\textsc{LowDiamTree}$ that maintains a forest $F$ that can process a polynomially-bounded number of edge insertions and deletions to $G$. 

Under these updates, the algorithm maintains the forest $F$ along with graph embedding $\Pi_{F \mapsto G}$ that embeds each edge in $F$ into a single edge in $G$ and vertex maps $\Pi_{V(G) \mapsto V(F)}, \Pi_{V(F) \mapsto V(G)}$ consistent with the graph embedding, such that, for some $\gamma_{lowDiamTree}= e^{O(\log^{20/21} m \log\log m})$, at any time:
\begin{enumerate}
    \item $\diam_F(\Pi_{V(G) \mapsto V(F)}(V)) \leq \gamma_{lowDiamTree}\cdot \diam(G)$, and
    \item  we have $\econg(\Pi_{F \mapsto G}) \leq \gamma_{lowDiamTree}$, and
    \item $F$ consists of at most $\gamma_{lowDiamTree} \cdot m$ vertices and edges.
\end{enumerate}
The algorithm maintains the forest $F$ and all maps explicitly. It is deterministic, can be initialized in time $m \cdot \gamma_{lowDiamTree}$, and thereafter processes each edge insertion/deletion in amortized time $\gamma_{lowDiamTree}$.
\end{theorem}

To the best of our knowledge, no previous result for maintaining low-diameter trees/forests is known. While we believe that some of the above-mentioned APSP algorithms can produce a tree/forest satisfying the above diameter properties, we believe they cannot maintain an embedding of $F$ into $G$ explicitly. We give a more detailed version of the theorem above in \Cref{thm:mainTheoremLowDiamTree}. %

\paragraph{Application \#3: Approximate Decremental SSSP.} Building on our new algorithm to maintain a low-diameter forest $F$ as described in \Cref{thm:mainTheoremLowDiamTreeOverview}, and our APSP data structure from \Cref{thm:mainTheoremAPSP}, we provide an alternative implementation of the high-level framework from \cite{FOCSbernstein2022deterministic} to obtain an algorithm that maintains a single-source shortest path tree. The technical result is summarized below.

\begin{theorem}\label{thm:mainSSSPGeneral}
Given an $m$-vertex graph $G$ with lengths in $[1, L]$ that undergoes a sequence of edge deletions, a dedicated source vertex $s \in V$ and an accuracy parameter $\eps = \Omega(1/\polylog m)$. Then, there is an algorithm that maintains a collection of forests $F_0, F_1, \ldots, F_{\log_2 L}$ along with vertex maps $\Pi_{V(G) \mapsto V(F_i)}, \Pi_{V(F_i) \mapsto V(G)}$ and an embedding $\Pi_{F_i \mapsto G}$ that maps each edge in $F_i$ to a single edge in $G$ for each forest $F_i$ such that, for some $\gamma_{SSSP} = e^{O(\log^{83/84} m \log\log m)}$, at any time:
\begin{enumerate}
    \item for every $v \in V$, if $\dist_G(s,v) < 2^i \cdot n$, then $\Pi_{F_i \mapsto G}(\pi_{F_i}(\Pi_{V(G) \mapsto V(F_i)}(s), \Pi_{V(G) \mapsto V(F_i)}(v))) \leq (1+\eps) \dist_G(s,v) + \eps \cdot 2^i$, i.e. the path between the two nodes in $F_i$ that vertices $s$ and $v$ are mapped to has length at most $(1+\eps)\dist_G(s,v)$, and
    \item $\econg(\Pi_{F_i \mapsto G}) \leq \gamma_{SSSP}$.
\end{enumerate}
The algorithm maintains each forest $F_i$ and the associated maps $\Pi_{V(G) \mapsto V(F_i)}, \Pi_{V(F_i) \mapsto V(G)}$ and $\Pi_{F_i \mapsto G}$ explicitly and the total number of changes to $F$ and these maps is at most $m \cdot \gamma_{SSSP} \log L$. The algorithm runs in time $m \cdot \gamma_{SSSP} \log L$.
\end{theorem}

Our algorithm should be compared to the recent result from \cite{FOCSbernstein2022deterministic} that obtains similar guarantees with larger subpolynomial factors (the update time is in $m^{1+\Omega(1/\sqrt{\log\log n})}$), except that \cite{FOCSbernstein2022deterministic} cannot provide any non-trivial bound on the congestion with which the shortest path trees/forests $F_i$ embed into $G$. We note that the work in \cite{henzinger2018decremental} obtains a similar result as \cite{FOCSbernstein2022deterministic} with much better subpolynomial factors, however, it again cannot give an guarantees on the congestion of mapping $F_i$ into $G$ and only works against an oblivious adversary while both our result and \cite{FOCSbernstein2022deterministic} are deterministic. We refer the reader for an in-depth discussion of related work to \cite{kyng2022incremental}. 

We also point out that in order to obtain our SSSP data structure, we give a new algorithm to maintain a sparse neighborhood cover in a graph undergoing edge deletions. We refer the reader to \Cref{thm:covering} for the formal statement and point out that our algorithm enforces much stronger properties than the recent algorithm by Chuzhoy \cite{chuzhoy2021decremental} that obtains a similar result.

\paragraph{Application \#4: A Simple Deterministic Almost-Linear Time Algorithm for Undirected, Approximate Minimum-Cost Flow.} A classic approach for computing the approximate $st$-maximum flow in an undirected graph $G = (V,E,u)$, where $u$ is a function mapping edges in $E$ to capacities and $s$ is the source and $t$ the sink vertex, is to solve (approximately) the following linear program (LP) which is the dual to a classic LP formulation to compute maximum flows 
\begin{align*}
    \min &\sum_{e \in E} u(e)l(e) &\\
    & l(P) \geq 1 & \forall P \in \mathcal{P}_{s,t} \\
    & l(e) \geq 0 & \forall e \in E
\end{align*}
Here $\mathcal{P}_{s,t}$ denotes the set of $st$-paths in $G$.
Applying the multiplicative weight method (MWU) (see \cite{arora2012multiplicative, fleischer2000approximating}) to the LP, one obtains the following simple algorithm (\Cref{alg:apxmaxflow}) to compute a $(1+\tilde{O}(\epsilon))$-approximate solution for both the primal and dual LP, that is the algorithm explicitly constructs the approximate maximum flow.

\begin{algorithm}
\label{alg:apxmaxflow}
\ForEach{$e \in E$}{$f(e) \gets 0; l(e) \gets m^{1/\epsilon}$.}
\While{there is a path $P \in \mathcal{P}_{s,t}$ with $l(P) \leq 1$}{
    Let $P$ be the shortest such $st$-path.\label{lne:pickShortestPath}\\
    $\Delta \gets \min_{e \in P} u(e)$.\\
    \ForEach{$e \in P$}{
        $f(e) \gets f(e) + \Delta$; $l(e) \gets l(e)\left(1 + \eps \cdot \frac{\Delta}{u(e)}\right)$.
    }
}
\Return $\log_{1+\eps}\left(\frac{1+\eps}{m^{1/\epsilon}}\right) \cdot f$.
\caption{$\textsc{ApproxMaxFlow}(G =(V,E,u), s,t, \eps)$}
\label{alg:MWU}
\end{algorithm}

For constant $\eps > 0$, it can be shown that the algorithm terminates after $\tilde{O}(m)$ iterations, and the approximation guarantees of the algorithm still hold if one relaxes \Cref{lne:pickShortestPath} (and the while-loop condition) and requires a $(1+\epsilon)$ shortest $st$-path in $G$ with respect to $l$. This also allows us to work with an approximation $\hat{l}$ of the lengths $l$ that is monotonically increasing and where every time a value $\hat{l}(e)$ is increased, it is increased by at least a $(1+\epsilon/2)$ multiplicative factor. Thus, the total number of changes to $\hat{l}(e)$ is at most $\Otil(1)$ times. 

In the recent work of \cite{FOCSbernstein2022deterministic}, the authors gave a data structure that can be queried for the $(1+\eps)$-approximate distance between $s$ and $t$ in the current graph $G$ w.r.t. to $\hat{l}$. The data structure is also able to output an approximate $st$-shortest path $P$ in time linear in the number of edges on $P$. But note that \Cref{alg:MWU} iterates over the edges in the path and therefore the algorithm would require $\Omega(n)$ time per iteration if the chosen paths are long.

To overcome this problem, \cite{FOCSbernstein2022deterministic} designed their data structure to allow the user to sample edges from the approximate shortest $st$-path $P$ according to their capacity. This allowed them to maintain \emph{estimators} of $f$ and $l$ that are then only updated very few times in each coordinate. They then gave a rather intricate argument by standard analysis based on Martingale theory to prove that the algorithm still is well-behaved, i.e. terminates quickly and guarantees good approximations.

In this article, we offer a different approach that is much simpler and more direct. We can simply use the data structure from \Cref{thm:mainSSSPGeneral}. We maintain the single-source shortest path forest $F$ in a dynamic tree structure. We can then find a $(1+\eps)$-approximate distance estimate for the distance from $s$ to $t$ by querying the dynamic tree structure. Further, the identified forest path $P_F$ can then be queried for the min-capacity edge and the flow can directly be added to $F$ via the dynamic tree structure, all in time $O(\log m)$. As shown in \cite{chen2022maximum}, dynamic tree structures on graphs that embed with $n^{o(1)}$ edge congestion into $G$ can be used to maintain $\hat{l}$ to be a $(1+\eps)$-approximation of the real lengths $l(e)$ which is sufficient for the MWU algorithm.

To obtain approximate minimum-cost flows instead of approximate maximum flows, a slightly different LP formulation is used resulting in a very similar outer loop via the MWU method. We refer the reader to \cite{fleischer2000approximating} for the full details.

\paragraph{Application \#5: Simplifying the Recent Almost-Linear Time Algorithm to Compute Gomory-Hu Trees.} In \cite{abboud2022breaking}, the authors designed a framework reducing the computation of Gomory-Hu trees in almost-linear time to a problem named the decremental minimum $U$-Steiner subgraph problem. In this problem, the graph is undergoing edge deletions over time, and the goal is to maintain a subgraph $H$ of $G$ that connects the vertices in $U$ with minimum length paths under length function $l$. The problem can be solved by an MWU algorithm that is very similar to the maximum flow MWU algorithm discussed in the last section. Finally, they show that the problem of finding $H$ in each iteration boils down to solving the decremental $(1+\eps)$-approximate SSSP problem. In \cite{abboud2022breaking}, this reduction only works for unit-lengthed graphs, and more recently \cite{gomoryHu} this was addressed by showing that the technique of edge sampling from approximate $st$-paths introduced in \cite{FOCSbernstein2022deterministic} can also be used to extend their algorithm to weighted graphs. Again, our new SSSP data structure can be used in lieu of the data structure from \cite{FOCSbernstein2022deterministic}. While \cite{gomoryHu} also has to design various additional components that are added to the data structure in \cite{FOCSbernstein2022deterministic} to correctly maintain the subgraphs $H$, our data structure delivers these properties out-of-the-box as one can simply use the approximate shortest paths forest $F$ directly to maintain the graphs $H$ rather explicitly (in the form of fractional flows).

\subsection{Related Work}
\label{subsec:relWork}

In this section, we review related work on distance-preserving vertex sparsifiers. In \cite{forster2019dynamic, forster2021dynamic}, the authors present an algorithm for unweighted dynamic graphs to maintain a probabilistic low-stretch spanning tree (LSST) $T$ with expected stretch $m^{o(1)}$, that is, for any vertices $u,v \in V$, $\mathbb{E}[\dist_T(u,v)] = m^{o(1)}$. The algorithms work against an oblivious adversary and \cite{forster2021dynamic} obtains $m^{o(1)}$ randomized amortized update time. This algorithm can be used to obtain maintain a distance-preserving vertex sparsifier as follows\footnote{We are not aware that the following reduction is known in the literature. Rather, probabilistic LSSTs are usually used directly to maintain a dynamic APSP algorithm by querying the LSSTs.}: instead of maintaining a single LSST $T$, let us maintain $\lambda = O(\log m)$ dynamic LSSTs $T_1, T_2, \ldots, T_\lambda$ such that at any time, for any two vertices $u,v \in V$, one of the trees $T_i$ has $u$ and $v$ at distance at most $m^{o(1)}$ (proving this claim is straightforward via Markov's Inequality and a simple Union Bound). Finally, maintain sparsifier $H$ for a vertex set $A \subseteq V$ as the union of graphs obtained from trees $T_i$ after finding the set of least common ancestors $A_i$ of $A$ in $T_i$ and then contracting all maximal paths in $T_i$ that have no internal vertex in $A_i$. Since $A_i$ can be shown to be of size $2|A|$, this yields that $H$ consists of only $\Otil(|A|)$ vertices and edges. It is not hard to show that using link-cut trees (see \cite{sleator1981data}), $H$ can be maintained efficiently given the trees $T_1, T_2, \ldots, T_\lambda$.

In \cite{chen2020fast}, the first non-trivial algorithm to dynamically maintain distance-preserving vertex sparsifiers was given for \emph{weighted} graphs. However, the algorithm still required polynomial update time and only worked against an oblivious adversary. Only recently, an algorithm with subpolynomial update time was given in \cite{chen2022maximum} (inspired by \cite{chen2020fast}), however, their algorithm still does not work against an adaptive adversary\footnote{Again, the algorithm in \cite{chen2022maximum} maintains probabilistic LSSTs but by the above reduction this again yields dynamic vertex sparsifiers.}.

Despite not being able to maintain the dynamic vertex sparsifiers against an adaptive adversary, \cite{chen2022maximum} succeeded in using them to solve a dynamic subproblem called min-ratio cycle problem that appears in their almost-linear time algorithm for minimum-cost flow. This was achieved, surprisingly, by arguing that the specific update sequence produced by the outer-loop of the flow algorithm is rather well-behaved and that the vertex sparsifier can thus still be maintained. More recently, \cite{detMaxFlow} obtained a deterministic almost-linear time algorithm for maximum flow, but remarkably, they used a deterministic variant of the \cite{chen2022maximum} vertex sparsification procedure which \emph{still does not work against a general adversary} but again only proves correctness of the dynamic algorithm for the specific update sequence produced by the outer-loop of the flow algorithm.
But in general, proving that update sequences are well-behaved is a difficult endeavor. \cite{chen2022maximum} and \cite{detMaxFlow} both use an intricate set-up for restarting data structures after failure, which adds significant complexity and makes it hard to modularize components.
Roughly speaking, both approaches rely on restarting vertex sparsifiers after they fail, while showing this only occurs a subpolynomial number of times.\footnote{The precise statement is a recursive version of this simpler statement, and incorporates a restarting strategy that is analyzed against a restricted adversary.}
Further, if one would want to use dynamic vertex sparsifiers in any other such algorithm, one would first have to prove well-behavedness of the update sequence again and tailor the vertex sparsifier maintenance algorithm to the update sequence. Even worse, some interesting algorithms may produce an update sequence that is simply not well-behaved, such as the incremental threshold min-cost flow algorithm of \cite{BLS23} running in $m^{1+o(1)}\sqrt{n}$ time.
This algorithm crucially needs a min-ratio cycle data structure that works against a very general adversary, as one can no longer use restarting strategies to cope with data structure failure.
This is because the algorithm cannot distinguish whether the lack in progress of a step is due to a failure of the vertex sparsifier maintenance algorithm or the fact that edges necessary to route flow have not arrived yet.

\subsection{Overview}
\label{subsec:overview}

\paragraph{High-Level Strategy.} In this overview, we sketch our algorithm to dynamically maintain a vertex sparsifier $H$ preserving distances between vertices in terminal set $A$ in the dynamic input graph $G = (V,E,l)$, as described in \Cref{thm:mainVS}. For the rest of the overview, we assume that $G$ has constant maximum degree at all times which can be assumed without loss of generality by standard reductions.

The key building block for our algorithm is given in the informal theorem below.

\begin{informaltheorem}\label{infThm:stepRed}
Consider a size reduction parameter $k$, an $n$-vertex graph $G=(V,E,l)$ undergoing at most $n/k$ edge deletions and insertions, and insertions of isolated vertices, such that at all times $G$ has maximum degree $\Delta$ and a monotonically increasing set $A$. Then for some $\gamma = n^{o(1)}$, there is an algorithm that maintains a vertex sparsifier $H$ with respect to $A$ consisting of at most $(n/k + |A|) \gamma$ edges and vertices, with maximum degree $\gamma \Delta$, with $\gamma$ recourse, with stretch $\gamma$ on the distances between terminals, and initialization time $n^{1+o(1)} \poly(\Delta k)$ and update time $n^{o(1)} \poly(\Delta k)$.  
\end{informaltheorem}

Note that for large terminal set $A$, \Cref{infThm:stepRed} already yields the desired result. However, typically, the terminal set $A$ is of very small size. Consider for example dynamic APSP algorithms, where the terminal set $A$ is usually merely of size $2$. To illustrate the overall technique, we assume that $A$ is of size at most $n^{o(1)}$ and show how to recursively use  \Cref{infThm:stepRed} to obtain a vertex sparsifier $H$ consisting of $n^{o(1)}$ vertices and edges that preserves distances between vertices in $A$.

To obtain such a vertex sparsifier, we choose a reduction parameter $k$ that is subpolynomial in $n$ but superpolynomial in $\gamma$. In doing so, we ensure that we chose $k = n^{o(1)}$ and for $K = \log_k(n)$, we have $\gamma^K = n^{o(1)}$. 

Given this choice of parameters, let us define $G_0 = G$ and for every $0 \leq i < K$, we define $G_{i+1}$ to be the graph obtained from maintaining the vertex sparsifier from \Cref{infThm:stepRed} with size reduction parameter $k$ and restart $G_{i+1}$ every $\lceil n/(k^{i+1}\gamma^K) \rceil$ updates to $G$. Finally, we output $H = G_{K}$. 

Let us briefly analyze the algorithm. The final stretch between terminals in $A$ in the vertex sparsifier is at most $\gamma^K = n^{o(1)}$, and the number of vertices and edges in $H$ is at most $(n/k^K + |A|) \gamma^K = (1+|A|)\gamma^K$. Finally, for the update time, we have that each update to $G$ causes at most $\gamma^i$ updates to $G_i$ by the upper bound on the recourse. Thus, each update to $G$ causes at most $\gamma^K$ updates to any such graph $G_i$. By rebuilding each graph $G_{i+1}$ after every interval of $\lceil n/(k^i\gamma^K) \rceil$ updates to $G$, we ensure that the update sequence to $G_{i+1}$ does at no point exceed a length of $n/k^{i+1}$ as desired\footnote{Technically, $G_{i+1}$ could consist of much less than $n/k^i$ vertices and therefore not allow for a sequence of $n/k^{i+1}$ updates, however, in such case one can add isolated vertices to the initial graph $G_{i+1}$ until it is of size $n/k^i$.}. To establish an amortized bound on the compute time, consider an update sequence consisting of $n$ updates to $G$. The total time required by the data structure to maintain graph $G_{i+1}$ given $G_i$ is at most 
\[
    \frac{n}{\lceil n/(k^i\gamma^K) \rceil} \cdot n\gamma^K/k^i \cdot \poly(\Delta\gamma^K k) + n\gamma^K \cdot n^{o(1)} \poly(\Delta\gamma^K k) = n^{1+o(1)}
\]
where the first term stems from the number of rebuilds, of which there are $\frac{n}{\lceil n/(k^i\gamma^K) \rceil}$, and the re-initialization on a graph of size at most $n\gamma^K/k^i$ for each such rebuild. The second term stems from the total recourse at level $i$ and the update time given by \Cref{infThm:stepRed}. Thus, all properties from \Cref{thm:mainVS} have been established.

For the rest of the overview, we outline an algorithm that implements \Cref{infThm:stepRed} focusing on achieving the claimed recourse bound. Our approach builds on the recent techniques by Andoni, Stein and Zhong \cite{andoni2020parallel} to statically compute vertex sparsifiers, which also step-wise reduces the size of the vertex sparsifier by a factor $k$. We show that dynamizing their construction for one size reduction step is rather straightforward but can only be realized with $\Omega(k)$ recourse in the vertex sparsifier.  But, this recourse is too high for a recursive application and thus we cannot usefully apply multiple size-reduction steps.
To improve the recourse
to the desired $\gamma$ factor that is subpolynomial and independent of reduction size $k$, we need to develop new techniques. We show that the path collection used by Andoni, Stein and Zhong \cite{andoni2020parallel} to build the edge set of the vertex sparsifier can be embedded in few deterministic Low-Stretch Spanning Trees (LSSTs). Using dynamic core graph constructions on these LSSTs and dynamic edge sparsification on the core graphs, both as developed in \cite{chen2020fast, chen2022maximum, detMaxFlow}, then preserves the distances between terminals with small stretch $\gamma$. Unfortunately, the path collection from \cite{andoni2020parallel} is fully-dynamic meaning at a later time a new path might be added. But we need to know all paths to ever be in the path collection when we compute the LSSTs at initialization time.

To this end, we make a simple but crucial observation: when we initially build our core graphs, we can construct both the current Andoni-Stein-Zhong paths for the current state of the graph, and \emph{the future paths for a long sequence of future states of the graph}. This allows our core graphs to ``prepare'' for future states of the graph.

\paragraph{The Andoni-Stein-Zhong Vertex Sparsifier.} We start by giving a brief introduction\footnote{Note that we give a presentation tailored to build intuition for our final algorithm, thus our presentation deviates from the presentation in \cite{andoni2020parallel}.} to the \emph{static} vertex sparsifier as presented by Andoni, Stein and Zhong \cite{andoni2020parallel}. Given an $n$-vertex graph $G$ with maximum degree $\Delta$, unique shortest paths, a size reduction parameter $k$ and a terminal set $A$. Sample a set $A'$ by adding each vertex in $V$ to $A'$ with probability $1/k$. Obtain $A''$ as the union of $A$ and $A'$. For each vertex $v$, denote by $p(v)$ the closest vertex in $A''$ to $v$ in graph $G$. Let $B_G(v, A'')$ denote the open ball around vertex $v$ in $G$ of radius equal to the distance from $v$ to the closest vertex in $A''$ (that is $p(v)$). Given this set-up, we construct the path collection $\mathcal{P}$ as follows:
\begin{itemize}
    \item for any two vertices $u \in V$ and $v \in B_G(u, A'')$ and edge $(v, x) \in E(G)$, add $\pi_{G}(u,v) \oplus (v,x)$ to $\mathcal{P}$ where $\pi_G(u,v)$ denotes a shortest $uv$-path in $G$.
\end{itemize}

We define the set of projected path $\widehat{\mathcal{P}}$ as follows: for every $uv$-path $P$ in $\mathcal{P}$, add to $\widehat{\mathcal{P}}$ the path $\pi_{G}(p(u), u) \oplus P \oplus \pi_G(v, p(v))$. Finally, we take the vertex sparsifier to be the graph $H$ over vertex set $A''$ with an edge $e = (a,b)$ for every $ab$-path $P \in \widehat{\mathcal{P}}$ with length $l_H(e) = l_G(P)$.

To gain some intuition for this construction, let us analyze the stretch of the sparsifier. We prove the following claim.

\begin{claim}
For any $u,v \in A$, we have $\dist_G(u,v) \leq \dist_H(u,v) \leq 4 \cdot \dist_G(u,v)$.
\end{claim}
\begin{proof}
For any $u,v\in V$, let $x_0 = u$ and find vertices $x_1, x_2, \ldots, x_k = v$ by letting $x_{i+1}$ be the last vertex on the shortest $uv$-path $\pi_G(u,v)$ that is incident to a vertex in the ball $B(x_i, A'')$ (and thus there exists an edge $(x_i, x_{i+1})$ in $H$) and stop once $x_i = v$. Then, since $p(x_0) = u, p(x_k) = v$ because $u,v \in A \subseteq A''$, we can upper bound the distance $\dist_H(u,v) \leq \sum_{i} l_H(p(x_i), p(x_{i+1}))$. Finally, we can use that since the pivot vertex $p(x_i)$ is at a distance equal to the radius of $B_G(x_i, A'')$, we have that $\dist_G(p(x_i), x_i) \leq \dist_G(x_i, x_{i+1})$ and thus by the triangle inequality it is easy to show that $l_H(p(x_i), p(x_{i+1})) \leq  \dist_G(p(x_i), x_i) + \dist_G(x_i, x_{i+1}) + \dist_G(x_{i+1}, p(x_{i+1})) \leq 4 \dist_G(x_i, x_{i+1})$. The proof then follows since the vertices $x_i$ segment the shortest $uv$-path.
\end{proof}

By a standard hitting set argument, we have that for every $v \in V$, $|B_G(v, A'')| = \tilde{O}(k)$ because w.h.p. one of $\tilde{O}(k)$ closest vertices to $v$ is in $A'$. Therefore, the size of the edge set of $H$ is bounded by $\tilde{O}(nk \Delta)$. The size of $A''$ and thus the vertex set of $H$ is $|A| + n/k$ in expectation.

\paragraph{A First Attempt at Dynamizing the ASZ-Vertex Sparsifier.} For the rest of the overview, we focus on $G$ being a graph undergoing only edge deletions (extending to the case where $G$ undergoes edge insertions and other operations is rather straightforward). Let us further assume for convenience that $G$ has unique shortest paths at all times.

Consider the following attempt: Initially compute a deterministic set $A'$ of size $\tilde{O}(n/k)$ that not only ensures that all balls $B_G(v, A'')$ are small, but also that the inverses of the ball, the so-called clusters $C_G(u, A'') = \{ v \;|\; u \in B_G(v, A'')\}$, satisfy that $|C_G(u, A'')| = \tilde{O}(k)$. Then, whenever an edge $(u,v)$ is deleted from $G$, simply add the endpoints $u, v$ to the set $A''$ and update the ASZ-Vertex Sparsifier to reflect this change. 

We first note that every ball $B_G(v, A'')$ has decreasing radius. That is because no vertex in $B_G(v, A'')$ is incident to a deleted edge as otherwise it would have been added to $A''$, but $B_G(v, A'')$ is the open ball with radius equal to the closest vertex in $A''$, so it contains no vertex in $A''$. Therefore any current shortest path $\pi_G(u,v)$ between $v$ and a vertex $u \in B_G(v, A'')$ has always been the shortest path between $u$ and $v$. This implies that the path collection $\mathcal{P}$ is monotonically decreasing over time, i.e. paths are removed one by one from $\mathcal{P}$ and at no time is a new path added. 

Further, since $B_G(v, A'')$ has a monotonically decreasing radius while $G$ has monotonically increasing distances, we have that $B_G(v, A'')$ is a monotonically decreasing set and so are the clusters $C_G(u, A'')$. The latter fact implies that a single edge deletion can affect at most $\tilde{O}(k)$ balls $B_G(v, A'')$ and thus only change $\tilde{O}(k)$ pivots $p(v)$. It thus follows that the number of changes to the set of projected paths $\mathcal{P}$ is at most $\tilde{O}(k^2 \Delta)$. And this is equal to the recourse of the vertex sparsifier $H$.

Unfortunately, any (super-)linear dependency in $k$ for the recourse is not tolerable when attempting to recursively use vertex sparsifiers. Thus, this simple attempt does not appear to work.

\paragraph{Low-Stretch Spanning Trees (LSSTs) to the Rescue.} Before we explain how the path collection $\widehat{\mathcal{P}}$ is still useful, let us briefly discuss the tools for maintaining LSSTs in graphs undergoing edge deletions from \cite{chen2022maximum}, which in turn built on \cite{chen2020fast}.

Given an initial graph $G=(V,E, l)$, a low-stretch spanning tree of $G$ is a tree $T$ such that on average each edge is only stretched slightly, more formally, $\sum_{e = (u,v) \in E}  l(T[u,v]) / l(e) \leq \tilde{O}(m)$ where $m = O(n\Delta)$ is the number of edges in $G$. It turns out to be useful to extend this definition to rooted forests. Given a rooted forest $F$ and denoting by $\root^F(v)$ the root of the component containing vertex $v$, we define the stretch of $e = (u,v)$ induced by $F$ with 
\begin{align*}
\str^{F,G}(e) \defeq
\begin{cases}
  1 + l(F[u,v])/ l(e) &~\text{ if } \root^F(u) = \root^F(v) \\
  1 + \left(l(F[u, \root^F(u)]) + l(F[v, \root^F(v)]) \right)/l(e) &~\text{ if } \root^F(u) \neq \root^F(v).
\end{cases}
\end{align*}
This forest stretch is essentially defined so that for edges internal to a tree of the forest, it equals tree stretch, but for edges crossing between components, it measures the detour caused by always forcing paths to move to the root of a component before exiting it.
\cite{chen2022maximum} gave an algorithm that initially computes a rooted forest $F$ and stretch upper bounds $\tilde{\str}(e)$ for every edge that satisfy that $\sum_{e \in E} \tilde{\str}(e) l(e) = \tilde{O}(1) \cdot \sum_{e \in E} l(e)$. Then, as $G$ undergoes edge deletions, the algorithm removes for every update to $G$ at most $\tilde{O}(1)$ edges from $F$, determines new roots in components without a root, and thereby ensures that thereafter for every edge $e \in E$, $\str^{F, G}(e) \leq \tilde{\str}(e)$ and that $F \subseteq G$.

Using this construction,
\cite{chen2022maximum} then builds \emph{core graphs}. That is a graph $C(G, F, \tilde{\str})$ that is obtained from $G$ by contracting all components of the current forest $F$ where the length of every edge $\hat{e}$ in the core graph that corresponds to edge $e = (u,v)$ in $G$ is set to $l_{C(G, F, \tilde{\str})}(\hat{e}) = \tilde{\str}(e) \cdot l(e)$. They then show the following statement.

\begin{informaltheorem}\label{inf:enforceDistByPathWitness}
Given a graph $G=(V,E,l)$, and a rooted forest $F$ that is updated as described above, for any two vertices $u, v$ that are roots in $F$ at the current time,
we have
for every current $uv$-path $P$ in $G$ 
that $\dist_G(u,v) \leq \dist_{C(G, F, \tilde{\str})}(u,v)
\leq \sum_{e \in P} \tilde{\str}(e) l(e)$.
\end{informaltheorem}
In \cite{chen2022maximum}, the above algorithm is run with probabilistic LSSTs to obtain a vertex sparsifier. That is, the stretch estimates are not only correct on average but in expectation, i.e. for every edge $e \in E$, we have $\mathbb{E}[\tilde{\str}(e)] = \tilde{O}(1)$. It is not hard to enforce that the vertices in our terminal set $A$ become roots in the forest $F$. Against an oblivious adversary, one can then argue that for any $u,v \in A$, we have that for the current shortest $uv$-path $\pi_G(u,v)$, we have $\mathbb{E}[\sum_{e \in \pi_G(u,v)} \tilde{\str}(e) \cdot l(e)] = \tilde{O}(1) \cdot \sum_{e \in \pi_G(u,v)} l(e) = \tilde{O}(1) \cdot \dist_G(u,v)$. Using Markov's inequality and \Cref{inf:enforceDistByPathWitness}, we thus derive that with probability at least $\frac{1}{2}$, the core graph $C(G, F, \tilde{\str})$ preserves the distance between $u$ and $v$ for any $u,v\in A$.
To boost probabilities, instead of maintaining a single core graph, we can sample  $\lambda = O(\log n)$ core graphs $C(G, F_i, \tilde{\str}_i)$ for every $1 \leq i \leq \lambda$ where $F_i$ is independently taken from all other forests and $\tilde{\str}_i$ is the set of stretch overestimates outputted when computing $F_i$. This construction ensures  whp. that for any pair of terminal vertices, in some core graph, the distance is preserved.
Taking a union of these graphs leads to a distance preserving graph on the terminal set with a few extra vertices added.\footnote{Note that \cite{chen2022maximum} was directly building vertex sparsifiers for an oblivious version of the so-called min-ratio cycle problem.
In this context, taking a union of core graphs does not work, for subtle reasons related to preserving the so-called \emph{gradient} (a vector over $\R^{E}$) under contraction in a dynamic graph. Hence, they had to treat each core graph separately, and reason about the expected stretch of a ``hidden flow'', instead of distances between vertex pairs. }
This yields a vertex sparsifer for terminal distances against an oblivious adversary. The sparsifer has too many edges, but low recourse measured under edge insertions/deletions and vertex splits. This in turn makes it amenable to edge sparsification using the dynamic spanner of \cite{chen2022maximum}.
Thus, a vertex sparsifier preserving distances against oblivious adversaries is implicit in \cite{chen2022maximum}.

Crucially, the success of this algorithm hinges on being able to take a union bound over the shortest paths that exist at each time in graph $G$. If these paths are not determined before the randomness is used, the adversary can pick one of exponentially many paths that is not preserved by any forest, and design an update sequence that leaves this path to become a shortest path.

In our algorithm, we use a different approach: instead of sampling forests $F_1, F_2, \ldots, F_{\lambda}$, we design them to preserve all paths in the projected path collection $\widehat{\mathcal{P}}$. Clearly, if this is the case, and again if we make the vertices in $A$ roots in the forests $F_1, F_2, \ldots, F_k$, then the distances between vertices in $A$ are preserved in the vertex sparsifier $H$ taken to be the union of core graphs $C(G, F_i, \tilde{\str}_i)$. However, the set $\widehat{\mathcal{P}}$ is a dynamic set, and we construct all forests $F_1, F_2, \ldots, F_{\lambda}$ at the start without knowing the update sequence. Thus, we need the set of all future paths in $\widehat{\mathcal{P}}$ to carry out this approach.

\paragraph{Finding All Future Paths.} Turning back to our dynamization of the ASZ vertex sparsifier, we already observed that the set $\mathcal{P}$ is a monotonically decreasing set. Recall that $\widehat{\mathcal{P}}$ is the set of projections, that is, for every $uv$-path $P \in \mathcal{P}$, it contains a path $\pi_G(p(u), u) \oplus P \oplus \pi_G(v, p(v))$ in $\widehat{\mathcal{P}}$. 

Now observe that for every $u \in V$, the ball $B_G(v, A'')$ is a monotonically decreasing set of initial size $\tilde{O}(k)$. But every pivot $p(u)$ of $u$ at any time in the algorithm is either equal to the initial pivot of $u$, or it is at a smaller distance than the initial pivot of $u$ and thus contained in the initial ball $B_G(v, A'')$. Thus, there are only $\tilde{O}(k)$ potential pivot vertices to consider for every vertex $u$. Moreover, the shortest path to any future pivot vertex $x$, is always already the shortest $ux$-path in the initial graph $G$.

Thus, we can initially construct the set $\widetilde{\mathcal{P}}$ to contain for every $uv$-path $P \in \mathcal{P}$, any $\tilde{u} \in B_G(u, A'') \cup \{p(u)\}$ and  $\tilde{v} \in B_G(v, A'') \cup \{p(v)\}$, the path $\pi_G(\tilde{u}, u) \oplus P \oplus \pi_G(v, \tilde{v})$. It is not hard to see, that the initial and any future set $\widehat{\mathcal{P}}$ is a subset of the set $\widetilde{\mathcal{P}}$. 

This allows us to construct forests $F_1, F_2, \ldots, F_{\lambda}$ that preserve all future paths between vertices in $A$ that would have been included in the dynamic set $\widehat{\mathcal{P}}$ by enforcing that they preserve the paths in $\widetilde{\mathcal{P}}$. Finally, we note that the set $\widetilde{\mathcal{P}}$ is only of size $\tilde{O}(k^2)$ larger than the set of paths $\mathcal{P}$ and thus of size at most $\tilde{O}(n k^3 \Delta)$. 

\paragraph{The Final Vertex Sparsifier Algorithm.} We summarize our construction. Consider graph $G$ undergoing edge deletions with unique shortest paths and with maximum degree $\Delta$ and terminals $A$. We initially compute a set $A'$ such that for every vertex $v \in V$, the ball $B_G(v, A')$ but also its cluster $C_G(v, A')$ are of size at most $\tilde{O}(k)$. We initialize the set $A''$ to $A \cup A'$ and observe that it also enforces the above property on balls and clusters. Computing set $A'$ and balls $B_G(v, A'')$, clusters $C_G(v, A'')$ and all shortest paths between vertices therein, and all vertex $v$ to initial pivot $p(v)$ shortest paths, can be done in time $\tilde{O}(n\Delta k)$ rather straightforwardly.

We then compute the path collection $\widetilde{\mathcal{P}}$ as follows: for every $u \in V, x \in B(G, A'')$ and edge $(x, v) \in E$ and any $\tilde{u} \in B_G(u, A'') \cup \{p(u)\}, \tilde{v} \in B_G(v, A'') \cup \{p(v)\}$, we to add to $\widetilde{\mathcal{P}}$ the path 
\[
    \pi_G(\tilde{u}, u) \oplus \pi_G(u,x) \oplus (x,v) \oplus \pi_G(v, \tilde{v}). 
\]
We then compute forests $F_1, F_2, \ldots, F_{\lambda}$ for $\lambda = O(\log n)$ such that every path $P$ in $\widetilde{\mathcal{P}}$ is preserved by at least one forest, i.e. the stretch $\tilde{\str}_i(P) = \frac{1}{l(P)} \cdot \sum_{e \in P} \tilde{\str}_i(e) \cdot l(e)$ is in $\tilde{O}(1)$ where $\tilde{\str}_i$ is the stretch overestimate outputted with forest $F_i$. Computing $\widetilde{\mathcal{P}}$ can be done in time $\tilde{O}(n\Delta k^4)$ because all segments from which a path $P \in \widetilde{\mathcal{P}}$ are already computed when computing balls and clusters, and each of the constantly many path segments from which it is assembled is contained in a ball except for the last edge and thus trivially of size $\tilde{O}(k)$. Computing the forests $F_1, F_2, \ldots, F_\lambda$ can be done in time near-linear in the number of edges on all paths in $\widetilde{\mathcal{P}}$ which again yields runtime $\tilde{O}(n\Delta k^4)$.

Next, we compute the core graphs $C(G, F_i, \tilde{str}_i)$ as previously described by contracting components in $F_i$ in the graph $G$ and adjusting the length of an edge $\hat{e}$ that corresponds to edge $e$ in $G$ to have length $l_{C(G, F_i, \tilde{\str}_i)}(\hat{e}) = \tilde{\str}_i \cdot l(e)$. We enforce that the vertices in $A''$ become roots in every forest $F_i$ and then form vertex sparsifier $\tilde{G}$ as the union of the core graphs, which yields a graph that preserves distances between vertices in $A$ consisting of $\tilde{O}(|A''|) = \tilde{O}(n/k + |A|)$ vertices and $\tilde{O}(m)$ edges. Building core graphs is rather straightforward and thus $\tilde{G}$ can be obtained in time $\tilde{O}(m)$ given the forests $F_1, F_2, \ldots, F_{\lambda}$. 

Finally, we obtain the vertex sparsifier $H$ from applying an edge sparsification procedure to $\tilde{G}$ which ensures that the distances in $H$ approximate the distances in $\tilde{G}$ while the number of edges in $H$ is only $n^{o(1)} \cdot |V(\tilde{G})| = (|A| + n/k)n^{o(1)}$.

As $G$ undergoes updates, it adds the endpoints of deleted edges to $A''$ and it then only remains to forward the edge deletions to the algorithm maintaining the forests $F_1, F_2, \ldots, F_{\lambda}$. Each edge deletion results in at most $\tilde{O}(1)$ edge removals from each forest $F_i$ and thus it is not hard to see that the core graphs and by extension the sparsifier $\tilde{G}$ only change by $\tilde{O}(1)$ edge deletions and vertex splits. Combining this with the powerful dynamic sparsification techniques from \cite{chen2022maximum, detMaxFlow}, we can then update $H$ to only change $n^{o(1)}$ edges and add at most one new isolated vertex. All of these updates can be processed in worst-case time $n^{o(1)} \poly(\Delta k)$ and as mentioned above the recourse of $H$ is only $n^{o(1)}$.

We point out that in our algorithm, we have to carefully control the maximum degree of $H$. While the average degree of $H$ is small since it is sparse, this might not be the case for the maximum degree. However, we show that it suffices to remove few edges during preprocessing and then a single additional edge from forests $F_1, F_2, \ldots, F_{\lambda}$ per update to $G$ to ensure that the maximum degree of $H$ remains bounded by $\Delta \cdot n^{o(1)}$.

This construction then yields an implementation of \Cref{infThm:stepRed} where $\gamma$ is chosen to subsume all $n^{o(1)}$ factors above.

\paragraph{Mapping a Forest $F$ in $H$ to a Forest $F'$ in $G$.} Finally, we point out that our vertex sparsifiers allow us to map a hiearchical forest $F$ in $H$ to a hiearchical forest $F'$ in $G$ with low-congestion. 

To gain some intuition, first consider $H$ to be derived by running the one-step size reduction from \Cref{infThm:stepRed} on $G$. Now to map the forest $F$ on $H$ into $G$, we can do the following: forest every vertex $v \in V(H)$, let $T_v$ denote the direct sum of all trees in $F_1, F_2, \ldots, F_{\lambda}$ rooted at $v$ glued in the vertex $v$. More precisely, there are trees $T_1, T_2, \ldots, T_k$ for $k \leq \lambda$ that are rooted in $v$ in the forests $F_1, F_2, \ldots, F_{\lambda}$. Note that the trees $T_1, T_2, \ldots, T_k$ might not be vertex-disjoint and in fact, $v$ is the root of every such tree by definition. The tree $T_v$ is formed by having a copy of each tree $T_1, T_2, \ldots, T_k$ in the graph and then contracting all copies of vertex $v$ into a single vertex identified with $v$.

Then, let $F'$ be the graph formed as the direct sum of all such trees $T_v$. Then, for each edge $\hat{e} = (\hat{u}, \hat{v})$ in $F \subseteq H$ that stems from core graph $C(G, F_i, \tilde{\str}_i)$ (recall $H$ is a subgraph of $\tilde{G}$ which is formed as the union of these core graphs) and originates from mapping the edge $e =(u,v) \in G$, add an edge between the vertices identified with vertices $u$ and $v$ in the core graph $C(G, F_i, \tilde{\str}_i)$ in the subgraph $T_v$ of $G'$ of length $l(e)$.

Note that this process adds for any such edge $\hat{e}$ as above, a path $\hat{u}$ to $\hat{v}$ to $F'$ of length $l(F[\hat{u} = \root^{F'}(u), u]) + l(e) + l(F[v, \hat{v} = \root^{F'}(v)])$. But this is exactly the length that the edge $\hat{e}$ was assigned in its core graph, and thus in the graphs $\tilde{G}$ and $H$. Arguing carefully, one can then establish that the distances in $F$ are preserved by $F'$.

Further, we note that each tree in forest $F_i$ is added to only one graph $T_v$ (the one for $v$ being the root of the tree). And trees in $F_i$ are vertex-disjoint. Thus, each vertex $w \in V$, appears in at most $\lambda$ trees $T_v$. And thus, $F'$ contains each vertex and edge in $G$ at most $\lambda$ times. Thus, one can establish that $F'$ embeds into $G$ with low congestion. 

Applying this mapping from graph $H$ into $G$ recursively in the multi-step vertex sparsifier reduction, one can still bound the congestion by a subpolynomial factor in $n$, as desired.

\section{Preliminaries}
\label{sec:prelim}

\paragraph{Standard Definitions.} In this article, we consider $G=(V,E,l)$ to be an $m$-edge, $n$-vertex graph with edge lengths being integers in the interval $[1, L]$ where $L$ is polynomially bounded in $n$. This is without loss of generality as all our results can then be extended to work with any upper bound $L$ at the cost of an additional $O(\log(L))$ factor in the runtime by standard reductions (see for example Proposition II.1.2, in \cite{bernstein2022deterministic}). 

\paragraph{Dynamic Graphs.} \label{para:dynG} 
We say $G$ is a \emph{dynamic} graph, if it undergoes \emph{batches} $U^{(1)}, U^{(2)}, \ldots$ of updates consisting of edge insertions/deletions and/or vertex splits that are applied to $G$. We stress that results on dynamic graphs in this article often only consider a subset of the update types and we therefore explicitly state for each dynamic graph which updates are allowed. We say that the graph $G$, after applying the first $t$ update batches $U^{(1)}, U^{(2)}, \ldots, U^{(t)}$, is at \emph{time} $t$ and denote the graph at this time by $G^{(t)}$. Additionally, when $G$ is clear, we often denote the value of a variable $x$ at the end of time $t$ of $G$ by $x^{(t)}$, or a vector $\bx$ at the end of time $t$ of $G$ by $\bx^{(t)}$. 

For each update batch $U^{(t)}$, we encode edge insertions by a tuple of tail and head of the new edge and deletions by a pointer to the edge that is about to be deleted. We further also encode vertex splits by a sequence of edge insertions and deletions as follows: if a vertex $v$ is about to be split and the vertex that is split off is denoted $v^{\text{NEW}}$, we can delete all edges that are incident to $v$ but should be incident to $v^{\text{NEW}}$ from $v$ and then re-insert each such edge via an insertion (we allow insertions to new vertices, that do not yet exist in the graph). 

For technical reasons, we assume that in an update batch $U^{(t)}$, the updates to implement the vertex splits are last, and that we always encode a vertex split of $v$ into $v$ and $v^{\text{NEW}}$ such that $\deg_{G^{(t+1)}}(v^{\text{NEW}}) \leq \deg_{G^{(t+1)}}(v)$. We let the vertex set of graph $G^{(t)}$ consist of the union of all endpoints of edges in the graph (in particular if a vertex is split, the new vertex $v^{\text{NEW}}$ is added due to having edge insertions incident to this new vertex $v^{\text{NEW}}$ in $U^{(t)}$).

\paragraph{Distances, Balls, Bunches and Clusters.} We denote by $\dist_G(u,v)$ the distance from vertex $u$ to $v$ in the graph $G$, and by $\dist(u, X)$ for some vertex set $X \subseteq V$, the distance from $u$ to the closest vertex in $X$. We assume that graph $G$ has unique shortest paths at any point in time which can be assumed w.l.o.g. at the cost of constant time additional overhead. We denote by $\pi_G(u,v)$ the unique shortest path in $G$ from $u$ to $v$. Given a path $\pi$ in $G$, for any two vertices $u,v$ on the path $\pi$, we denote by $\pi[u,v]$ the segment of the path $\pi$ from $u$ to $v$. Again, we extend this notion and denote by $\pi[u,X]$, for any $u$ and set $X \subseteq V$ such that $u$ is on $\pi$ and at least one vertex in $X$ is on $\pi$, the subsegment of $\pi$ from $u$ to the vertex in $X$ closest to $u$ on $\pi$ (whenever we use this notation, we ensure that there is a unique vertex in $X$ that minimizes this distance to avoid ambiguity).

We define $B_G(v, r) = \{ w \in V \;|\; \dist_G(v,w) < r\}$ to be the open ball around $v$ of radius $r$ in $G$ and $\bar{B}_G(v,r) = \{ w \in V \;|\; \dist_G(v,w) \leq r\}$ to be the closed ball around $v$. For convenience, we define for any two vertices $v,x$, the ball $B_G(v,x) = B_G(v, \dist_G(v,x))$, and for $X \subseteq V$, $B_G(v, X) = B_G(v, \dist_G(v, X))$, and similiarly define $\bar{B}_G(v, x)$ and $\bar{B}_G(v, X)$.
We define cluster as inverses to balls, defining $C_G(v, r) = \{ w \in V \;|\; v \in B_G(w,r)\}$ and define $C_G(v, x)$ and $C_G(v, X)$ analogously.

\paragraph{Graph Embeddings and Hierarchical Graphs.} Next, we discuss the precise definitions for graph embeddings used in this article. We start by defining a graph embedding.

\begin{definition}[Graph Embedding]
Given two graphs $H, G$, and a vertex map $\Pi_{V(H) \mapsto V(G)}$ that maps every vertex in $H$ to a vertex in $G$, we say that a map $\Pi_{H \mapsto G}$ is a graph embedding of $H$ into $G$ if it maps each $e = (u,v) \in H$ to a $xy$-path $\Pi_{H \mapsto G}(e)$ in $G$ for $x = \Pi_{V(H) \mapsto V(G)}(u)$, and $y = \Pi_{V(H) \mapsto V(G)}(v)$. We say that $\Pi_{H \mapsto G}$ is a \emph{flat} graph embedding if $\im(\Pi_{H \mapsto G}) \subseteq E(G)$, i.e. every edge in $H$ maps to a single edge in $G$.

Since $\Pi_{V(H) \mapsto V(G)}$ is implicitly defined by graph embedding $\Pi_{H \mapsto G}$, we often omit to state the vertex map explicitly.
\end{definition}

\begin{definition}[Edge Congestion of Paths and Embeddings]\label{def:graphEmbeddingEdgeCong}
Given a set of paths $P_1, P_2, \ldots, P_k$ in graph $G$, we define the edge congestion induced by a collection of paths for an edge $e \in E(G)$ by \[\econg( \{ P_i \}_{i \in [1, k]}, e) = \sum_{i \in [1, k]} \sum_{e' \in P_i} \vecone[e = e'].\]  We define the edge congestion by $\econg( \{ P_i \}_{i \in [1, k]}) = \max_{e \in E} \econg( \{ P_i \}_{i \in [1, k]}, e)$. We define the congestion of a graph embedding $\Pi_{H \mapsto G}$ by $\econg(\Pi_{H \mapsto G}) = \econg(\im(\Pi_{H \mapsto G}))$.
\end{definition}

\begin{definition}[Hierarchical Forest/ Tree]\label{def:hierarchicalTree}
Given a graph $G$ and a vertex set $X \subseteq V(G)$, we say a forest $F$ associated with vertex maps $\Pi_{X \mapsto V(F)}$, $\Pi_{V(F) \mapsto V(G)}$ and graph embedding $\Pi_{F \mapsto G}$ if $X \subseteq \im(\Pi_{V(F) \mapsto V(G)})$ is a \emph{hierarchical forest} over $X$. If $X = V(G)$, we also say that $F$ is a hierarchical forest over $G$.

We often refer to the $V(T)$ as a \emph{node set} to distinguish from the vertex set $V(G)$. We say a node $x \in V(T)$ is identified with vertex $v \in V(G)$ if $v = \Pi_{V(F) \mapsto V(G)}(x)$.

We also say that $F$ is a hierarchical tree if $F$ is a tree graph. We say that $F$ is a \emph{flat} hierarchical forest/ tree if $\Pi_{F \mapsto G}$ is a flat graph embedding. 
\end{definition}

\section{A Fully-Dynamic Vertex Sparsifier}
\label{sec:vertexSparsifer}
The main result of this section is summarized by the following theorem. It allows us to either extract distances directly from the data structure or has them preserved in a vertex sparsifier of low recourse.

\begin{restatable}{theorem}{workhorseSparse}\label{lma:workhorseSparsifier}
Given a size reduction parameter $k$, a number of levels $1 \leq K \leq o\left(\frac{\log^{1/6} m}{\log\log m}\right)$, a degree threshold $\Delta$ and an $n$-vertex $m$-edge (multi-)graph $G = (V,E, l)$ undergoing a sequence of edge insertions and deletions and isolated vertex insertions/ deletions, where at all times, lengths in $G$ are in $[1, L]$ and the maximum degree in $G$ is at most $\Delta$. Then, for some $\gamma_{\ell} = \exp(O(\log^{2/3} m \cdot \log\log m)), \gamma_{recVS} =\Otil(m^{4/K} \cdot (\gamma_{\ell})^4)$, there is a deterministic algorithm that can be initialized in time $\tilde{O}(m \cdot k^4 + m \cdot \Delta + m \gamma_{\ell})$, processes every update with worst-case time $\Otil(k^4 + \Delta \cdot  \gamma_{recVS} (\gamma_{\ell})^{O(K^2)} k^2)$ and explicitly maintains
\begin{enumerate}
    \item a monotonically growing vertex set $A \subseteq V(G)$, such that $A$ contains at any time the endpoints of all previously deleted/inserted edges and vertices that were added to $G$ after initialization (however, vertices that are removed from $G$ are also removed from $A$), and $|A| \leq n/k + 2q$.
    \item a pivot function $p$ that maps each vertex $v \in V$ to its closest vertex $p(v)$ in $A$ (ties broken arbitrarily but consistently) in the current graph $G$. 
    \item \label{prop:dynamicShortestPathPivotForest} a dynamic forest $F \subseteq G$ such that at any time, for every $v \in V$, $\dist_{F}(v, p(v)) = \dist_G(v, p(v))$ where $F$ undergoes at most $\Otil(k)$ changes per update to $G$ that are explicitly outputted and every tree in $F$ consists of at most $\Otil(k)$ vertices, and 
    \item the exact distances and shortest paths from each vertex $v \in V$ to the vertices in its ball $B_G(v, p(v)) \cup \{p(v)\}$ where each shortest path consists of $\O(k)$ edges and every edge appears on at most $\Otil(k^2)$ such paths. Further, each such shortest path has already been the shortest path to said vertex since the initialization. 
    \item \label{prop:workhorseDistancePreserve} a dynamic graph $H$ with vertex set $A \subseteq V(H)$ where  the algorithm initially outputs $H^{(0)}$ with $\Delta \cdot \gamma_{recVS} \cdot m/k$ edges and vertices. Then at any stage $t \geq 1$, it outputs a batch of updates $U_H^{(t)}$ consisting of edge insertions and deletions, and isolated vertex insertions such that when $U_H^{(t)}$ is applied to $H^{(t-1)}$ it yields $H^{(t)}$ and:
    \begin{enumerate}
        \item $|U_H^{(t)}| \leq \Delta \cdot \gamma_{recVS}$, and 
        \item \label{prop:maxDegreeVS} at any stage, the maximum vertex degree of $H$ is at most $\Delta \cdot \gamma_{recVS}$, and
        \item at any stage, $H$ has lengths in $[1, nL]$, and 
        \item at any stage, for every two vertices $u,v \in V(H)$, we have $\dist_G(u,v) \leq \dist_{H}(u,v)$ and if $u,v \in A$, we additionally have $ \dist_{H}(u,v) \leq (\gamma_{\ell})^{O(K)} \cdot \dist_G(u,v)$. \label{prop:approxVS}
        \item at any stage, given any edge $e = (u,v) \in H$, the algorithm can return a $uv$-path $P$ in $G$ with $l_G(P) \leq l_H(e)$ in time $O(|P|)$. \label{prop:mapBackToPath}
    \end{enumerate}
\end{enumerate}
\end{restatable}

Henceforth, we assume wlog that the dynamic graph $G$ input to \Cref{lma:workhorseSparsifier} has unique distances. This removes ambiguity when defining the pivot function $p$ and can be implemented with additional $O(1)$ overhead per operation (see \cite{demetrescu2004new}).

In \Cref{subsec:algoVS}, we give the algorithm behind \Cref{lma:workhorseSparsifier}. In \Cref{subsec:analysisVS}, we analyze the algorithm from \Cref{subsec:algoVS} and prove that it indeed satisfies the guarantees given in \Cref{lma:workhorseSparsifier}. Finally, in \Cref{subsec:mappingHierachVS}, we show how to extend \Cref{lma:workhorseSparsifier} to support the maintenance of low-diameter hierarchical trees. 

\subsection{The Algorithm} 
\label{subsec:algoVS}
\paragraph{Maintaining Pivots, Balls and Shortest Paths.} We initialize the set $A$ by using the following procedure by Thorup and Zwick (see \cite{thorup2001compact}, Theorem 3.1). We derandomize their result with standard techniques to obtain the following theorem. For completeness, we state the derandomized procedure and prove correctness in \Cref{sec:bunchesAndClusters}.

\begin{theorem}
\label{thm:TZschemes}
Given a graph $G=(V,E,l)$ and a size reduction parameter $k$, there is an algorithm $\textsc{Center}(G, k)$ that in time $\tilde{O}(mk)$ computes a set $A \subseteq V$ of size at most $n/k$ such that for every vertex $v \in V$, we have $|B_G(v, A)| \leq 2k \log n$ and $|C_G(v,A)| \leq 2k \log n$.
\end{theorem}

After initializing $A$ to be the set returned by procedure $\textsc{Center}(G, k)$ from \Cref{thm:TZschemes}, we maintain $A$ by adding vertices that are affected by the updates applied to $G$. We maintain the pivot function $p$ over the vertices such that for each vertex $v \in V$, at any time, $p(v)$ is the closest vertex to $v$ from the current set $A$ in the current graph $G$.

From \cite{thorup2005approximate}, we further have the following result that is also immediate from the subpath-property of shortest paths.

\begin{theorem}\label{thm:realTZ}
Given a graph $G=(V,E,l)$ with unique distances, a set $A \subseteq V$ and a pivot function $p$ that maps each vertex $v \in V$ to the closest vertex $p(v)$ in the set $A$. Then, the union of vertex-to-pivot paths $\{ \pi_{G}(v, p(v)) \}_{v \in V}$ forms a forest $F$ where each component can be rooted at a vertex in $A$.
\end{theorem}

We discuss in the runtime analysis the precise implementation details for maintaining pivots, associated balls and shortest paths within the ball and to the pivots of vertices, and describe how to maintain the forest $F$ described in \Cref{thm:realTZ} for every version of $G$.

\paragraph{Maintaining LSSFs.} Further, we maintain a set of Low-Stretch Spanning Forests that we then use to obtain the final sparsifier $H$. Before we describe the precise algorithm, we state the definition of a Low-Stretch Spanning Forest and state a Lemma from \cite{chen2022maximum} that allows us to maintain a low-stretch spanning forest efficiently. 

\begin{definition}[Rooted Spanning Forest, Stretch Induced by Forests]
  \label{def:spanningforest}
  A \emph{rooted spanning forest} of a graph $G = (V, E, l)$ is a forest $F$ on $V$ such that each connected component of $F$ has a unique distinguished vertex known as the \emph{root}. We denote the root of the connected component of a vertex $v \in V$ as $\root^F(v)$. We define the stretch of an edge $(u,v) \in E$ induced by $F$ by
   \begin{align*}
    \str^{F,G}(e) \defeq
    \begin{cases}
      1 + l(F[u,v])/ l(e) &~\text{ if } \root^F(u) = \root^F(v) \\
      1 + \left(l(F[u, \root^F(u]) + l(F[v, \root^F(v)]) \right)/l(e) &~\text{ if } \root^F(u) \neq \root^F(v).
    \end{cases}
  \end{align*}
    We also define the stretch induced by $F$ on a $uv$-path $P$ by $\str^{F,G}(P) = \sum_{e \in P} \str^{F, G}(e)$. We say that $\wstr(e)$ is a \emph{stretch overestimate} for the stretch of $e$ if $\str(e) \leq \wstr(e)$. Given stretch overestimates on the edges $\wstr$, we also define overestimates on paths by $\wstr(P) = \frac{1}{l(P)} \cdot \sum_{e \in P} \wstr(e) \cdot l(e)$ which implies $\str^{F, G}(P) \leq \wstr(P)$. 
\end{definition}

\begin{lemma}[see Lemma 6.5, \cite{chen2022maximum}]
  \label{lemma:globalstretch}
  Given an $m$-edge graph $G=(V,E,l)$ with maximum-degree at most $\Delta$, a weight function $w$ over the edges, and a parameter $k$. There is a deterministic algorithm that is initialized in time $\O(m)$ and outputs a tree $T$, stretch overestimates $\wstr(e)$ for each edge $e$ and maintains a rooted spanning forest $F \subseteq G$. The algorithm then supports the following updates in worst-case update time $\O(\Delta \cdot k)$:
  \begin{itemize}
      \item $\textsc{InsertEdge}(e)/ \textsc{DeleteEdge}(e)/ \textsc{AddIsolatedVertex}()$: If an edge $e$ is inserted, it is inserted into $G$. If an edge $e$ is removed, we remove the edge from $G$ and $F$ (if it exists in $F$). If a new vertex is added, we add this vertex to both $G$ and $F$ as an isolated vertex and then return the identifier of the new vertex.
      \item $\textsc{DeleteEdgeFromForest}(e)$: The operation deletes the edge $e$ from forest $F$, or for the last operation, 
      adds a new vertex to $F$ that is isolated and returns the identifier of the new vertex. 
  \end{itemize}
    Note that the above operations enforce that at any time, we have $F \subseteq G$. Under these operations, the algorithm maintains $F$ such that: 
  \begin{enumerate}
  \item the edge set of $F$ is a monotonically decreasing set and the vertex set of $F$ is only changed by the data structure operation $\textsc{AddIsolatedVertex}()$ which adds a new (isolated) vertex to $F$, and
  \item initially the forest $F$ has at most $O(m/k)$ connected components and every update increases the number of connected components by at most $\gamma_{recLSD} = O(\log^2 m)$.  \label{item:cccountLem}
  \item every vertex that becomes a root in the forest $F$ at some point, remains a root for the rest of time. Further, all vertices incident to the set $S$, to an edge inserted or deleted or to/from $G$, or having been added to $G$ as an isolated vertex at some point, become roots in the forest $F$.   \label{item:affectedVerticesBecomeRoots}
  \item every connected component $T$ of $F$ is incident to at most $k \cdot \Delta$ edges. \label{item:degboundLem}
  \item the stretch overestimates $\wstr(e)$ for edges $e \in E^{(0)}$, the initial edge set, remain fixed to the value that they are initialized to, new edges $e$ added by either edge insertions or vertex splits have $\wstr(e)$ initialized to $1$ and then fixed throughout. 
  
  At any time, we have for any edge $e$ in the current graph that $\str^{F, G}(e) \leq \wstr(e)$, i.e. $\wstr(e)$ is a stretch overestimate at all times. Finally, the algorithm guarantees that $\sum_{e \in E^{(0)}} w(e) \cdot \wstr(e) \le \gamma_{LSST} \cdot \sum_{e \in E^{(0)}} w(e)$ for some $\gamma_{LSST} = \Otil(1)$. \label{item:avgstretchboundLem}
  \end{enumerate}
\end{lemma}

We remark that  \Cref{lemma:globalstretch} differs from Lemma 6.5 in \cite{chen2022maximum} in the following ways: we added an initial set $S$ which can be implemented by adding a self-loop to every vertex in $S$ in the initial graph and deleting the self-loops before resuming with the real updates to $G$; Property \ref{item:affectedVerticesBecomeRoots} was not explicitly stated in \cite{chen2022maximum} but can be extracted from their proof straightforwardly; and Property \ref{item:degboundLem} is not given, but instead Lemma 6.5 in \cite{chen2022maximum} maintains a partition of the edge set of $F$ denoted by $\cW$ in their Lemma from which this Property can be derived straightforwardly.

Note further that the notation in \Cref{lemma:globalstretch} differs slightly from the notation in the rest of the section where we rather describe the set of updates to $F$ by batches of updates that are given to the data structure after every update to $G$. Here, we use a slightly different way to formalize the interface since we later use the output of the data structure, i.e. the forest $F$, to create the update sequence to $F$ until a certain condition for $F$ is met. Thus, the update sequence to $F$ while processing an update to $G$ is created using the data structure itself and is only known by the end of this iterative process.

In our algorithm, we want to use a collection of LSSFs from \Cref{lemma:globalstretch} to encode all distances in the graph $G$. 
Therefore, we generate in the initial graph $G$ the following path set $\mathcal{P}$ as follows:
\begin{itemize}
    \item for any two vertices $u \in V$, $v \in B_G(u,A)$ and edge $e = (v,x) \in E(G)$, add path $\pi_G(u,v) \concat e $ to $\mathcal{P}$.
\end{itemize}

Next, we construct the projection of this set onto potential pivots. Note that for any vertex $v \in V$, we have that the pivot of $v$ at any point of the algorithm has to be in the set $B_G(v, A) \cup \{p(v)\}$ where $p(v)$ is the pivot of $v$ at initialization time. This follows from the fact that the distance to $A$ is monotonically decreasing over time as we later show in \Cref{obs:decrementalMaintenanceofTZ}. We construct this collection denoted by $\projP$ as follows:
\begin{itemize}
    \item for any $ux$-path $P$ in $\mathcal{P}$, and any $\hat{u} \in B_G(u, A) \cup \{p(u)\}$ and $\hat{x} \in B_G(x, A) \cup \{p(x)\}$, add the path $\pi_G(\hat{u}, u) \concat P \concat \pi_G(x, \hat{x})$ to $\projP$.
\end{itemize}

Again, we construct $\projP$ at initialization time. Given this set, the algorithm constructs iteratively path collections $\projP = \projP_0 \supseteq  \projP_1 \supseteq \ldots \supseteq \projP_{\lambda} = \emptyset$ and forests $F_0, F_1, \ldots, F_{\lambda- 1}$ by invoking \Cref{lemma:globalstretch} on the graph $G$ with initial vertex subset $A$, edge weights $w_i(e) \defeq \econg(\projP_{i}, e) \cdot l(e)$ for all $e \in E$ and parameter $k$, and we let the corresponding stretch overestimates be denoted by $\wstr_0, \wstr_1, \ldots, \wstr_{\lambda - 1}$. Here, for every $0 \leq i < \lambda$, the paths in $\projP_{i+1}$ are the paths $P \in \projP$ such that for all $j \leq i$, $\wstr_j(P) > 2 \cdot \gamma_{LSST}$. This concludes the initialization of the LSSFs.

We then update the data structures from \Cref{lemma:globalstretch} by forwarding the updates to $G$ to each of these data structures.

\paragraph{Maintaining Vertex Sparsifier $\tilde{G}$ from $F_0, F_1, \ldots, F_{\lambda-1}$.} Next, we decribe how to obtain a vertex sparsifier $\tilde{G}$ of $G$ as the direct sum of the core graphs with respect to forests $F_0, F_1, \ldots, F_{\lambda -1}$. Core graphs have already played a prominent role in \cite{chen2022maximum} and are defined as follows.

\begin{definition}[Core graph]
  \label{def:coregraph}
 Given a graph $G=(V,E,l)$, a rooted spanning forest $F$ of $G$ and stretch overestimate $\wstr(e)$ for each edge $e \in E$. We define the \emph{core graph} $\mathcal{C}(G, F, \wstr)$ to be the graph obtained from contracting every connected component in $F$ into the root vertex of the component, i.e. the vertex set of $\mathcal{C}(G, F, \wstr)$ is the set of roots of $F$. We define the length function $l_{\mathcal{C}(G, F, \wstr)}$ of $\mathcal{C}(G, F, \wstr)$ as follows: for every $e = (u, v) \in E(G)$ with image $\hat{e} \in E(\mathcal{C}(G, F, \wstr))$, we define its length as $l_{\mathcal{C}(G, F, \wstr)}(\hat{e}) \defeq \wstr(e) \cdot l(e)$. 
\end{definition}

Given this definition, it would be most natural to maintain vertex sparsifier $\hat{G}$ as the union of core graphs $\mathcal{C}(G,  F_0, \wstr_0)$, $\mathcal{C}(G, F_1, \wstr_1), \ldots, \mathcal{C}(G, F_{\lambda-1}, \wstr_{\lambda-1})$ where each forest $F_i$ is maintained by a data structures from \Cref{lemma:globalstretch}. But while such a vertex sparsifier $\hat{G}$ preserves distances reasonably well, it is hard to maintain since the update sequence to $\hat{G}$ would have to undergo vertex splits and merges. 

Instead, we maintain the vertex sparsifier $\tilde{G}$ which we take as the direct sum of the core graphs $\mathcal{C}(G,  F_0, \wstr_0)$, $\mathcal{C}(G, F_1, \wstr_1), \ldots, \mathcal{C}(G, F_{\lambda-1}, \wstr_{\lambda-1})$ where each forest $F_i$ is maintained by a data structures from \Cref{lemma:globalstretch},  and additionally we have an edge of length $0$ between any vertices $u \in V(\mathcal{C}(G, F_i, \wstr_i)), v \in V(\mathcal{C}(G, F_j, \wstr_j))$ that are identified with the same vertex $v$ in $V(G)$. The sequence of updates to $\tilde{G}$ can now be described only by edge updates, isolated vertex insertions, and vertex splits, but no vertex merges which is crucial for efficiency.

\paragraph{Maintaining $H$ as an Edge-Sparsifier of $\tilde{G}$.} We start this section by revisiting the following result that is given in \cite{detMaxFlow}. Here, we state a stronger version of Theorem 7.2 in \cite{detMaxFlow}: most importantly, bounds on recourse and update time are worst-case whereas in \cite{detMaxFlow} these bounds are stated amortized. Our strengthening is achieved via two observations: firstly, while the original statement in \cite{detMaxFlow} only updates graph $H$ via edge insertions/ deletions and vertex insertions, we additionally allow for the vertex splits to $G$ to be forwarded to $H$. This gives us more control later as we will see and the bound is obtained from inspecting the proof in \cite{detMaxFlow} carefully; secondly, we obtain worst-case bounds because the algorithm in \cite{detMaxFlow} uses a standard batching technique that can be de-amortized using standard techniques (see for example \cite{willard1985adding, thorup2004fully, gutenberg2020fully}). In fact, Section 8 in \cite{detMaxFlow} already claims that the de-amortization technique works in the same way as it is specified here.

Finally, the precise statement in \cite{detMaxFlow} only claims to work for unweighted/ unit-length graphs, and does not allow for edge insertions. But the more general version below is easily obtained by a simple length bucketing scheme and by adding edges that are inserted directly into the sparsifier (both of which are standard techniques for dynamic spanners, see for example \cite{bernstein2022fully}). Again, when using the sparsifier in \cite{detMaxFlow}, this strengthening is already used by appealing to the standard techniques for spanners.

\begin{theorem}[see Theorem 7.2, \cite{detMaxFlow}]
\label{thm:spanner}
Given an $m$-edge $n$-vertex undirected, dynamic graph $G$ with quasi-polynomially-bounded lengths, undergoing update batches $U_G^{(1)}, U_G^{(2)}, \ldots$ consisting of edge insertions/ deletions and isolated vertex insertions and vertex splits. There is a deterministic algorithm with parameter $1 \le \Lambda \le o\left(\frac{\log^{1/6} m}{\log\log m}\right)$, that maintains a spanner $H$ and an embedding $\Pi_{G \to H}$ such that for some $\gamma_{\ell} = \exp(O(\log^{2/3} m \cdot \log\log m))$ and $\gamma_{recES} = \Otil(1)$, we have
\begin{enumerate}
    \item \label{prop:sparsitySpanner} \underline{Sparsity and Low-Recourse:} the algorithm initially outputs $H^{(0)}$ with $\gamma_{recES} \cdot n \cdot \gamma_{\ell}$ edges. Then at any stage $t \geq 1$, where $G$ undergoes updates $U_G^{(t)}$ it outputs a batch of updates $U_H^{(t)}$ such that when applied to $H^{(t-1)}$ produce $H^{(t)}$ such that $H^{(t)} \subseteq G^{(t)}$. We further have that
    \begin{itemize}
        \item $U_H^{(t)}$ contains the vertex split updates that are present in $U_G^{(t)}$, 
        \item all other updates in $U_H^{(t)}$ are edge insertions/ deletions or isolated vertex insertions.
        \item $|U_H^{(t)}| \leq |U_G^{(t)}| \cdot \gamma_{recES} \cdot  n^{1/\Lambda}\gamma_{\ell}$.
    \end{itemize}
    \item \underline{Distance Preservation:} at any stage, for any $u,v \in V$, $\dist_G(u,v) \leq \dist_H(u,v) \leq (\gamma_{\ell})^{O(\Lambda)} \cdot \dist_G(u,v)$. \label{prop:distancePresSpanner}
\end{enumerate}
The algorithm takes initialization time $\tilde{O}(m \gamma_{\ell})$ and processing the $t$-th update batch $U_G^{(t)}$ takes worst-case update time $\tilde{O}(|U_G^{(t)}| \cdot n^{1/\Lambda}(\gamma_{\ell})^{O(\Lambda^2)}(\Delta_{\max}(G))^2)$.
\end{theorem}

We can now describe how we maintain the final sparsifier $H$. In \Cref{algo:initSparsifier}, we describe the initialization procedure. The algorithm is rather straightforward and follows the discussion of the previous sections: it initializes the forests $F_0, F_1, \ldots, F_{\lambda-1}$ as discussed previously, maintains $\tilde{G}$ as discussed, and then maintains the final sparsifier $H$ by applying \Cref{thm:spanner} to graph $\tilde{G}$ from which it obtains graph $\tilde{H}$, and then contracting all vertices in $\tilde{H}$ that are identified with the same vertex $v$ in $V(G)$. Thus, $V(H) \subseteq V(G)$. 

Additionally, in the while-loop of the initialization algorithm, the algorithm checks for the existence of vertices that have large degree in the graph $\tilde{H} \cap \cC(G, F_i, \wstr_i)$ over all $i$, where we take the intersection to mean the graph $\tilde{H}$ where only the edges are present that originate from the core graph $\cC(G, F_i, \wstr_i)$. While such a vertex $v$ and index $i$ exists, we then call the subprocedure given in \Cref{algo:reduceDegree} to reduce the degree of $v$ w.r.t. to the graph $\tilde{H} \cap \cC(G, F_i, \wstr_i)$. The while-loop ensures that on termination, the degrees of vertices in $\tilde{H}$ are small.

Finally, the algorithm outputs the graph $H$ that is obtained from $\tilde{H}$ by contracting all vertices in $\tilde{H}$ that are identified with the same vertex in $G$.

\begin{algorithm}
Initialize data structures $\mathcal{D}_0, \mathcal{D}_1, \ldots, \mathcal{D}_{\lambda-1}$ to obtain initial stretch overestimates $\wstr_0, \wstr_1, \ldots, \wstr_{\lambda-1}$ and maintain forests $F_0, F_1, \ldots, F_{\lambda -1}$ as described in the previous section.\\
Maintain $\tilde{G}$ as described in the previous section as the union of core graphs $\cC(G, F_i, \wstr_i)$. \\
Let $\tilde{H}$ be the graph maintained by applying \Cref{thm:spanner} to graph $\tilde{G}$ with parameter $\Lambda = K$.

\tcc{ Here $\gamma_{degConstr}$ is a constant fixed later. }
\While(\label{lne:whileLoopDegreeConstraint}){$\exists v \in V(\tilde{H})$ and $0 \leq i < \lambda$ with $\deg_{\tilde{H} \cap \cC(G, F_i, \wstr_i)}(v) > 2 \cdot \gamma_{degConstr} \cdot \Delta$}{
    $\textsc{ReduceDegree}(v, i, \gamma_{degConstr})$.
}
\Return graph $H^{(0)}$ obtained from contracting all vertices in $\tilde{H}$ that are identified with the same vertex $v \in V(G)$.
\caption{$\textsc{InitSparsifier}()$}
\label{algo:initSparsifier}
\end{algorithm}

Let us next describe the subprocedure given in \Cref{algo:reduceDegree} that achieves this goal. The procedure heavily relies on the following classic result on tree partitioning that is obtained straightforwardly from \cite{frederickson1983data} (in \cite{frederickson1983data} the procedure is assumed to run on a graph of maximum degree $3$ but the extension is straightforward).  

\begin{theorem}[see \cite{frederickson1983data}, Lemma 1]\label{thm:partitionFrederickson}
Given a tree $T$ spanning a subset of vertices in an $m$-edge graph $G$ of maximum degree $\Delta$ and a positive integer $z$. Then, there is a procedure $\textsc{FindSets}(T, G, z)$ that returns a set of edges $E'$ such that every connected component in $T \setminus E'$ is incident to at most $\Delta \cdot z$ edges and all but one component is incident to at least $z$ edges in $G$. The algorithm runs in time $O(m)$.
\end{theorem}

\Cref{algo:reduceDegree} uses this procedure to compute a set $E'$ of edges in $F_i$ that when removed reduce the degree of $v$ significantly. It deletes these edges in $E'$ from $F_i$ by forwarding them to data structure $\mathcal{D}_i$ which then updates the forest $F_i$. Thereafter, graphs $\tilde{G}$ and $\tilde{H}$ are updated accordingly. As we will show in the analysis, for reasonably large value $\gamma_{degConstr}$, the process results in the degree of $v$ decreasing significantly, and the overall process is terminating quickly. 

\begin{algorithm}
 Let $\hat{E}_i$ be the set of edges in $\tilde{H} \cap \cC(G, F_i, \wstr_i)$ incident to $v$; let $E_i$ denote the pre-images of the edges in $\hat{E}_i$.\\
Let $T_i$ be the tree in forest $F_i$ that has $v$ as its root. \\
$E' \gets \textsc{FindSets}(T_i, G[E_i], z)$.\\
Invoke operation $\textsc{DeleteEdgeFromForest}(\cdot)$ on data structure $\mathcal{D}_i$ to remove the edges in $E'$ from forest $F_i$.\\
Update $\tilde{G}$ and $\tilde{H}$ accordingly.
\caption{$\textsc{ReduceDegree}(v, i, z)$}
\label{algo:reduceDegree}
\end{algorithm}

Finally, to process the $t$-th update to $G$, we invoke \Cref{algo:updateSparsifier} with parameter $t$. The procedure first forwards the update to the data structures $\mathcal{D}_i$ which results in updates to $F_i$ and then updates $\tilde{G}$ and $\tilde{H}$ accordingly. It then picks the vertex $v$ in $\tilde{H}$ of largest degree with respect to graph $\tilde{H} \cap \cC(G, F_i, \wstr_i)$ for some $i$ and invokes the subprocedure from \Cref{algo:reduceDegree} to decrease its degree in $\tilde{H}$.

\begin{algorithm}
\ForEach{$0 \leq i < \lambda$}{
    Update $\mathcal{D}_i$ by using one of the update operations  $\textsc{InsertEdge}(e)/ \textsc{DeleteEdge}(e)/ \textsc{AddIsolatedVertex}()$ to forward the $t$-th update to $G$ to each such data structure.
}
Update $\tilde{G}$ and $\tilde{H}$ accordingly.\\
Let $v \in V(\tilde{H})$ and $0 \leq i < \lambda$ be chosen to maximize $\deg_{\tilde{H} \cap \cC(G, F_i, \wstr_i)}(v)$.\label{lne:pickHighestDegreeVertex}\\
$\textsc{ReduceDegree}(v, i, \gamma_{degConstr})$.\label{lne:reduceDegree}\\
Update $H$ accordingly.\\
\Return let $U_H^{(t)}$ to be the set of updates required to update $H^{(t-1)}$ to $H^{(t)}$ where $U_H^{(t)}$ only consists of edge insertions/deletions and isolated vertex insertions.
\caption{$\textsc{MaintainSparsifier}(G, t)$}
\label{algo:updateSparsifier}
\end{algorithm}

Finally, it updates $H$ which is defined to be the graph $\tilde{H}$ where vertices in $\tilde{H}$ that are identified with the same vertex in $G$ are contracted. The algorithm returns an update batch that reflects the changes to $H$.

This update batch $U_H^{(t)}$ is obtained as follows: for every edge inserted/ deleted to $\tilde{H}$, we forward this update straightforwardly. For every vertex split in $\tilde{H}$ where a vertex $v$ is split and $v^{NEW}$ is split off, we add to $U_H^{(t)}$ the following updates to emulate the split: we first delete all edges incident to $v^{NEW}$ from the vertex identified by $v$ from $H$, then if the vertex in $G$ that is identified with $v^{NEW}$ is not yet present in $H$, we add it via an isolated vertex insertion, and then we add all edges incident to $v^{NEW}$ back into the graph $H$.

Note that the number of updates in $U_H^{(t)}$ might be larger than the number of updates in $U_{\tilde{H}}^{(t)}$ because of the emulation process. However, we show that replacing vertex splits with the above batch of updates necessary to emulate the vertex splits does not increase the number of updates significantly where we leverage that $\tilde{H}$ and $H$ have small maximum degree at all times.

\subsection{Analysis} 
\label{subsec:analysisVS}

\paragraph{Analyzing Pivots, Balls, and Shortest Paths.} The following claim summarizes the key insight into the first part of the data structure.

\begin{claim}\label{obs:decrementalMaintenanceofTZ}
For any $u \in V$, we have that $B_G(u, A)$ and $C_G(u, A)$ are monotonically decreasing sets. We further have that for any vertex $v \in B_G(u,A)$ and edge $(v,x) \in E$ if the current shortest path $\pi_{G}(u,x) = \pi_G(u,v) \concat (v,x)$, then we have that $\pi_{G}(u,x)$ is equal to the shortest path from $u$ to $x$ in the initial graph $G^{(0)}$.
\end{claim}
\begin{proof}
Let us prove the first statement. Consider first the ball $B_G(u, A)$. We claim that the radius of this ball is monotonically decreasing over time. To see this, assume for the sake of contradiction that the radius of $B_G(u,A)$ would increase due to an edge deletion $(u,v)$. This implies that the vertex in $A$ closest to $u$ before the deletion, moves further away from $u$ due to the deletion. But this implies that the deletion affects the previous shortest path from $u$ to this vertex in $A$. This then implies that one of the endpoints of the deleted edge is strictly closer to $u$ than the previously closest vertex in $A$. And since we have that the endpoint of the deleted edge closer to $u$ has its shortest path to $u$ unaffected by the edge deletion, the radius of $B_G(u,A)$ strictly decreases, which yields the desired contradiction. Finally, we have that in a decremental graph distances only increase and thus a ball of monotonically decreasing radius can only decrease over time. For the monotonicity property of open clusters, it suffices to use that they are defined as inverses of balls $B_G(u, A)$. 

The fact about the shortest paths follows from the fact that only the last vertex on such a shortest path $\pi_{G}(u,x)$ can be an affected vertex which can be seen by inspecting the line of reasoning above. Thus all edges on it were in $G^{(0)}$. 
\end{proof}

Given this claim it is now straightforward to calculate the time spent on maintaining pivots, balls and shortest paths.

\begin{claim}\label{clm:runtimeTZ}
With initial time $\Otil(mk^3)$ and worst-case time $\Otil(k^2)$ per update to $G$, the algorithm can maintain all pivots, balls, shortest paths and forest $F$ as described in \Cref{lma:workhorseSparsifier}.
\end{claim}
\begin{proof} 
From \Cref{thm:TZschemes}, we can compute the initial set $A$ in time $\tilde{O}(mk)$. Since $A$ is then maintained by adding affected vertices, and each affected vertex can be identified in time $O(1)$ per update, and each update adds at most $2$ new affected vertices to $A$, we can maintain $A$ in the claimed runtime.

To maintain the balls and shortest paths, note that given the set $A$, we can compute the initial balls and shortest paths required in time $\tilde{O}(k^3)$ per vertex $v \in V$, by running the following procedure: we initialize the set of explored vertices $\textsc{Explored}(v)$ to just contain the vertex $v$. Then, we iteratively search for the edge of minimum weight for every vertex $w \in \textsc{Explored}(v)$ that goes to a vertex $w'$ not in $\textsc{Explored}(v)$. We then add to $\textsc{Explored}(v)$ a vertex $z$ from the set 
\[\argmin_{w' \not\in \textsc{Explored}(v)} \min_{w \in \textsc{Explored}(v), (w,w') \in E} \dist_G(v,w) + l(w,w').
\]
The algorithm stops once it explores the first vertex $z$ that is not in $B_G(v, A)$. It then declares the last vertex added to be the initial pivot $p(v)$ of $v$. It is not hard to see that Dijktra's analysis yields that this computes a shortest path tree in $G$ rooted at vertex $v \in V$ that contains the shortest paths of all vertices in $B_G(v, A) \cup \{p(v)\}$. Using sorted-adjacency lists, it takes at most $O(k^2)$ time to find the next vertex to add, since $\textsc{Explored}(v)$ never exceeds $\tilde{O}(k)$ by \Cref{thm:TZschemes}, and thus finding each mimimum weight edge leaving the current set of vertices explored takes at most $O(k)$ time per vertex $w$ already in $\textsc{Explored}(v)$. The total time spent on this procedure is thus at most $\tilde{O}(nk^3)$. Since each relevant computed shortest path consists of at most $\tilde{O}(k)$ edges (since it is contained in $G[B_G(v, A) \cup \{p(v)\}]$), we can also output all shortest paths explicitly within the same time bound.

Finally, we observe that by \Cref{obs:decrementalMaintenanceofTZ}, at any time, any shortest path in the ball $B_G(v, A)$ and the shortest path from $v$ to $p(v)$, are already shortest paths in $G^{(0)}$ and since the balls $B_G(v,A)$ are monotonically decreasing over time, it suffices to remove shortest paths to vertices that are no longer in $B_G(v, A) \cup \{p(v)\}$ from the initial shortest path set that was outputted. From our previous analysis and the bound on the size of clusters from \Cref{thm:TZschemes} and  \Cref{obs:decrementalMaintenanceofTZ}, we further can output the paths (with all edges) that are no longer relevant shortest paths as defined in \Cref{lma:workhorseSparsifier}, in time $\Otil(k^2)$ per update.

Since we maintain the pivot paths explicitly and the forest $F$ described in \Cref{lma:workhorseSparsifier} is the union of pivot paths, maintaining $F$ is straightforward in the claimed time and congestion.
\end{proof}

\paragraph{Analyzing the LSSFs $F_0, F_1, \ldots, F_{\lambda-1}$.} Let us start the analysis by establishing some properties of the path collection $\projP$ that we embed into the forests $F_0, F_1, \ldots, F_{\lambda -1}$. 

\begin{claim}\label{clm:numberOfProjectedPathsPlusHopTotal}
We have that $\projP$ is of size at most $\tilde{O}(mk^3)$ and each path in $\projP$ consists of at most $6 \cdot b \log n + 3$ edges. It takes time $\tilde{O}(mk^4)$ to construct the set $\projP$.  
\end{claim}
\begin{proof}
Let us start by analyzing the size of the set $\mathcal{P}$. We have for every edge $(v,x) \in E$, and vertex $u$ with $v \in B(u, A)$ that $\mathcal{P}$ has a path $\pi_G(u,v) \concat (v,x)$. There are $m$ edges in $E$, and since each vertex $v$ has its cluster $C_G(v, A)$ of size at most $\tilde{O}(k)$ by \Cref{thm:TZschemes}, there are at most $\tilde{O}(mk)$ many such paths. 

But recall that $\projP$ has for any $ux$-path $P$ in $\mathcal{P}$, and any $\hat{u} \in B_G(u, A) \cup \{p(u)\}$ and $\hat{x} \in B_G(x, A) \cup \{p(x)\}$, the path $\pi_G(\hat{u}, u) \concat P \concat \pi_G(x, \hat{x})$ to $\projP$. Using the upper bound on the size of the balls $B_G(u, A)$ and $B_G(x, A)$ from \Cref{thm:TZschemes}, we can thus upper bound the number of paths in $\projP$ by $\tilde{O}(k^2 \cdot |\mathcal{P}|) = \tilde{O}(k^3 \cdot m)$.

The bound on the number of edges for each path $P \in \projP$ follows by bounding the number of edges for each segment using \Cref{thm:TZschemes}. Similarly, the runtime follows from \Cref{clm:runtimeTZ}, the size bound on $\projP$, the bound on the number of edges on each path in the collection, and the fact that each path is formed from a constant number of shortest paths already explicitly computed.
\end{proof}

Next, let us analyze the number of LSSFs required to embed the collection of paths $\projP$.

\begin{lemma}\label{lem:LSSFsCoverPaths}
The process of embedding the path collection $\projP$ stops after finding forests $F_0, F_1, \ldots, F_{\lambda-1}$ for $\lambda = O(\log m)$. After the process terminates, we have that for every $P \in \projP$ there exists an index $0 \leq i < \lambda$, such that $\wstr_i(P) \leq 2 \cdot \gamma_{LSST}$.
\end{lemma}
\begin{proof}
Our proof follows by showing that for every $0 \leq i < \lambda$, $\sum_{P \in \projP_{i+1}} l(P) \leq \frac{1}{2} \cdot \sum_{P \in \projP_{i}} l(P)$. The implication then follows since each edge length is polynomially upper-bounded in $m$, and thus $\sum_{P \in \projP_{0}} l(P) = m^c$ for some constant $c$, which implies that if there are $\lceil c \cdot \log_2(m) \rceil$ iterations of the algorithm, we have $\sum_{P \in \projP_{\lceil c \cdot \log_2(m) \rceil + 1}} l(P) < 1$ which by the lower bound on edge lengths implies that $\projP_{\lceil c \cdot \log_2(m) \rceil + 1} = \emptyset$. Thus, $\lambda \leq \lceil c \cdot \log_2(m) \rceil + 1 = O(\log m)$.

It remains to prove the claim. We prove by contradiction. Assume that the statement does not hold for some $i$. We have from the definition of $\projP_{i+1}$ that
\[
\sum_{ P \in \projP_{i+1}} \wstr_i(P) \cdot l(P) > 2 \cdot \gamma_{LSST} \cdot \sum_{P \in \projP_{i+1}} l(P).
\]
Using $\projP_{i+1} \subseteq  \projP_{i}$ on the LHS, and the assumption that $\sum_{P \in \projP_{i+1}} l(P) > \frac{1}{2} \cdot \sum_{P \in \projP_{i}} l(P)$ on the RHS, we thus derive
\[
\sum_{ P \in \projP_{i}} \wstr_i(P) \cdot l(P) >  \gamma_{LSST} \cdot \sum_{P \in \projP_{i}} l(P). 
\]
It remains to observe that since $w_i(e) = \econg(\projP_i, e) \cdot l(e)$, we have $\sum_{ P \in \projP_{i}} l(P) \cdot \wstr_i(P) = \sum_{e \in E} \econg(\projP_i, e) \cdot l(e) \cdot \wstr_i(e) = \sum_{e \in E} w_i(e) \cdot \wstr_i(e)$ by definition of $\wstr_i(P)$. Using the same line of reasoning, we obtain that $\sum_{P \in \projP_{i}} l(P) = \sum_{e \in E} w_i(e)$. Combining these inequalities yields $\sum_{e \in E} w_i(e) \cdot \wstr_i(e) > \gamma_{LSST} \cdot  \sum_{e \in E} w_i(e)$. But this gives the desired contradiction as it violates the guarantee given by Property \Cref{item:avgstretchboundLem} in \Cref{lemma:globalstretch} that $\sum_{e \in E} w_i(e) \cdot \wstr_i(e) \leq \gamma_{LSST} \cdot \sum_{e \in E} w_i(e)$, as desired.
\end{proof}

\paragraph{Analyzing Vertex Sparsifiers $\tilde{G}$.} In this section, we prove various properties on the vertex sparsifier $\tilde{G}$. We start by proving the pivotal Lemma of this section: that distances between vertices in $A$ are preserved by $\hat{G}$. 

\begin{lemma}\label{lma:vsLowerBoundStretch}
At any time, for any two vertices $u,v\in A$, let $\hat{u}$ and $\hat{v}$ be any vertices in $V(\tilde{G})$ such that $\hat{u}$ is identified with $u$ and $\hat{v}$ is identified with $v$, we have $\dist_{G}(u,v) \leq \dist_{\hat{G}}(\hat{u}, \hat{v})$.
\end{lemma}
\begin{proof}
It is not hard to see that in lieu of proving this statement for $\tilde{G}$, it suffices to establish that for any two vertices $u,v \in A$, $\dist_G(u,v) \leq \dist_{\hat{G}}(u,v)$.

Consider any shortest path $\pi_{\hat{G}}(u,v)$ for $u,v \in A \subseteq V(G)$. Consider now the following mapping procedure of $\pi_{\hat{G}}(u,v)$ to a path $P$ in $G$: for every edge $\hat{e} \in \pi_{\hat{G}}(u,v)$, where $\hat{e}$ originates from the core graph $\cC(G, F_i, \wstr_i)$ and has pre-image $e = (x,y)$ in $G$, we replace the edge $\hat{e}$ on $\pi_{\hat{G}}(u,v)$ by the path segment $F_i[\root^{F_i}(x), x] \concat (x,y) \concat F_i[y, \root^{F_i}(y)]$. By \Cref{def:coregraph}, we have that $\hat{e} = (\root^{F_i}(x), \root^{F_i}(y))$. Thus, the mapping yields a proper (although not necessarily simple) $uv$-path $P$ in $G$, as required.

Finally, by \Cref{def:spanningforest} and \Cref{def:coregraph}, we have that for each edge $\hat{e}$ in $\cC(G, F_i, \wstr_i)$  with pre-image $e$ in $G$ where $\hat{e} = (\root^{F_i}(x), \root^{F_i}(y))$, the path segment $F_i[\root^{F_i}(x), x] \concat (x,y) \concat F_i[y, \root^{F_i}(y)]$ has length at most $l_{\hat{G}}(\hat{e})$ with respect to the length function $l$ of $G$. Thus, $\dist_G(u,v) \leq l_G(P) \leq l_{\hat{G}}(\pi_{\hat{G}}(u,v)) = \dist_{\hat{G}}(u,v)$. 
\end{proof}

\begin{lemma}\label{lma:vsUpperBoundStretch}
At any time, for any two vertices $u,v\in A$, let $\hat{u}$ and $\hat{v}$ be any vertices in $V(\tilde{G})$ such that $\hat{u}$ is identified with $u$ and $\hat{v}$ is identified with $v$, we have $\dist_{\tilde{G}}(\hat{u},\hat{v}) \leq 10 \cdot \gamma_{LSSF} \cdot \dist_G(u,v)$.
\end{lemma}
\begin{proof}
For our analysis, we consider $\pi_G(u,v)$ which is the shortest path in $G$ from $u$ to $v$. Again, the proof follows if we can show that we can upper bound the distance between $u$ and $v$ in $\hat{G}$. 

\underline{Reducing to the case $\pi_G(u,v) \subseteq G^{(0)}$:} We first show that we can assume wlog that $\pi_G(u,v)$ is in the initial graph $G^{(0)}$. That is since all endpoints of edges not in $G^{(0)}$ are in the set $A$, and thus we can segment each general path between vertices in $A$ into paths with no such new edge between vertices in $A$ and new edges where new edges are present in every core graph $\mathcal{C}(G, F_i, \wstr_i)$ with stretch overestimate equal $1$, and thus it appears in $\hat{G}$ with the same length as in $G$, which yields the Lemma.

\underline{Segmenting $\pi_G(u,v)$:} Let us now show the Lemma conditioned on the fact that each edge on $\pi_G(u,v)$ is already in $G^{(0)}$. We define a sequence of vertices $b_0, b_1, \ldots, b_{\tau}$ as follows: we let $b_0 = u$, and for any $0 \leq j < \tau$, define $b_{j+1}$ to be the first vertex on $\pi_G(u,v)$ after vertex $b_j$ that is outside the ball $B_G(b_j, A)$. We let $b_{\tau}$ be the first vertex such that $b_{\tau} = v$. Note that by the subpath property of shortest paths, we have $\pi_G(b_j, b_{j+1}) = \pi_G(u,v)[b_j, b_{j+1}]$. We next show that we have 
\begin{equation}\label{eq:segmentUpperBound}
    \dist_H(p(b_j), p(b_{j+1})) \leq 10 \cdot \gamma_{LSSF} \cdot \dist_G(b_j, b_{j+1})
\end{equation}
which yields the final Lemma, which can be obtained by summing over the segments and observing $p(b_0) = u, p(b_{\tau}) = v$.

\underline{Establishing \eqref{eq:segmentUpperBound}:} To establish the claimed inequality for every $0 \leq j < \tau$, we first observe that by choice of $b_{j+1}$ and \Cref{obs:decrementalMaintenanceofTZ}, we have $\pi_G(b_j, b_{j+1}) = \pi_{G^{(0)}}(b_j, b_{j+1})$ and thus $\pi_G(b_j, b_{j+1}) \in \mathcal{P}$. Since the current pivots of $b_j$ and $b_{j+1}$ are vertices in the initial open balls of $b_j$ and $b_{j+1}$ (or the initial pivot itself), we have again from \Cref{obs:decrementalMaintenanceofTZ} that $P = \pi_G(p(b_j), b_j) \concat \pi_G(b_j, b_{j+1}) \concat \pi(b_{j+1}, p(b_{j+1}))$ is in the set $\projP$. Note that by definition, we have that 
\begin{align} \label{eq:upperBoundPathLength}
\begin{split}
l(P) &= \dist_G(b_j, p(b_j)) + \dist_G(b_j, b_{j+1}) + \dist_G(b_{j+1}, p(b_{j+1}))\\
    & \leq 2 \cdot (\dist_G(b_j, p(b_j)) + \dist_G(b_j, b_{j+1}))\\
    & \leq 4 \cdot \dist_G(b_j, b_{j+1})
\end{split}
\end{align}
where the second inequality is obtained from the triangle inequality and the definition of the pivot function, and the third from the fact that $b_{j+1}$ is not contained in the ball $B_G(b_j, p(b_j))$. 

From \Cref{lem:LSSFsCoverPaths}, we have that there is an index $0 \leq i < \lambda$ such that $\wstr_i(P) \cdot l(P) \leq 2 \cdot \gamma_{LSSF} \cdot l(P) \leq 8 \cdot \gamma_{LSSF} \cdot \dist_G(b_j, b_{j+1})$ where we use in the last inequality our derivation in \eqref{eq:upperBoundPathLength}. 

Next, let $c_0, c_1, c_2, \ldots, c_{\kappa}$ be defined such that $c_0 = p(b_j)$, and for every $0 \leq \ell < \kappa$, we have that $c_{\ell+1}$ is the first vertex on $P$ after vertex $c_{\ell}$ that is in a different connected component in $F_i$ than $c_{\ell}$. If such a vertex does not exist, we terminate and set $\kappa = \ell$. Letting for every $0 < \ell \leq \kappa$, $e_{\ell} = (x_{\ell}, c_{\ell})$ be the incoming edge to $c_{\ell}$ on $P$. Then, we have by \Cref{def:coregraph}, that the image $\hat{e}_{\ell}$ of $e_{\ell}$ in $\cC(G, F_i, \wstr_i)$ has $l_{\cC(G, F_i, \wstr_i)}(\hat{e})\defeq \wstr_i(e) \cdot l(e)$ and clearly, the edges $\hat{e}_1, \hat{e}_2, \ldots, \hat{e}_{\kappa}$ form a path $P'$ between the vertices in $\cC(G, F_i, \wstr_i)$ that correspond to the connected components in $F_i$ that contain $c_0$ and $c_{\kappa}$, respectively. The length of this path in $\cC(G, F_i, \wstr_i)$, and thus in $\hat{G}$, is 
\[
l(P') = \sum_{\ell = 1}^{\kappa} l_{\cC(G, F_i, \wstr_i)}(\hat{e}_{\ell}) = \sum_{\ell = 1}^{\kappa} \wstr_i(e_{\ell}) \cdot l(e_{\ell}) \leq \sum_{e \in P} \wstr_i(e_{\ell}) \cdot l(e_{\ell}) = \wstr_i(P) \cdot l(P).
\]
Finally, we observe that by Property \ref{item:affectedVerticesBecomeRoots} from \Cref{lemma:globalstretch} and the fact that we initialize each of these data structures with set $A^{(0)}$, and $A$ then evolves by adding the endpoints of edges inserted and deleted to $A$, it is ensured that every vertex in $A$ is a root in the forest $F_i$. Since $p(b_j), p(b_{j+1}) \in A$ (the map $p$ has its image in $A$), we thus have that both of these vertices are the roots of the connected components that contain them. We thus conclude that $\hat{G}$ contains a path from $p(b_j)$ to $ p(b_{j+1})$ of total length $2 + \wstr_i(P) \cdot l(P) < 10 \cdot \gamma_{LSSF} \cdot \dist_G(b_j, b_{j+1})$, as desired. 
\end{proof}

\begin{claim}\label{clm:edgeLengths}
Every edge length in $\tilde{G}$ is in $\{0\} \cup [1, nL]$.
\end{claim}
\begin{proof}
By \Cref{def:coregraph} and \Cref{def:spanningforest}, it is immediate that every edge length in $\hat{G}$ is in $[1, nL]$ since each edge length corresponds to the length of a path in $G$, and since each such path consists of at most $n$ vertices, the maximum length of any such path is upper bounded by $nL$. But since $\hat{G}$ differs from $\tilde{G}$ exactly by the fact all $0$ length edges are contracted, the claim follows.
\end{proof}

Next, we bound the recourse of $\tilde{G}$ and the running time.

To this end, we henceforth denote by $U_{F_i}^{(0)}, U_{F_i}^{(1)}, \ldots$ the update batches where $U_{F_i}^{(t)}$ consists of all edge deletions to $F_i$ that were issued by \Cref{algo:updateSparsifier} to the data structure $\mathcal{D}_i$ during the processing of the $t$-th update to $G$. That is, the forest $F_i$ at the time $t$ is obtained from applying all updates in  $U_{F_i}^{(0)}, U_{F_i}^{(1)}, \ldots, U_{F_i}^{(t)}$ to the initial forest $F_i$. We denote the initial forest $F_i$ by $F_i^{(-1)}$.

We denote by $U_{\tilde{G}}^{(0)}, U_{\tilde{G}}^{(1)}, \ldots$ the update batches where the batch $ U_{\tilde{G}}^{(t)}$ consists of all updates to $\tilde{G}$ while processing the $t$-th update to $G$. We again denote by $\tilde{G}^{(-1)}$ the initial graph $\tilde{G}$. We stress that the update batches $ U_{\tilde{G}}^{(t)}$ not only consist of edge insertion and deletions and insertions of isolated vertices, but also, of vertex splits.

\begin{claim}\label{clm:recGTilde}
The graph $\tilde{G}$ is initialized in time $\Otil(m)$ and consists of at most $\Otil(m)$ edges and $\Otil(m/k)$ vertices. The update batches $U_{\tilde{G}}^{(0)}, U_{\tilde{G}}^{(1)}, \ldots$ are such that for any time $t \geq 0$, $|U_{\tilde{G}}^{(t)}| \leq \Otil(1 + \sum_{0 \leq i < \lambda} |U_{F_i}^{(t)}|)$ and the update batch $U_{\tilde{G}}^{(t)}$ can be computed in time $\Otil(|U_{\tilde{G}}^{(t)}| \cdot k \cdot \Delta)$.
\end{claim}
\begin{proof}
We recall that $\tilde{G}$ is maintained as the direct sum of core graphs over the forests $F_0, F_1, \ldots, F_{\lambda-1}$ plus $0$ length edges between any pair of vertices in the resulting graph that are identified with the same vertex in $G$. 

For every $0 \leq i < \lambda$, given an initial forest $F_i$, it is straightforward using \Cref{def:coregraph} to create the core graph $\cC(G, F_i, \wstr_i)$ in time $\Otil(m)$ using a dynamic tree data structure to maintain the trees of forest $F_i$. Adding the $0$ length edges between vertices in different core graphs can be implemented straightforwardly in $\Otil(m)$ time. The number of edges is immediate from this analysis and the number of vertices in each core graph is at most $O(m/k)$ by \Cref{item:cccountLem}. Since $\tilde{G}$ is the direct sum of $O(\log m)$ core graphs, the bound on the vertices follows.

Next, let us analyze an update to $G$. Such an update can cause for any $0 \leq i < \lambda$ that $F_i$ further undergoes an additional number of at most $O(\log^2 m)$ deletions and at most one isolated vertex insertion to $F_i$ by \Cref{lemma:globalstretch}. Since $F_i$ has a monotonically decreasing edge set, each such update to $F_i$ can be encoded as a single update to $\tilde{G}$ in the form of the insertion of an isolated vertex or a vertex split. Additionally, if $G$ undergoes an edge insertion or deletion, we also have to insert or delete the corresponding projected edge from $\cC(G, F_i, \wstr_i)$. Finally, whenever a new root is added to the core graph $\cC(G, F_i, \wstr_i)$, up to $\lambda - 1$ new $0$ length edges are added to $\tilde{G}$. Thus the number of updates is bounded for every $0 \leq i < \lambda$, and our analysis of the recourse to $\tilde{G}$ thus follows from $\lambda = O(\log m)$.

The update time follows immediately from \Cref{lemma:globalstretch} and our discussion above.
\end{proof}

\paragraph{Analyzing the Final Sparsifier $H$.} Let us now analyze the final sparsifier $H$. We start by proving that the initialization procedure given in \Cref{algo:initSparsifier} is efficient and produces a graph $H$ with few vertices and edges.

\begin{claim}\label{clm:initSparisfier}
Choosing $\gamma_{degConstr}$ to be a reasonably large value such that $\gamma_{degConstr} = \Otil(|V(\tilde{H})|^{1/K} \cdot \gamma_{\ell})$, we have that \Cref{algo:initSparsifier} can be implemented in time $\Otil(m \gamma_{\ell} + mk^4+m\Delta)$, and that after the algorithm terminates, graph $\tilde{H}$ has at most $\Otil(m \gamma_{\ell} / k)$ vertices and edges.
\end{claim}
\begin{proof}
We start by analyzing the runtime of the algorithm, and the number of edges and vertices of $\tilde{H}$ just up until the first iteration of the while-loop starting in \Cref{lne:whileLoopDegreeConstraint} in \Cref{algo:initSparsifier}.

By \Cref{clm:numberOfProjectedPathsPlusHopTotal}, we have that the collection of paths $\projP$ takes time $\Otil(mk^4)$ to compute. Before starting the data structure from \Cref{lemma:globalstretch} to maintain forests $F_i$, we have to construct $\projP_i$ which can again be achieved in time $\Otil(mk^4)$ by explicitly checking each path in $\projP$. This allows us to compute the weight function $w_i$ within the same time bound. Initializing the data structure \Cref{lemma:globalstretch} further takes time $\Otil(m)$ for each forest $F_i$. Since $\lambda = O(\log m)$ by \Cref{lem:LSSFsCoverPaths}, we can thus upper bound the total time to initialize all forests $F_0, F_1, \ldots, F_{\lambda- 1}$ by $\Otil(mk^4)$. It is not hard to see from \Cref{def:coregraph} and \Cref{def:spanningforest} that given these forests, the graph $\tilde{G}$ can be initialized in time $\Otil(m)$. Finally, by \Cref{thm:spanner}, the time to initialize $\tilde{H}$ is $\Otil(m \gamma_{\ell})$, as desired.

From \Cref{lemma:globalstretch}, \Cref{def:coregraph} and \Cref{def:spanningforest}, and \Cref{thm:spanner}, we can further bound the number of initial edges and vertices in $\tilde{H}$ by $\Otil(m \gamma_{\ell}/k)$.

We next analyze the while-loop iterations of the while-loop starting in \Cref{lne:whileLoopDegreeConstraint} in \Cref{algo:initSparsifier}. To this end, we define the following potential function for every $v \in V(\tilde{H})$ by 
\[\Phi(v) \defeq \sum_{0 \leq i < \lambda} \max\{\deg_{\cC(G, F_i, \wstr_i) \cap \tilde{H}}(v) - \Delta \cdot \gamma_{degConstr} , 0\}
\]
and define the global potential $\Phi(\tilde{H}) \defeq \sum_{v \in V(\tilde{H})} \Phi(v)$. It is not hard to see that initially $\Phi(\tilde{H}) \leq 2 \cdot |E(\tilde{H})| = \Otil(m\gamma_{\ell}/k)$. 

Next, observe that whenever the while-loop is run, and in the while-loop statement the procedure $\textsc{ReduceDegree}(v, i, \gamma_{degConstr})$ given in \Cref{algo:reduceDegree}, it computes a set of edges $E'$ to be deleted from $T_i$ being the tree rooted at $v$ in forest $F_i$ via the data structure $\mathcal{D}_i$. Observe that the guarantees of \Cref{thm:partitionFrederickson} imply that the degree of $v$ w.r.t. $\cC(G, F_i, \wstr_i) \cap \tilde{H}$ when entering the call to $\textsc{ReduceDegree}(v, i, \gamma_{degConstr})$ is at least $|E'| \cdot \gamma_{degConstr}$ since all but one component in $T_i \setminus E'$ are incident to at least $z$ edges in $E_i$ and the number of components is $|E'| + 1$. Similarly, \Cref{thm:partitionFrederickson} implies that the runtime required by the subroutine is at most $O(|E'| \gamma_{degConstr} \Delta)$. 

Consider now the effect of not feeding $E'$ into $\mathcal{D}_i$, but instead of directly deleting edges $E'$ from $F_i$, and then splitting the corresponding vertices in $\tilde{G}$ and $\tilde{H}$. This would directly yield a decrease of $\Phi(\tilde{H})$ by at least $\frac{|E'| \cdot \gamma_{degConstr}}{2}$ since by the while-loop condition we have that at least half the edge originally incident to $v$ in  $\cC(G, F_i, \wstr_i) \cap \tilde{H}$ contribute to $\Phi(\tilde{H})$, and since none of the at least $|E'| \cdot \gamma_{degConstr}$ edges originally incident to $v$ contributes to $\Phi(\tilde{H})$ after the set $E'$ is removed from $F_i$ and the vertex splits applied to $\tilde{H}$.

In fact, since $F_i$ has decreasing edge set under updates forwarded to \Cref{lemma:globalstretch} and ensures that edges $E'$ are deleted when forwarded, we have that if we implement the vertex splits suggested by $\mathcal{D}_i$ directly to $\tilde{G}$ and $\tilde{H}$ that the same drop in potential occurs since it only further splits the graphs. 

However, when using the data structures from \Cref{lemma:globalstretch} and \Cref{thm:spanner} to process the resulting vertex splits from $\mathcal{D}_i$, while it forwards these vertex splits directly to $\tilde{G}$ and $\tilde{H}$, it additionally updates $\tilde{H}$ via up to $\gamma \cdot |E'|$ for some $\gamma = \Otil(|V(\tilde{H})|^{1/K} \gamma_{\ell})$ updates to $\tilde{H}$ edge insertions/ deletions. But these edge insertions/deletions can increase the potential $\Phi(\tilde{H})$ by at most $\gamma \cdot |E'|$. Thus, choosing $\gamma_{degConstr} = 2 \gamma$, we ensure that the potential $\Phi(\tilde{H})$ drops by at least $\frac{1}{2} \gamma_{degContr} \cdot |E'|$. 

It is now straightforward to see that the total number of edges that are fed to any of the data structure $\mathcal{D}_i$ to be removed from some forest $F_i$ is at most $\Otil(m / (k \cdot \gamma_{degContr}))$. This bounds the runtime required by all calls to the procedure from \Cref{thm:partitionFrederickson} by $\Otil(m \Delta /k)$. It is not hard to verify that this in fact dominates asymptotically the runtime of the while-loop. It further bounds the number of vertices and edges in the graph $\tilde{H}$ obtained at the end of \Cref{algo:initSparsifier} by $\Otil(m \gamma_{\ell}/ k + m \gamma/ (k \gamma_{degContr})) = \Otil(m \gamma_{\ell}/ k)$.

Finally, we observe that $H$ can be obtained in time $O(|E(\tilde{H})|)$ and that the number of edges is equal to the number of edges in $\tilde{H}$ and the number of vertices at most the number of vertices in $\tilde{H}$ since $H$ is obtained via a set of contractions applied to $\tilde{H}$.
\end{proof}

We next prove that the maximum degree of $\tilde{H}$ is bounded at all times.

\begin{claim}\label{clm:degreeUpperBoundH}
At any time, $\tilde{H}$ has maximum degree at most $\Otil(\Delta \cdot m^{2/K} \cdot (\gamma_{\ell})^2)$. For any time $t \geq 1$, we apply $\Otil(\Delta \cdot m^{2/K} \cdot (\gamma_{\ell})^2)$ changes to the forests $F_0, F_1, \ldots, F_{\lambda-1}$ and the update batch $|U_H^{(t)}| = \Otil(\Delta \cdot m^{2/K} \cdot (\gamma_{\ell})^2)$.
\end{claim}
\begin{proof}
Since there are only $\lambda = \Otil(1)$ core graphs, it suffices to show that the maximum degree $\deg_{\tilde{H} \cap  \cC(G, F_i, \wstr_i)}(v)$ is bounded by $\Otil(\Delta \cdot m^{1/K} \cdot \gamma_{\ell})$ for every $v$ and $i$ to prove the claim.

Clearly, the maximum degree $\deg_{\tilde{H} \cap  \cC(G, F_i, \wstr_i)}(v)$ is bounded by $\delta = 2 \cdot \gamma_{degConstr} \cdot \Delta$ after \Cref{algo:initSparsifier} terminates since otherwise the while-loop starting in \Cref{lne:whileLoopDegreeConstraint} of the initialization algorithm would not have terminated.

Note that thereafter, forwarding a single update to $G$ to the data structures maintaining the forests, $\tilde{G}$ and $\tilde{H}$ causes at most $\gamma = \Otil(m^{1/K} \gamma_{\ell})$ recourse to $\tilde{H}$, and since we have by induction on time that $\tilde{H}$ had maximum degree $\Otil(\Delta \cdot m^{1/K} \cdot \gamma_{\ell})$, this then implies that when the vertex $v$ and index $i$ are picked in \Cref{lne:pickHighestDegreeVertex}, the procedure in \Cref{algo:reduceDegree} finds at most $\Otil\left(\frac{\Delta \cdot m^{1/K} \cdot \gamma_{\ell} +  \gamma}{\gamma_{degConstr}}\right)$ edges $E'$ to delete from $F_i$. But since these edge deletions and the update to $G$ can cause at most $(1+ |E'|) \gamma$ recourse to $\tilde{H}$, we have that the total recourse per update to $G$ is bounded by $\eta = \Otil(\Delta \cdot m^{2/K} \cdot (\gamma_{\ell})^2)$.

To bound the maximum degree of any vertex w.r.t. to graph $\tilde{H} \cap  \cC(G, F_i, \wstr_i)$, we next consider the following two-player game. 

\begin{definition}[Pile Splitting Game, see \cite{levcopoulos1988balanced}]
Given a row of piles $S_0, S_1, \ldots, S_k$. Initially, each pile contains at most $\delta$ stones. Then, consider the following two-player game where players take turns:
\begin{itemize}
    \item at the beginning of every round, the \emph{stone placer} adds up to $\eta$ new stones on the various different piles and deletes an arbitrary number of stones from the piles, 
    \item then, the \emph{pile splitter} takes the pile $S_i$ that contains the most stones and splits $S_i$ into the $S_{k+1}, S_{k+2}, \ldots, S_{k + k'}$ for some $k'$ and $S_i$ such that after the splitting each such pile contains at most half of the stones originally placed on $S_i$.
\end{itemize}
We call this game the \emph{Pile Splitting Game}.
\end{definition}
\begin{theorem}[see \cite{levcopoulos1988balanced}]\label{thm:pileSplittingThm}
The maximum number of stones on any pile throughout the entire \emph{Pile Splitting Game} is at most $O(\delta + \eta \log m)$ where $m$ is the total number of stones placed by the end of the game.
\end{theorem}

Now, let there be a pile $S_j$ for every $0 \leq i < \lambda$ and $v \in \cC(G, F_i, \wstr_i)$ and let the initial degree $\deg_{\tilde{H} \cap  \cC(G, F_i, \wstr_i)}(v)$ be the number of stones on $S_j$ upon initialization. We have from our previous analysis that every pile $S_j$ has at most $\delta$ stones. Further, the recourse to $G$ can then be mapped to the game as follows: for every edge insertion/ deletion, we add/ remove a stone to the endpoints of the edge in $H$ such that the number of stones corresponds to the updated degree for every pile/ vertex. An isolated vertex insertion to $G$ adds the corresponding vertex that is added to $\cC(G, F_i, \wstr_i)$ to a new empty pile. Note that per update to $G$, we add at most $2\eta$ stones by our analysis above.

We then map vertex splits of the current update phase to $G$ to the updates of the pile splitter where we use that \Cref{algo:updateSparsifier} invokes \Cref{algo:reduceDegree} on the vertex with maximum degree which corresponds to the highest pile. 

We note that there is only a small mismatch: if the maximum degree $\deg_{\tilde{H} \cap  \cC(G, F_i, \wstr_i)}(v)$ is at most $2 \cdot \gamma_{degConstr} \cdot \Delta$, the procedure from  \Cref{algo:reduceDegree} does not necessarily split the vertex in $\tilde{H}$ such that it has half its degree only incident to each newly split vertex. However, in this case, we have that all degrees are at most $\delta$ and we can thus restart the \emph{Pile Splitting Game}. 

This yields that for any vertex $v \in \cC(G, F_i, \wstr_i)$, we have that the degree $\deg_{\tilde{H} \cap  \cC(G, F_i, \wstr_i)}(v)$ is bounded by $\Otil(\delta + \eta \log m)$ by \Cref{thm:pileSplittingThm} which yields the claim.
\end{proof}

We can now prove \Cref{lma:workhorseSparsifier}. 

\begin{proof}[Proof of \Cref{lma:workhorseSparsifier}.]
We have from \Cref{clm:runtimeTZ} that the first 4 Properties of \Cref{lma:workhorseSparsifier}. It remains to establish  Property \ref{prop:workhorseDistancePreserve} and to give the runtime analysis. We start by doing the former.

We have from \Cref{clm:initSparisfier} that $H^{(0)}$ has at most  $\Otil(m \gamma_{\ell} / k)$ vertices and edges. We further have for every $t \geq 1$ that 
 \begin{enumerate}
        \item \underline{$|U_H^{(t)}| \leq \gamma_{recVS}$:} the bound follows from the fact that by \Cref{item:degboundLem} the number of updates at each time to $\tilde{H}$ is  $\Otil(\Delta \cdot m^{2/K} \cdot (\gamma_{\ell})^2)$ and that the maximum degree in $\tilde{H}$ is at most $\Otil(\Delta \cdot m^{2/K} \cdot (\gamma_{\ell})^2)$. But this yields that emulating the vertex splits to $\tilde{H}$ as described, which multiplies the number of updates by at most a multiple of the maximum degree in $\tilde{H}$ yields that the number of updates in $U_H^{(t)}$ is at most $\gamma_{recVS}$, since we choose $\gamma_{recVS}$ to be a reasonably large value in $\Otil(\Delta^2 \cdot m^{4/K} \cdot (\gamma_{\ell})^4)$.
        
        \item \underline{the maximum vertex degree of $H$ is at most $\gamma_{recVS} \cdot \Delta$:} this follows trivially from the fact that for each vertex $v$ in $G$, there are at most $\lambda = O(\log m)$ vertices in $\tilde{H}$ that are identified with $v$. Thus, the contraction procedure applied to $\tilde{H}$ to obtain $H$ increases the maximum degree by at most $\lambda$, and the claim now follows directly from  \Cref{item:degboundLem}.
        \item \underline{$H$ has lengths in $[1, nL]$:} this follows immediately from \Cref{clm:edgeLengths} and the fact that $0$ length edges in $\tilde{H}$ are contracted as they only exist between vertices in $\tilde{H}$ that identify with the same vertex in $G$.
        \item \underline{for every two vertices $u,v \in A$, we have $\dist_G(u,v) \leq \dist_{H}(u,v) \leq (\gamma_{\ell})^{O(K)} \cdot \dist_G(u,v)$:} the claim follows immediately from chaining \Cref{lma:vsLowerBoundStretch} and \Cref{lma:vsUpperBoundStretch} with the guarantees in \Cref{thm:spanner} and the fact that all vertices contracted in $\tilde{H}$ into a super-vertex in $H$ were already connected by a length $0$ path as can be seen from the guarantees of \Cref{thm:spanner}.
        \item \underline{for any edge $e = (u,v) \in H$, a $uv$-path $P$ in $G$ with $l_G(P) \leq l_H(e)$ can be outputted in time $O(|P|)$:} Every such edge corresponds to a projected edge in one of the core graphs, and the path between the endpoints $u$ and $v$ that are roots in this core graph can be returned straightforwardly.
\end{enumerate}

It remains to carry out the runtime analysis. Again, we use \Cref{clm:runtimeTZ} to upper bound the runtime related to maintaining pivots, balls, shortest paths and forest $F$. From \Cref{clm:initSparisfier}, we have that the rest of the data structure and graph $H$ can be initialized in the claimed time.

For the update time, we observe that from \Cref{clm:degreeUpperBoundH}, that the update to $H$ is obtained by first updating the data structures $\mathcal{D}_i$ with a total of $\Otil(\Delta \cdot m^{2/K} \cdot (\gamma_{\ell})^2)$ many updates. By \Cref{lemma:globalstretch}, this results in update time spend on updating these data structures of $\Otil(\Delta \cdot m^{2/K} \cdot (\gamma_{\ell})^2)$ and causes up to $\Otil(\Delta \cdot m^{2/K} \cdot (\gamma_{\ell})^2)$ updates to $\tilde{G}$. This in turn causes the algorithm from \Cref{thm:spanner} to process updates to $\tilde{G}$ and output $\tilde{H}$ to spend total time $\Otil(\Delta \cdot m^{2/K} \cdot (\gamma_{\ell})^{O(K^2)} m^{1/K} k^2)$ where we use that by the properties enforced by \Cref{lemma:globalstretch}, the maximum degree of any vertex in $\tilde{G}$ is $\Otil(k)$. This latter time bound subsumes the time spent in \Cref{algo:updateSparsifier} to compute the highest degree vertex and to find the edge set $E'$. Finally, the time to compute the update batch to update $H$ appropriately given $\tilde{H}$ is $\Otil(|U_H^{(t)}|)$. 
\end{proof}

\subsection{Mapping a Hierarchical Forest through the Vertex Sparsifier}
\label{subsec:mappingHierachVS}

In this section, we prove that we can extend \Cref{lma:workhorseSparsifier} such that given a hierarchical forest $F$ on $H$, we can map the tree to $G$ with similar guarantees. The precise statement is summarized below.

\begin{lemma}[Extension of \Cref{lma:workhorseSparsifier}]\label{lma:extensionOfMainThm}
Given inputs as in \Cref{lma:workhorseSparsifier} and let vertex set $A \subseteq V(G)$ and graph $H$ be maintained by the data structure from \Cref{lma:workhorseSparsifier}. 

Further, given a dynamic flat hierarchical forest $F$ over a monotonically increasing set $A_F \subseteq V(H)$ in graph $H$ along with vertex maps $\Pi_{A_F \mapsto V(F)}, \Pi_{V(F) \mapsto V(H)}$ and flat graph embedding $\Pi_{F \mapsto H}$, along with parameters $\gamma_{congRep}$ and $\gamma_{recRep}$ such that at any time  $\econg(\Pi_{F \mapsto H})$ is  bounded by $\gamma_{congRep}$ and the number of changes to $F$ caused by an update to $G$ is upper bounded by $\gamma_{recRep}$. We require the vertex maps to be such that whenever a vertex is added to the pre-image, its image remains constant for the rest of the algorithm.

Then, the algorithm can maintain a flat hierarchical forest $F'$ over set $A_F$ in graph $G$ along with vertex maps $\Pi_{A_F \mapsto V(F')}, \Pi_{V(F') \mapsto V(G)}$ and flat graph embedding $\Pi_{F' \mapsto G}$ such that at any time $\econg(\Pi_{F' \mapsto G})$ is bounded by $\Otil(\gamma_{congRep} \cdot \Delta \cdot \gamma_{recVS})$ and the number of changes to $F'$ per update to $G$ is $\Otil(\gamma_{recRep} + \gamma_{congRep} \cdot \gamma_{recVS} \cdot \Delta)$, and we have for any two vertices $u,v \in A_F$ that $l_G(\Pi_{F' \mapsto G}(\pi_{F'}(\Pi_{A_F \mapsto V(F')}(u),\Pi_{A_F \mapsto V(F')}(v)))) \leq l_H(\Pi_{F \mapsto H}(\pi_{F}(\Pi_{A_F \mapsto V(F)}(u),\Pi_{A_F \mapsto V(F)}(v))))$. Further, we have that the vertex maps are such that whenever a vertex is added to the pre-image, its image remains constant for the rest of the algorithm.

Having $A_F, F$ and the vertex maps and graph embedding associated with $F$ maintained, and $H$ as maintained by the data structure from \Cref{lma:workhorseSparsifier}, the algorithm to maintain $F'$ and the vertex maps and graph embedding associated with $F'$ requires additional initialization time $\tilde{O}(m \cdot \gamma_{congRep})$ and processes every update with additional worst-case time $\Otil(\gamma_{recRep} + \gamma_{congRep} \cdot \gamma_{recVS} \cdot \Delta)$.
\end{lemma}

\paragraph{The Algorithm.} For our algorithm, we use the following graphs. 

\begin{definition}[Direct Sum of Core Graph Trees]\label{def:coreGraphDirectSum}
Given a vertex $v \in V(H)$, we define $T_v$ to be the graph obtained by contracting the root vertices of all trees $T$ in some forest $F_i$ rooted at $v$ for some $0 \leq i < \lambda$. We denote by $(u,i)$ the vertex in $T_v$ associated with the vertex $u \in V(G)$ in the tree $T$ rooted at $v$ in $F_i$ (assuming $u \in T$).
\end{definition}

We give in \Cref{alg:mapHierForest} the procedure to initialize the forest $F'$ described in \Cref{lma:extensionOfMainThm}. The algorithm first initializes $F'$ to have $\gamma_{congRep}' \defeq \gamma_{congRep} \cdot \Delta \cdot \gamma_{recVS}$ copies $T_{v,j}$ of tree $T_v$ as defined above for every $v$. It then constructs a map $\Lambda$ from the vertices in $V(F)$ to integers in $[0, \gamma_{congRep}')$. This map can be seen as an injective map that maps each node $x$ in $V(F)$ to trees $T_{v,j}$ where $v$ is the vertex in $V(H) \subseteq V(G)$ identified with $x$ such that no two nodes are mapped to the same tree, i.e. it assigns each node $x \in V(F)$ a distinct tree $T_{v,j}$. It is easy to see that it suffices to map to $\gamma_{congRep}'$ many numbers by the bound on the edge congestion of graph embedding $\Pi_{F \mapsto H}$ which multiplied by the maximum vertex degree in $H$ (which is $\Delta \cdot \gamma_{recVS}$ by \Cref{lma:workhorseSparsifier}) yields the vertex congestion of $\Pi_{F \mapsto H}$ which in turn upper bounds the number of nodes in $V(F)$ identified with a single vertex in $V(H) \subseteq V(G)$.

The algorithm then maps the edges in $F$ into $F'$ as follows: each edge $\hat{e}'$ in $F$ that is mapped to $\hat{e} = (x,y)$ in $H$ via $\Pi_{F \mapsto H}$ (recall the graph embedding is flat) is then mapped into $F'$ by mapping its endpoints into the trees $T_{\Pi_{V(F) \mapsto V(H)}(x),j}$ and $T_{\Pi_{V(F) \mapsto V(H)}(y), j'}$ for $j =  \Lambda(x)$ and $j' =  \Lambda(y)$. 

The precise mapping of $\hat{e}'$ first checks from which core graph $\hat{e}$ originates. We denote this core graph by $\cC(G, F_i, \wstr_i)$. Then, the algorithm finds the pre-image $e = (u,v) \in E(G)$ of $\hat{e} \in E(H)$. Finally, it maps the endpoint of $\hat{e}'$ identified with $x \in V(H)$ to the vertex $(u, j)$ in the copy $T_{x,j}$ of $T_x$. It maps the other endpoint analogously. 

\begin{algorithm}
$F' \gets \emptyset$.\\
\ForEach{$v \in V(H), j \in [0, \gamma'_{congRep})$}{
    Add to $F'$ a copy $T_{v, j}$ of the tree $T_v$ as defined in \Cref{def:coreGraphDirectSum}.
}
Construct a map $\Lambda = \Pi_{V(F) \mapsto [0, \gamma'_{congRep})}$ that maps nodes from $V(F)$ to numbers in $[0, \gamma'_{congRep})$ such that any two nodes $x,y \in V(F)$ that are identified with the same vertex in $V(H)$ have $\Lambda(x) \neq \Lambda(y)$. \\
\ForEach{$\hat{e}' \in E(F)$}{
    $\hat{e} \gets \Pi_{F \mapsto H}(\hat{e}')$.\\
    Let $i$ be such that $\hat{e} = (x,y)$ originates from core graph $\cC(G, F_i, \wstr_i)$.\\
    Let $e$ be the pre-image of $\hat{e}$ in $G$.\\
    Add to $F'$ the edge $(u', v')$ where $u'$ is the vertex in  $T_{\Pi_{V(F) \mapsto V(H)}(x),j}$ that is associated with vertex $u$ in $G$ in the forest $F_i$; and $v'$ the vertex in $T_{\Pi_{V(F) \mapsto V(H)}(y), j'}$ that is associated with vertex $v$ in $G$ in the forest $F_i$.
}
\Return $F'$.
\caption{$\textsc{MapHierarchicalForest}(F, H, G)$}
\label{alg:mapHierForest}
\end{algorithm}

Finally, the algorithm maintains $F'$ in the same way as it was initialized. It first processes the deletions to trees $F_i$ and forwards them to the copies $T_v$ in the graph $F'$ by deleting all corresponding edges. It then handles isolated vertex insertions into $F_i$ by forwarding them to $F'$ straightforwardly. 

Finally, it processes changes to $F$. If $F$ underwent a node insertion, the algorithm carefully updates the vertex map $\Lambda$ to map the new vertex in $F$ to a number that associates it with a unique tree $T_{v,j}$. For any edge deletion to $F$, we remove the corresponding projected edge from $F'$, and for every edge insertion to $F$, we add the corresponding projected edge to $F$, as previously described in \Cref{alg:mapHierForest}.

The vertex map $\Pi_{V(F') \mapsto V(G)}$ is rather trivial to maintain, as is the flat graph embedding $\Pi_{F' \mapsto G}$. We maintain the vertex map $\Pi_{A_F \mapsto V(F')}$ by mapping each vertex $v$ in $A_F$ that is mapped to node $x \in V(F)$ via $\Pi_{A_F \mapsto V(F)}$ to the root vertex $v$ of the tree $T_{v, \Lambda(x)}$. 

\paragraph{Analysis.} We can now give the proof of \Cref{lma:extensionOfMainThm}.

For the analysis, we call the edges in $F'$ that originate from copies of trees $T_v$ as defined in \Cref{def:coreGraphDirectSum} as \emph{forest edges} and let the remaining edges that are in one-to-one correspondence with edges in $F$ be called the \emph{projected edges}.

We note that the correctness of the maps $\Pi_{A_F \mapsto V(F')}, \Pi_{V(F') \mapsto V(G)}$ and the flat embedding $\Pi_{F' \mapsto G}$ is straightforward to verify. We thus only need to establish that these maps satisfy the claimed properties and that the algorithm runs in the claimed time. We also can verify immediately that the vertex maps remain constant on each element in the pre-image after it has been added.

\underline{Congestion of $\Pi_{F' \mapsto G}$:} We have that every edge $e \in E(G)$ occurs at most once in every forest $F_0, F_1, \ldots, F_{\lambda-1}$ for $\lambda = O(\log m)$. Thus, we have that there are at most $O(\log m)$ copies of any edge $e \in E(G)$ in the direct sum of all trees $T_v$ over $v \in V(H)$. The algorithm maintains $\gamma'_{congRep}$ copies of each such direct sum of trees over vertices in $H$ and thus there are at most $\Otil(\gamma'_{congRep})$ copies of each edge in $G$ among the forest edges of $F'$. 

For the remaining edges in $F'$, we have that these edges are from forest $F$ mapped via $\Pi_{F' \mapsto H}$ first, and then a copy of their pre-image in $G$ is added to $F'$. Thus, the number of copies of each edge $e \in E(G)$ that is added to $F'$ due to this mapping is at most $\econg(\Pi_{F \mapsto H}) \leq \gamma_{congRep}$ which bounds the number of copies of each edge in $G$ among the projected edges of $F'$.

Combining these two facts bounds the initial edge congestion of $\Pi_{F' \mapsto G}$ by $\Otil(\gamma'_{congRep})$ and since $F'$ is maintained in the same way as it is initialized, the congestion remains bounded, as desired.

\underline{Recourse of $F'$:} To bound the number of changes to $F'$, we have by \Cref{clm:degreeUpperBoundH} that the number of edge deletions and isolated vertex insertions to forests $F_0, F_1, \ldots, F_{\lambda-1}$ is at most $\gamma_{recVS} $ which results in at most $\gamma_{recVS} \cdot \gamma_{congRep} \cdot 2\Delta$ such updates among the forest edges of $F'$. Here, an additional factor $2\Delta$ was added since we define $T_v$ as a direct sum of trees rooted in the same vertex $v$ that are then merged in the root vertex. Whenever a deletion to a forest occurs, this might lead to new roots, and new roots might have to be merged with each other. However, since the maximum degree in $G$ is $\Delta$, we can implement such a merging procedure of two roots with at most $2\Delta$ edge insertions and deletions. 

Further, it is easy to verify that each change to $F$ results in at most one change to $F'$. This establishes the desired bound on the recourse of $F'$.

\underline{Stretch Bound:} The stretch bound follows immediately from the definition of the length function $l_H$ as given in \Cref{def:coregraph} which allows us to give a one-on-one map between segments of $\pi_{F'}(\Pi_{A_F \mapsto V(F')}(u),\Pi_{A_F \mapsto V(F')}(v)))$ between any two roots of trees $T_{v,j}$ that are in $F'$ and edges $\hat{e}' \in \pi_{F}(\Pi_{A_F \mapsto V(F)}(u),\Pi_{A_F \mapsto V(F)}(v))$ where each segment is of length in $G$ at most equal to the length of the edge $\Pi_{F \mapsto H}(\hat{e}')$ with respect to the length function $l_H$.

\underline{Runtime:} It is straightforward to observe the claimed runtime from the previous discussion, the fact that all embeddings associated with $F$ can be evaluated in constant time and by implementing $\Lambda$ using a data structure that keeps track of which copies of each tree $T_v$ are currently in use.

\section{A Toolbox for Dynamic Shortest Path Problems via Vertex Sparsifiers with Low Congestion}
\label{sec:sparsifier_hierarchy}

Next, let us describe how to construct a sparsifier hierarchy from the vertex sparsifier maintenance algorithm from \Cref{sec:vertexSparsifer}. We use this hierarchy to obtain the (Informal) Theorems \ref{thm:mainVS}, \ref{thm:mainTheoremAPSP}, and \ref{thm:mainTheoremLowDiamTreeOverview}.

In this section, we assume that the input graph $G$ has at least as many edges as vertices, i.e. that $n \leq m$, and that the number of updates to $G$ is at most $m$. This assumption is without loss of generality by standard techniques in dynamic graph algorithms. We further assume that $\log^{1/21} m$ is integer.

\subsection{Fully-Dynamic APSP with Worst-Case Subpolynomial Update Time}
\label{subsec:fullyDynAPSP}

We start by proving the following Theorem which augments the interface of the data structure \Cref{thm:mainTheoremAPSP} by an operation $\textsc{QueryDiameterWitnessPair}()$ that we require in the next sections.

\begin{restatable}{theorem}{mainTheoremAPSPFormal}\label{thm:mainTheoremAPSPFormal}
Given an $m$-edge input graph $G = (V,E,l)$ with polynomially-bounded lengths and maximum-degree $3$, there is a data structure $\textsc{DynamicAPSP}$ that initially outputs an empty set $X \subseteq V$ and supports a polynomial number of updates of the following type:
\begin{itemize}
    \item $\textsc{InsertEdge}(e)/ \textsc{DeleteEdge}(e)$: adds/removed edge $e$ into/from $G$. If the edge is inserted, its associated length $l(e)$ has to be in $[1,L]$ and the maximum degree is not allowed to exceed $3$.
    \item $\textsc{InsertIsolatedVertex}() / \textsc{DeleteIsolatedVertex}(v)$: inserts an isolated vertex to $G$ and returns its identifier/ deletes vertex $v$ from $G$ where $v$ has to be isolated.
    \item $\textsc{QueryDist}(u,v)$: returns a distance estimate $\widehat{\dist}(u,v) \in [\dist_G(u,v), \gamma_{ApproxAPSP} \cdot \dist_G(u,v)]$.
    \item $\textsc{AddDiameterSetVertex}(x)$/$\textsc{RemoveDiameterSetVertex}(x)$: adds/removes the vertex $x \in V$ to/from set $X$.
    \item $\textsc{QueryDiameterWitnessPair}()$: assumes that the set $X$ is of size at least $2$ upon invocation and if so it returns two vertices $x,y \in X$ such that  $\dist_G(x,y) \geq \diam(X) / \gamma_{ApproxAPSP}$.
\end{itemize}
For some $\gamma_{ApproxAPSP} =  e^{O(\log^{6/7} m \log\log m)}, \gamma_{timeAPSP} =  e^{O(\log^{20/21} m \log\log m)}$, the data structure can be initialized in time $m \cdot \gamma_{timeAPSP}$, and thereafter processes each edge/vertex update in worst-case time $\gamma_{timeAPSP}$ and each query in worst-case time $O(\log m)$. 
\end{restatable}
\begin{remark}\label{rmk:pathReporting}
The algorithm also supports operation $\textsc{QueryPath}(u,v)$ that returns a $uv$-path $P$ in $G$ with $l(P) \leq \gamma_{ApproxAPSP} \cdot \dist_G(u,v)$ in time $O(|P|\log m)$. 
\end{remark}

\paragraph{Maintaining the Hierarchy.} We maintain $\Lambda + 1$ hierarchy levels for $\Lambda = \log^{1/21} m$, and at each level $0 \leq i \leq \Lambda$, we maintain a graph $G_i$. We let $G_0 = G$, and for $0 \leq i < \Lambda$, we obtain $G_{i+1}$ by running the algorithm from \Cref{lma:workhorseSparsifier} on the graph $G_i$ with size reduction parameter $k =m^{1/\Lambda}$, number of internal levels $K = \log^{3/21} m$ and degree threshold $\Delta_i = 3 \cdot \gamma_{recVS}^i$. We rebuild all data structures on levels $i$ and above after every $u_{i} \defeq m^{1-(i+1)/\Lambda}$ updates to the graph $G$ (i.e. we re-initialize the graphs $G_{i+1}, G_{i+2}, \ldots, G_{\Lambda}$ maintained by these data structures). Finally, after every update to $G$, we run a classic APSP algorithm like Floyd-Warshall's algorithm on the graph $G_{\Lambda}$. 

\paragraph{The Distance Query Algorithm.} In this section, we describe the query for a distance estimate $\widehat{\dist}(u,v)$ upon inputting vertices $u,v \in V$. We defer the description and analysis of the query that returns two vertices $x,y \in V$ such that they are at a distance roughly equal to the diameter of $G$ to the very end of this section, but already point out that it uses the distance query procedure as a subroutine.

We denote by $p_{i+1}$ the pivot function, and by $A_{i+1}$ the set of vertices whose distance is preserved by $G_{i+1}$ as maintained by the algorithm from \Cref{lma:workhorseSparsifier} that is currently run on the graph $G_i$. We define the pivot function $p_{0}$ to be the identity on the vertex set $V$, and denote by $\hat{p}_i$ the function that maps any vertex $v \in V$ to vertex $p_i(p_{i-1}( \ldots p_1(p_0(v))))$. We do not maintain these functions $\hat{p}_i$ explicitly. 

\begin{algorithm}
\For{$i = 0, 1, \ldots, \Lambda$}{
    \If(\label{lne:ifConditionQuery}){$\hat{p}_{i}(v) \in B_{G_{i}}(\hat{p}_{i}(u), A_{i+1})$ \textbf{ or } $i = \Lambda$}{
        $d_{i, u} = \sum_{j = 0}^{i-1} \dist_{G_{j}}(\hat{p}_{j}(u), \hat{p}_{j+1}(u))$.\\
        $d_{i, v} = \sum_{j = 0}^{i-1} \dist_{G_{j}}(\hat{p}_{j}(v), \hat{p}_{j+1}(v))$.\\
    
        \Return $\widehat{\dist}(u,v) = d_{i, u}  + \dist_{G_{i}}(\hat{p}_{i}(u), \hat{p}_{i}(v)) + d_{i, v}$.
    }
}
\caption{$\textsc{QueryDist}(u,v)$}
\label{alg:query}
\end{algorithm}

\paragraph{Analysis of the Query.} Henceforth, we denote by $\gammaVSApprox = (\gamma_{\ell})^{O(K)}$ the precise worst-case guarantee on the distance-preservation approximation in \Cref{thm:spanner} for our choice of $K$.

\begin{claim}
Assuming that $m$ is reasonably large, before the $(i+1)$-th iteration of the for-loop in \Cref{alg:query}, we have for
\[
d_{i, u} = \sum_{j = 0}^{i-1} \dist_{G_{j}}(\hat{p}_{j}(u), \hat{p}_{j+1}(u)) \text{ and }
        d_{i, v} = \sum_{j = 0}^{i-1} \dist_{G_{j}}(\hat{p}_{j}(v), \hat{p}_{j+1}(v))
\]
that $d_{i,u}, d_{i,v}, \dist_{G_{i}}(\hat{p}_{i}(u), \hat{p}_{i}(v)) \leq (12 \cdot \gammaVSApprox)^{i} \cdot \dist_G(u,v)$. 

Further, if the for-loop terminates within the $(i+1)$-th iteration, it returns a distance estimate $\widehat{\dist}(u,v)$ such that $\dist_G(u,v) \leq \widehat{\dist}(u,v) \leq 3 \cdot (12 \cdot \gammaVSApprox)^{(i+1)} \dist_G(u,v)$.    
\end{claim}
\begin{proof}
We prove the claim by induction on $i$. For $i = 0$, we have trivially that $d_{0, u}, d_{0, v} = 0$ and $\dist_{G_0}(\hat{p}_0(u), \hat{p}_0(v)) = \dist_G(u,v)$ by definition. 

Next, let us analyze the $(i+1)$-th iteration of the algorithm in \Cref{alg:query}. We start by using the fact that 
\Cref{lma:workhorseSparsifier} yields that 
\begin{align}
\begin{split}\label{eq:upperBoundQuery}
    \dist_{G_{i+1}}(\hat{p}_{i+1}(u), \hat{p}_{i+1}(v)) \leq \gammaVSApprox \cdot \dist_{G_{i}}(\hat{p}_{i+1}(u), \hat{p}_{i+1}(v)). 
    \end{split}
\end{align}
We then use the triangle inequality and the induction hypothesis to derive
\begin{align}
\begin{split}\label{eq:upperBoundQuery2}
    \dist_{G_i}(\hat{p}_{i+1}(u), \hat{p}_{i+1}(v)) &\leq \dist_{G_{i}}(\hat{p}_{i+1}(u), \hat{p}_{i}(u)) + \dist_{G_{i}}(\hat{p}_{i}(u), \hat{p}_{i}(v)) + \dist_{G_{i}}(\hat{p}_{i}(v), \hat{p}_{i+1}(v))\\
    & \leq d_{i, u} + \dist_{G_{i}}(\hat{p}_{i}(u), \hat{p}_{i}(v)) + d_{i,v}\\
    & \leq 3 \cdot (12 \cdot \gammaVSApprox)^{i} \cdot \dist_G(u,v).
    \end{split}
\end{align}
We conclude via the following case analysis:
\begin{itemize}
    \item \underline{If the $(i+1)$-th iteration terminates:} since the $(i+1)$-th iteration terminates, we have that none of the conditions in the if-statement in \Cref{lne:ifConditionQuery} were satisfied, as the algorithm returns upon entering the if-statement. We thus have $\hat{p}_{i+1}(v) \not\in B_{G_{i+1}}(\hat{p}_{i+1}(u), A_{i+2})$ which is equivalent to 
    \[
    \dist_{G_{i+1}}(\hat{p}_{i+1}(u), \hat{p}_{i+2}(u)) < \dist_{G_{i+1}}(\hat{p}_{i+1}(u), \hat{p}_{i+1}(v)).
    \]
    Thus, we have that $d_{i+2,u} = d_{i+1,u} + \dist_{G_{i+1}}(\hat{p}_{i+1}(u), \hat{p}_{i+2}(u)) < (12\gammaVSApprox )^{i+1} \cdot \dist_G(u,v)$ by using \eqref{eq:upperBoundQuery} and \eqref{eq:upperBoundQuery2}, and finally the induction hypothesis and that $12\gammaVSApprox \geq 2$ which yields the desired constant to bound all terms of the geometric sum of $d_{i+1,u}$. 
    
    Further, we have $d_{i+2,v} = d_{i+1,u} + \dist_{G_{i+1}}(\hat{p}_{i+1}(v), \hat{p}_{i+2}(v))$ and we have from the minimality of the distance to the pivot and the triangle inequality that
    \begin{align*}
         \dist_{G_{i+1}}(\hat{p}_{i+1}(v), \hat{p}_{i+2}(v)) &\leq  \dist_{G_{i+1}}(\hat{p}_{i+1}(v), \hat{p}_{i+1}(u)) + \dist_{G_{i+1}}(\hat{p}_{i+1}(u), \hat{p}_{i+2}(u))\\
         &< 2 \cdot \dist_{G_{i+1}}(\hat{p}_{i+1}(v), \hat{p}_{i+1}(u))\\
         &\leq   
         6 \gamma_{approxVS} (12 \cdot \gamma_{approxVS})^i \cdot \dist_G(u,v) 
    \end{align*}
    and using again that $d_{i+1, v} \leq 6 \gamma_{approxVS} (12 \cdot \gamma_{approxVS})^i \cdot \dist_G(u,v)$ by induction and the fact that the terms form a geometric sum.
    
    This completes the proof that the invariant holds before the $(i+2)$-th iteration of the for-loop.

    \item \underline{Otherwise:} In the case that the algorithm terminates, we can use \eqref{eq:upperBoundQuery} and \eqref{eq:upperBoundQuery2} and finally the induction hypothesis, to straightforwardly derive that $\widehat{\dist}(u,v) \leq 3 \cdot (12 \cdot \gammaVSApprox)^{i+1} \dist_G(u,v)$. For the lower bound, it suffices to inspect \Cref{lma:workhorseSparsifier} and \Cref{thm:spanner} to see that all distances in $G_j$ for any $j \geq 0$ are overestimates.
\end{itemize}
\end{proof}

\begin{claim}
The query operation can be implemented in worst-case time $O(\Lambda)$.
\end{claim}
\begin{proof}
Note first that $\hat{p}_0(x) = p_0(x) = x$ and is thus trivial to evaluate. For $i \geq 1$, $\hat{p}_i(x) = p_i(\hat{p}_{i-1}(x))$ which again can be evaluated in time $O(1)$ give $\hat{p}_{i-1}(x)$. 

It is thus not hard to see, using \Cref{lma:workhorseSparsifier}, that the if-condition in \Cref{lne:ifConditionQuery} can be evaluated in time $O(1)$, and upon entering the if-statement, which occurs only once, the sums $d_{i,u}$ and $d_{i,v}$ can be calculated in time $O(i)$ since each term is maintained explicitly by one of the data structures from \Cref{lma:workhorseSparsifier}. Finally, if the if-statement is entered with $i = \Lambda$, we can simply read $\dist_{G_{i}}(\hat{p}_{i}(u), \hat{p}_{i}(v))$ off as we maintain all pairwise distances between vertices in $G_{\Lambda}$. Otherwise, if the if-statement is entered with $i < \Lambda$, then one of the other two conditions must have been true, which yields that we can extract the distance $ \dist_{G_{i}}(\hat{p}_{i}(u), \hat{p}_{i}(v))$ from \Cref{lma:workhorseSparsifier}. 
\end{proof}

Since $i \leq \Lambda$ as the algorithm terminates if it is in the $(\Lambda +1)$-th for-loop iteration, plugging in the values of $\gamma_{\ell}$, $\gammaVSApprox$ and $\Lambda$, we thus obtain the following corollary. Here, the extension to return a witness path is straightforward from the construction of the distance estimates in \Cref{alg:query} and the properties of \Cref{lma:workhorseSparsifier}.

\begin{corollary}\label{cor:summarizeDistanceQuery}
The query operation returns in worst-case time $O(\Lambda)$ a distance estimate $\widehat{\dist}(u,v)$ such that $\dist_G(u,v) \leq \widehat{\dist}(u,v) \leq \gamma_{approxDistQuery} \cdot \dist_G(u,v)$ for $\gamma_{approxDistQuery} = e^{O(\log^{6/7} m \log\log m)}$. The operation can further be extended to return a $uv$-path $P$ in $G$ of length at most $\widehat{\dist}(u,v)$ in time $\Otil(|P|)$.
\end{corollary}

\paragraph{Runtime Analysis.} Finally, we analyze the runtime of algorithm.

\begin{claim}\label{clm:runtimeFulylDynAPSP}
The algorithm takes initialization time $e^{O(\log^{14/15} m \log\log m)} \cdot m$ and thereafter can process each update to $G$ in time $e^{O(\log^{14/15} m \log\log m)}$.
\end{claim}
\begin{proof}
We have from \Cref{lma:workhorseSparsifier} that each graph $G_{i+1}$ undergoes at most $\prod_{0 \leq j \leq i} (\gamma_{recVS} \cdot \Delta_j) = (\gamma_{recVS})^{4 i^2}$ updates for every update to $G$ unless $G_{i+1}$ is re-initialized because of a rebuild of level $i$. 

The data structure from \Cref{lma:workhorseSparsifier} that is inputted $G_i$ and maintains $G_{i+1}$ further processes each update to $G_i$ with worst-case update time $\gamma_{recVS}  (\gamma_{\ell})^{O(K^2)} \Delta_i = (\gamma_{recVS})^{O(i)}  (\gamma_{\ell})^{O(K^2)}$. Using that $i \leq \Lambda$, we obtain that each update to $G_i$ can be processed in time $e^{O(\log^{20/21} m \log\log m})$ by the data structure that maintains $G_{i+1}$. 

Using the upper bound on the number of updates to $G_{i+1}$ per update to $G$, and the fact that the data structure at level $i$ that maintains $G_{i+1}$ is rebuilt every $u_i = m^{1-(i+1)/\Lambda}$ updates to $G$, we have that since $G_{i+1}$ was last re-initialized, that there were at most $u_i \cdot (\gamma_{recVS})^{4i^2}$ updates to $G_i$. 

Thus, we can prove by induction that the maximum number of vertices and edges in $G_{i+1}$ is at most $m_{i+1} \defeq 2 m^{1-(i+1)/\Lambda} \cdot (\gamma_{recVS})^{4(i+1)^2}$. This is clearly true for $m_0 \leq 2m$ which follows since graph $G_0 = G$ undergoes at most $m$ updates overall which trivially upper bounds the number of edges in $G$ by $2m$. For the inductive step $i \mapsto i+1$, we have from \Cref{lma:workhorseSparsifier} that the number of vertices and edges in $G_{i+1}$ is then, by straightforward calculations, at most 
\begin{align*}
\gamma_{recVS} \cdot \Delta_i \cdot (m_i/k + u_i) &\leq  3  \cdot (\gamma_{recVS})^{(i+1)} \cdot (2 m^{1-(i+1)/\Lambda} \cdot (\gamma_{recVS})^{4i^2} + m^{1-(i+1)/\Lambda})\\
&\leq  3  \cdot (\gamma_{recVS})^{(i+1)} \cdot 3 m^{1-(i+1)/\Lambda} \cdot (\gamma_{recVS})^{4i^2}\\
& < 2m^{1-(i+1)/\Lambda} \cdot (\gamma_{recVS})^{4(i+1)^2}\\
& = m_{i+1}.
\end{align*}

We can now use the bounds on the maximum number of edges in $G_{i+1}$ to straightforwardly argue about the time spent on rebuilding each level $i$ and above. However, we cannot afford to rebuild such a data structure during a single time step. Instead, we use a standard de-amortization technique (see for example \cite{gutenberg2020fully, bernstein2022fully}). Here, instead of rebuilding the data structure at level $i$ on the spot to produce the graph $G_{i+1}$, we can instead start already $u_i$ updates to $G$ earlier to obtain such a new data structure and forward the additional updates to $G_i$ that are issued in the meantime to this data structure such that it reflects the current graph $G_i$ once it is needed. Using the parameters above, we can prove that this increases the number of edges in $G_{i+1}$ to at most $2^{i+1} m_{i+1}$. 

But now, the time to rebuild the data structure at level $i$ can be split evenly over the sequence of $u_i$ updates (here it is essential that \Cref{lma:workhorseSparsifier} processes updates with worst-case update time itself). It is not hard to show that this yields a worst-case update time of $e^{O(\log^{20/21} m \log\log m})$ by \Cref{lma:workhorseSparsifier}. 

Finally, the runtime spend on computing APSP on $G_{\Lambda}$ can be bounded straight-forwardly by $O(m_{\Lambda}^3) = O((\gamma_{recVS})^{12\Lambda^2}) = e^{O(\log^{20/21} m \log\log m})$ worst-case time per update to $G$, as desired. 
\end{proof}

\paragraph{Obtaining a Witness-Pair for the Diameter of $X$.}  It remains to describe how to support queries for a pair of vertices $x,y \in X$ in $G$ such that $\dist_G(x,y) \geq \diam_G(X) / \gamma_{approxAPSP}$ as described in \Cref{thm:mainTheoremAPSP}.

To implement this query efficiently we need the following dynamic tree data structure which can be obtained rather straight-forwardly from link-cut trees \cite{sleator1981data} or top tree \cite{alstrup2005maintaining}.

\begin{theorem}\label{thm:fancyTreeDS}
Given a directed $n$-vertex rooted forest graph $F = (V, E, l)$, there is a data structure that maintains an initially empty set $M \subseteq V$ implicitly and that supports the following operations:
\begin{itemize}
    \item $\textsc{InsertEdge}(e)/\textsc{DeleteEdge}(e)$: adds/deletes edge $e$ to/from $F$. This operation assumes that the operation results in $F$ being a rooted forest graph again.
    \item $\textsc{MarkVertex}(v)/ \textsc{UnmarkSubtreeRootedAt}(v)$: the former operation adds vertex $v$ to the set $M$; the latter operation removes all vertices that are ancestors of $v$ (including $v$) from the set $M$.
    \item $\textsc{QueryFarthestMarkedAncestor}(v)$: returns a vertex $w \in M$ that is an ancestor of $v$ of maximal length on the $wv$-path or $\bot$ if no such vertex exists.
\end{itemize}
The algorithm takes initially time $\Otil(n)$ to preprocess $F$ and thereafter supports every operation in worst-case time $\Otil(1)$.
\end{theorem}

To support the query for a witness of the diameter of set $X$ in $G$, we maintain a dynamic tree $T$ (via a data structure $\mathcal{D}_T$ from \Cref{thm:fancyTreeDS}) with $\Lambda + 1$ levels where level $0$ has a node for each vertex in $V(G_0)$, level $1$ has a node for each vertex in $V(G_1)$, and so on up until level $\Lambda$ which has a node for every vertex in $V(G_{\Lambda})$. We henceforth talk interchangeably of the nodes in $T$ and the vertices in $V(G_i)$ for every level $i$. Further, we have an edge $(u,v)$ in $T$ for every $u \in V(G_i)$ and $v \in V(G_{i+1})$ where $v = p_{i+1}(u)$ and let each such edge be of weight $\dist_{G_i}(u, v)$. We mark all vertices that are in the set $X$. 

The query operation is given in \Cref{alg:queryWitnessPair}. We pick an arbitrary vertex $x \in X$ and then try to find a vertex $y \in X$ that is far in $G$ from $x$ using the distance oracles. In order to achieve this goal, the query algorithm uses various operations on the data structure $\mathcal{D}_T$. Before returning $y$, however, the algorithm reverts all such information such that $\mathcal{D}_T$ is in the same state after the query as it was before the query.

\begin{algorithm}
Let $x$ be an arbitrary vertex in $X$.\\
$y \gets x$; $d \gets 0$.\\
\For(\label{lne:witnessForLoop}){$i = 0, 1, \ldots, \Lambda -1$}{
    \ForEach{$w \in B_{G_i}(\hat{p}_i(x), A_{i+1}) \cup \{\hat{p}_{i+1}(x)\}$}{
        $y' \gets \mathcal{D}_T.\textsc{QueryFarthestMarkedAncestor}(w)$.\label{lne:queryFarthestVertex}\\
        \If{$y \neq \bot$}{
            $d' \gets \textsc{QueryDist}(x, y')$.\\
            \lIf{$d' > d$}{
                $d \gets d'$; $y \gets y'$.
            }
        }
        $\mathcal{D}_T.\textsc{UnmarkSubtreeRootedAt}(w)$.\label{lne:unmarkLowerTier}
    }
}
Revert all operations to the data structure $\mathcal{D}_T$.\\
\Return $y$.
\caption{$\textsc{QueryDiameterWitnessPair}()$}
\label{alg:queryWitnessPair}
\end{algorithm}

Let us establish the correctness of the query operation.

\begin{claim}
The vertices $x,y \in X$ returned by the procedure above satisfy \[
\dist_G(x,y) \geq \diam_G(X)/ (2 \gamma_{approxDistQuery}).
\]
\end{claim}
\begin{proof}
Consider any vertex $y'$ that is found in \Cref{lne:queryFarthestVertex} of \Cref{alg:queryWitnessPair} during the $(i+1)$-th iteration of the for-loop starting in  \Cref{lne:witnessForLoop} (which means that the counter is equal to $i$). We have that $y' \in X$ since we have that $y'$ is in the set of marked vertices $M$ at the time that it was queried via operation $\textsc{QueryFarthestMarkedAncestor}$ and before the query algorithm started, we had $M$ equal to $X$ and thereafter no invocation of the procedure $\textsc{MarkVertex}(v)$ which implies that $y' \in X$.

We claim that when invoking the distance query algorithm in \Cref{alg:query} for tuple $(x, y')$ (note that the distance query algorithm is non-symmetric), we have that it returns in the $(i+1)$-th iteration and thus the distance estimate $\widehat{\dist}(x,y')$ is taken to be the sum $d_{i, x} + \dist_{G_i}(\hat{p}_i(x), \hat{p}_i(y')) + d_{i, y'}$. To make this observation, note that the distance query returns in the $(i+1)$-th iteration if and only if 
 for every $j < i$, we have $\hat{p}_j(y') \not\in B_{G_j}(\hat{p}_j(x), A_{j+1})$ and $\hat{p}_i(z) \not\in B_{G_i}(\hat{p}_i(x), A_{i+1})$. But the former guarantee implies that none of the invocations of the procedure $\mathcal{D}_T.\textsc{UnmarkSubtreeRootedAt}(\cdot)$ in \Cref{lne:unmarkLowerTier} in the for-loop iterations up to the $(i+1)$-th iteration could have unmarked the vertex $y'$ and thus at the beginning of the $(i+1)$-th iteration, we still have $y' \in M$. Further, the same argument shows that $y' \not\in M$ after the $(i+1)$-th while-loop iteration and thus in no iteration thereafter as $M$ is decreasing over the course of the query algorithm.

Now, consider any vertex $z \in X$ at maximal distance $\dist_G(x,z)$ from $x$. We have that $\dist_G(x,z) \geq \diam_G(X)/ 2$ from the triangle inequality. Let $i$ be the index such that if we invoke the distance query from \Cref{alg:query} for tuple $(x, z)$, it returns in the $(i+1)$-th iteration with distance estimate $\widehat{\dist}(x,z)$ is taken to be the sum $d_{i, x} + \dist_{G_i}(\hat{p}_i(x), \hat{p}_i(z)) + d_{i, z}$. It is not hard to argue again that during the $(i+1)$-th for-loop iteration before iterating over vertex $ \hat{p}_i(z)$ in the foreach-loop, we have $z \in M$.

Consider next the $(i+1)$-th iteration of the foreach-loop in \Cref{lne:witnessForLoop} and the foreach-loop iteration with $w = \hat{p}_i(z)$. Let $y'$ be the vertex chosen in \Cref{lne:queryFarthestVertex} during this iteration. Note that since $y'$ is at least as far from $w$ in $T$ as $z$ (which was still marked and therefore available), we have that $d_{i,y'} \geq d_{i,z}$. But this implies that 
\[
\widehat{\dist}(x,y') = d_{i, x} + \dist_{G_i}(\hat{p}_i(x), w) + d_{i, y'} \geq d_{i, x} + \dist_{G_i}(\hat{p}_i(x), z) + d_{i,z} = \widehat{\dist}(x,z).
\]
The claim now follows immediately from \Cref{cor:summarizeDistanceQuery}.
\end{proof}

The runtime for all update operations is established straightforwardly by noting that the tree $T$ can be maintained from the information supplied by the data structures from \Cref{lma:workhorseSparsifier} maintained in the hierarchy and using the guarantees from \Cref{thm:fancyTreeDS}. Further, we have that \Cref{alg:queryWitnessPair} can be implemented in time $\Otil(\Lambda k) = \Otil(k)$ as the size of each ball $B_{G_i}(\hat{p}_i(x), A_{i+1}) \cup \{\hat{p}_{i+1}(x)\}$ is of size at most $\Otil(k)$ again by \Cref{lma:workhorseSparsifier} and the guarantees from \Cref{thm:fancyTreeDS} which implies that rolling back $\ell$ operations takes $\Otil(\ell)$ time. Finally, we note that the query time mismatches the stated query time in \Cref{thm:mainTheoremAPSPFormal}, however, this can be remedied by calling the procedure $\textsc{QueryDiameterWitnessPair}()$ in \Cref{alg:queryWitnessPair} after every update to $G$ and $X$ and then store the two vertices $x,y$ that it returns. Upon a query invocation on the data structure in \Cref{thm:mainTheoremAPSPFormal}, we can then simply return these two pre-stored vertices instead of running the query algorithm.

This concludes the proof of \Cref{thm:mainTheoremAPSPFormal} which follows from the series of claims in this section.

\subsection{Maintaining a Low-Diameter Hierarchical Tree}
\label{subsec:maintainLowDiam}

In this section, we prove the following Theorem which is an extended version of \Cref{thm:mainTheoremLowDiamTreeOverview}.

\begin{theorem}\label{thm:mainTheoremLowDiamTree}
Given an $m$-edge input graph $G = (V,E,l)$ with polynomial lengths in $ [1,L]$ and maximum degree $3$. There is a data structure $\textsc{LowDiamTree}$ that maintains a flat hierarchical forest $F$ over $G$ that supports a polynomially-bounded number of updates of the following type:
\begin{itemize}
    \item $\textsc{InsertEdge}(e)/ \textsc{DeleteEdge}(e)$: adds/removed edge $e$ into/from $G$. If the edge is inserted, its associated length $l(e)$ has to be in $[1,L]$ and the maximum degree is not allowed to exceed $3$; if it is deleted, it has to be ensured that thereafter graph $G$ is still connected.
\end{itemize}
Under these updates, the algorithm maintains the flat hierarchical forest $F$ over $G$ along with graph embedding $\Pi_{F \mapsto G}$ where the length function $l_F$ of $F$ is defined by $l_F(e) = l(\Pi_{F \mapsto G}(e))$ for every edge $e \in E(F)$ and vertex maps $\Pi_{V(G) \mapsto V(F)}, \Pi_{V(F) \mapsto V(G)}$, such that, for some $\gamma_{lowDiamTree}= e^{O(\log^{20/21} m \log\log m})$, at any time:
\begin{enumerate}
    \item $\diam_F(\Pi_{V(G) \mapsto V(F)}(V)) \leq \gamma_{lowDiamTree}\cdot \diam(G)$, and
    \item  \label{prop:lowCongLowDiam} we have $\econg(\Pi_{F \mapsto G}) \leq \gamma_{lowDiamTree}$, and
    \item \label{prop:fewEdgesLowDiam} $F$ consists of at most $\gamma_{lowDiamTree} \cdot m$ vertices and edges.
\end{enumerate}
The algorithm maintains the flat hierarchical forest $F$ and all maps explicitly. Vertex maps are such that once an element is added to the pre-image, its image remains fixed until the element is again removed.

The algorithm is deterministic, can be initialized in time $m \cdot \gamma_{lowDiamTree}$, and thereafter processes each edge insertion/deletion in amortized time $\gamma_{lowDiamTree}$.
\end{theorem}

We point out that above, Property \ref{prop:fewEdgesLowDiam} follows immediately from Property \ref{prop:lowCongLowDiam} but we include it as an additional Property to inform the intuition of the reader.

\paragraph{The Algorithm.} The Theorem above can be obtained rather straightforwardly by augmenting the hierarchy from \Cref{sec:sparsifier_hierarchy} in the following way: we let $F_i$ denote the forest described in Property \ref{prop:dynamicShortestPathPivotForest} of \Cref{lma:workhorseSparsifier}, obtained by the data structure run in the hierarchy on graph $G_i$ that maintains graph $G_{i+1}$ and pivot set $A_{i+1}$. We let $F_{\Lambda}$ denote the forest (or rather tree) obtained as the single-source shortest path tree from an arbitrary vertex $v \in V(G_{\Lambda})$ in $G_{\Lambda}$ where $F_{\Lambda}$ is re-computed after every update to $G$. We let $A_{\Lambda +1} = \emptyset$. Note that each connected component of the forest $F_i$ contains exactly one vertex from $A_{i+1}$.

Next, for every $0 \leq i < \Lambda$, we have that $F_i$ is a forest in $G_i$ and thus it trivially forms a flat hierarchical forest over $V(G_i)$, and we feed $F_i$ to the data structure at level $i-1$ and use \Cref{lma:extensionOfMainThm} to obtain flat hierarchical forest $F_i^{i-1}$ over $V(G_i)$ in $G_{i-1}$, then feed $F_i^{i-1}$ to the data structure at level $i-2$ and use \Cref{lma:extensionOfMainThm} to obtain $F_i^{i-2}$, and so on until we obtain forest $F_i^{0}$ over $V(G_i)$ in the graph $G_0 = G$. We let $\Pi_{V(G_i) \mapsto V(F_i^j)}$ be the vertex map that maps the terminal set $V(G_i)$ into the tree $F_i^j$.

Finally, we define the hierarchical tree $F$ described in \Cref{thm:mainTheoremLowDiamTree} as the direct sum of forests $F^{0}_{0}, F^0_1, F^0_2, \ldots, F^0_{\Lambda}$ where we merge the nodes in $\Pi_{V(G_0) \mapsto F^{0}_0}(A_1)$ with the nodes in $\Pi_{V(G_1) \mapsto F^{0}_1}(A_1)$, then the nodes in $\Pi_{V(G_1) \mapsto F^{0}_1}(A_2)$ with the nodes in $\Pi_{V(G_2) \mapsto F^{0}_2}(A_2)$, and so on until level $\Lambda$. 

We define the vertex map $\Pi_{V(G) \mapsto V(F)}$ to map each vertex $v \in V(G)$ to the node in $F$ that is identified with the vertex $v$ in the set $V(G_0)$. It is straightforward to define the flat graph embedding associated with $F$ from the corresponding maps associated with forests $F^{0}_{0}, F^0_1, F^0_2, \ldots, F^0_{\Lambda}$ which then also implicitly defines the vertex map $\Pi_{V(F) \mapsto V(G)}$.

\paragraph{Analysis.} In the following analysis, we do not obtain tight bounds but rather upper bound all subpolynomial factors loosely by $\gamma_{lowDiamTree}$ to simplify the proof.

\begin{claim}
The forest $F$ along with associated maps $\Pi_{V(G) \mapsto V(F)}, \Pi_{V(F) \mapsto V(G)}, \Pi_{F \mapsto G}$ can be maintained such that $F$ has amortized recourse $e^{O(\log^{20/21} m \log\log m})$, that $\econg(\Pi_{F \mapsto G}) = e^{O(\log^{20/21} m \log\log m})$ and $\diam_F(\Pi_{V(G) \mapsto V(F)}(V)) \leq \gamma_{lowDiamTree}\cdot \diam(G)$, and that the map $\Pi_{V(G) \mapsto V(F)}$ has its node in its image being a leaf node in $F$. Vertex maps are such that once an element is added to the pre-image, its image remains fixed until the element is again removed.

The algorithm to maintain these objects takes initialization time $m \cdot e^{O(\log^{20/21} m \log\log m})$ and thereafter amortized update time $e^{O(\log^{20/21} m \log\log m})$ per update to $G$.
\end{claim}
\begin{proof}
For $i < \Lambda$, we have from \Cref{lma:workhorseSparsifier} that the forest $F_i$ is maintained by the data structure, and that they undergo $k \cdot e^{O(\log^{20/21} m \log\log m}) = e^{O(\log^{20/21} m \log\log m})$ changes per update to $G$ (this is explicit in the analysis of \Cref{clm:runtimeFulylDynAPSP}). For $i = \Lambda$, the bound follows by the bound on the size of $G_{\Lambda}$ analyzed in \Cref{clm:runtimeFulylDynAPSP}.

Thus, using \Cref{lma:extensionOfMainThm} recursively, we obtain that each forest $F_i^0$ has associated graph embedding $\Pi_{F^0_i \mapsto G}$ with edge congestion in $(\gamma_{recVS})^{O(\Lambda^2)} = e^{O(\log^{20/21} m \log\log m})$ by using that $F_i$ has congestion $1$ into graph $G_i$ and by leveraging the degree bounds $\Delta_i$ for each graph $G_i$. It is also not hard to see from \Cref{clm:runtimeFulylDynAPSP} that $\Pi_{F^0_i \mapsto G_i}$ can be maintained in amortized runtime $e^{O(\log^{20/21} m \log\log m})$. Above, we only obtain amortized bounds since we cannot apply the de-amortization of the batching technique for the precise definition of $F$ above but rather have to rebuild every now and then at a larger cost per update.

Again from \Cref{lma:workhorseSparsifier}, we have that the forests $F_i$ are maintained by the data structure, and that they undergo $k \cdot e^{O(\log^{20/21} m \log\log m}) = e^{O(\log^{20/21} m \log\log m})$ changes per update to $G$ (by the analysis from \Cref{clm:runtimeFulylDynAPSP}). And thus, via  \Cref{lma:extensionOfMainThm}, we can upper bound the recourse to each forest $F_i^0$ by $e^{O(\log^{20/21} m \log\log m})$ per update to $G$. Finally, we have that to obtain $F$ from forests $F^{0}_{0}, F^0_1, F^0_2, \ldots, F^0_{\Lambda}$ we take the direct sum and merge certain nodes. But here we only merge at most one node from each forest identified with the same vertex $v \in V$ and using the bound on the edge congestion of each map $\Pi_{F_i^0 \mapsto G}$ and that the graph $G$ is a graph with constant maximum degree, we thus have that the merge procedure can be implemented to cause additional recourse to $F$ at most $e^{O(\log^{20/21} m \log\log m})$ per update to any forest $F_i^0$. 

The bound on the diameter of the vertex set $V$ mapped into $V(F)$ with respect to $F$ can be obtained by observing that every component $C$ of a forest $F_i$ for any $0 \leq i \leq \Lambda$ is taken as the union of pivot paths to a root vertex $v$, thus it is a truncated shortest path tree in $G_i$ from $v$. But by the distance preserving properties of $G_i$ with respect to $G$, this implies that each such component has diameter at most $\diam(G) \cdot e^{O(\log^{20/21} m \log\log m})$. We further have that the tree $T$ obtained as the union of forests $F_0, F_1, \ldots, F_{\Lambda}$ has any path $P$ in $T$ such that it can be segmented to have at most two segments in each forest $F_i$, and thus the diameter the tree $T$ with respect to length functions of $F_0, F_1, \ldots, F_{\Lambda}$ is at most $\diam(G) \cdot e^{O(\log^{20/21} m \log\log m})$. Finally, it is not hard to see that instead of taking forests $F_0, F_1, \ldots, F_{\Lambda}$ directly, when using forests $F_0^0, F_1^0, \ldots, F_{\Lambda}^0$ to form $F$, we have that all vertices in $V$ are mapped via $\Pi_{V(G) \mapsto V(F)}$ into the same connected component of $F$ and that this connected component has diameter at most equal to the diameter of $T$ with respect to the lengths in $G$ by \Cref{lma:extensionOfMainThm}, as desired.

The time required to maintain the maps $\Pi_{V(G) \mapsto V(F)}, \Pi_{V(F) \mapsto V(G)}, \Pi_{F \mapsto G}$ can be bound asymptotically be the number of changes to $F$, which yields the runtime bound.
\end{proof}

\paragraph{Extending \Cref{thm:mainTheoremLowDiamTree} to Maintaining a Low-Depth Hierarchical Tree $T$.} While the goal of \Cref{thm:mainTheoremLowDiamTree} is to have the forest $F$ to be a flat hierarchical tree, for some applications it is simpler to work with a hierarchical tree that has low depth, i.e. where every path in the forest consists of only few edges. Here, we show how to maintain such a hierarchical forest $T$, in fact, a hierarchical tree (note the forest can no longer be flat) using the algorithm above. A subtle but crucial detail here is that it is no longer possible to maintain the embedding paths of $\Pi_{T \mapsto G}$ explicitly, i.e. output the new path whenever it changes. Since mapping into $G$ is crucial for many applications, we show that instead, we can map from $T$ back into $F$ where we can then find the embedding paths explicitly. The Lemma below summarizes our results.

\begin{lemma}[Extension of \Cref{thm:mainTheoremLowDiamTree}.]\label{lma:maintThmLowDiamExtension}
Given inputs as in \Cref{thm:mainTheoremLowDiamTree}, and let $F, \Pi_{F \mapsto G}$, $\Pi_{V(F) \mapsto V(G)}$, $\Pi_{V(G) \mapsto V(F)}$ be the objects maintained by the algorithm from \Cref{thm:mainTheoremLowDiamTree}. 

Then, the algorithm can additionally maintain a hierarchical tree $T$ over $G$ with graph embedding $\Pi_{T \mapsto G}$ and vertex maps $\Pi_{V(G) \mapsto V(T)}$ and $\Pi_{V(T) \mapsto V(G)}$ such that: 
\begin{enumerate}
    \item $\diam_{T}(\Pi_{V(G) \mapsto V(T)}(V)) \leq \gamma_{lowDiamTree}\cdot \diam(G)$, and
    \item $\econg(\Pi_{T \mapsto G}) \leq \gamma_{lowDiamTree}$, and
    \item every path in $T$ consists of at most $\gamma_{lowDiamTree}$ many edges.
\end{enumerate}
While the algorithm only explicitly maintains the forest $F$ and the vertex maps $\Pi_{V(G) \mapsto V(T)}$ and $\Pi_{V(T) \mapsto V(G)}$ explicitly (but not the graph embedding $\Pi_{T \mapsto G}$), it also explicitly maintains a vertex map $\Pi_{V(T) \mapsto V(F)}$ such that for any two nodes $x,y \in V(T)$, we have that 
\[
\Pi_{F \mapsto G}(\pi_F(\Pi_{V(T) \mapsto V(F)}(x), \Pi_{V(T) \mapsto V(F)}(y))) = \Pi_{T \mapsto G}(x,y).
\]
The algorithm from \Cref{thm:mainTheoremLowDiamTree} can maintain $T$ and the vertex maps $\Pi_{V(G) \mapsto V(T)}$, $\Pi_{V(T) \mapsto V(G)}$, and $\Pi_{V(T) \mapsto V(F)}$ explicitly with at most a constant asymptotic increase in initialization and amortized update time.
\end{lemma}

We point out that the tree $T$ can be taken to be the direct sum of forests $F_0, F_1, \ldots, F_{\Lambda}$ where the vertices in $A_1$ in $F_0$ are merged with the vertices in $A_1$ in $F_1$, the vertices in $A_2$ in $F_1$ are merged with the vertices in $A_2$ in $F_2$, and so on. 

It is not hard to see that this process yields a tree $T$ with properties as described in the Lemma above. For the bound on the number of edges on any path in $T$, we observe that any path in $T$ contains at most two segments in every forest $F_i$, and each path in $F_i$ is of length at most $\Otil(k) = e^{O(\log^{20/21}m \log\log m)}$ by \Cref{lma:workhorseSparsifier}.

\subsection{Maintaining a Vertex Sparsifier}

Finally, we describe how to maintain a vertex sparsifier as described in \Cref{thm:mainVS}. Here, we give an extended version of \Cref{thm:mainVS} that we believe to be useful for various applications. We obtain the theorem almost immediately by a straightforward combination of the techniques from \Cref{subsec:fullyDynAPSP} and
\Cref{subsec:mappingHierachVS}.

\begin{theorem}\label{thm:mainTheoremVSExt}
Given an $m$-edge input graph $G = (V,E,l)$ with polynomial lengths in $ [1,L]$ and maximum degree at most $3$. Then, for some $\gamma_{vertexSparsifier} = e^{O(\log^{20/21} m \log\log m})$, there is a data structure $\textsc{MaintainVertexSparsifier}$ that initially outputs an empty set $A$, and graph $H$ consisting of at most $\gamma_{vertexSparsifier}$ vertices and edges, and supports a polynomial number of updates of the following type:
\begin{itemize}
    \item $\textsc{InsertEdge}(e)/ \textsc{DeleteEdge}(e)$: adds/removed edge $e$ into/from $G$. If the edge is inserted, its associated length $l(e)$ has to be in $[1,L]$ and the maximum degree of $G$ is not allowed to exceed $3$.
    \item $\textsc{AddTerminalVertex}(a)/ \textsc{RemoveTerminalVertex}(a)$: adds/ removes the vertex $a \in V(G)$ to/from the terminal set $A$.
\end{itemize}
The algorithm processes the $t$-th update and outputs a batch of updates $U_H^{(t)}$ consisting of edge insertions/deletions, and isolated vertex insertions/deletions and that when applied to the graph $H^{(t-1)}$ yields graph $H^{(t)}$ such that, we have: 
\begin{itemize}
    \item we have $A \subseteq V(H^{(t)}) \subseteq V(G)$, and
    \item for all vertices $u,v \in V(H^{(t)})$, we have $\dist_{G^{(t)}}(u,v) \leq \dist_{H^{(t)}}(u,v)$ and further if $u,v \in A^{(t)}$ then we also have $\dist_{H^{(t)}}(u,v) \leq \gamma_{vertexSparsifier} \cdot \dist_{G^{(t)}}(u,v)$, and
    \item the number of edges and vertices in $H^{(t)}$ is at most $(1+|A|) \cdot \gamma_{vertexSparsifier}$, and
    \item we have $\sum_{t' \leq t} |U^{(t')}_H| \leq \gamma_{vertexSparsifier} \cdot t$.
\end{itemize}
The algorithm is deterministic, and initially takes time $m \cdot \gamma_{vertexSparsifier}$. Every update is processed in worst-case time $\gamma_{vertexSparsifier}$.

Further, given a dynamic flat hierarchical forest $F$ over a monotonically increasing set $A_F \subseteq V(H)$ in graph $H$ along with vertex maps $\Pi_{A_F \mapsto V(F)}, \Pi_{V(F) \mapsto V(H)}$ and flat graph embedding $\Pi_{F \mapsto H}$, along with parameters $\gamma_{congRep}$ and $\gamma_{recRep}$ such that at any time  $\econg(\Pi_{F \mapsto H})$ is  bounded by $\gamma_{congRep}$ and the number of changes to $F$ caused by an update to $G$ is upper bounded by $\gamma_{recRep}$. We require the vertex maps to be such that whenever a vertex is added to the pre-image, its image remains constant for the rest of the algorithm.

Then, the algorithm can maintain a flat hierarchical forest $F'$ over set $A_F$ in graph $G$ along with vertex maps $\Pi_{A_F \mapsto V(F')}, \Pi_{V(F') \mapsto V(G)}$ and flat graph embedding $\Pi_{F' \mapsto G}$ such that at any time $\econg(\Pi_{F' \mapsto G})$ is bounded by $\gamma_{congRep} \cdot \gamma_{vertexSparsifier}$ and the number of changes to $F'$ per update to $G$ is $\Otil(\gamma_{recRep} + \gamma_{congRep}  \cdot \gamma_{vertexSparsifier})$, and we have for any two vertices $u,v \in A_F$ that $l_G(\Pi_{F' \mapsto G}(\pi_{F'}(\Pi_{A_F \mapsto V(F')}(u),\Pi_{A_F \mapsto V(F')}(v)))) \leq l_H(\Pi_{F \mapsto H}(\pi_{F}(\Pi_{A_F \mapsto V(F)}(u),\Pi_{A_F \mapsto V(F)}(v))))$. Further, we have that the vertex maps are such that whenever a vertex is added to the pre-image, its image remains constant for the rest of the algorithm.

Having $A_F, F$ and the vertex maps and graph embedding associated with $F$ maintained, the algorithm to maintain $F'$ and the vertex maps and graph embedding associated with $F'$ requires additional initialization time $\tilde{O}(m \cdot \gamma_{congRep})$ and processes every update with additional worst-case time $\Otil(\gamma_{recRep} + \gamma_{congRep} \cdot \gamma_{vertexSparsifier})$.
\end{theorem}

In this section, we only provide the implementation of the above theorem with amortized update time guarantees instead of worst-case update times. We refer the reader to \Cref{subsec:fullyDynAPSP} for a detailed discussion on how to de-amortize the algorithm.

Let us start by giving the algorithm to maintain the sparsifier $H$ as described in the first half of the theorem, i.e. without yet specifying how to map a given forest $F$ from $H$ back into $G$.

\paragraph{The Sparsifier Maintenance Algorithm.} We run the algorithm from \Cref{subsec:fullyDynAPSP} on the graph $\hat{G}$ that is initialized to $G$. We obtain the vertex sparsifier hierarchy $G_0, G_1, \ldots, G_{\Lambda}$ on $\hat{G}$ as described in the \Cref{subsec:fullyDynAPSP}.

Throughout, we maintain the set $A$ as specified, and maintain $H = G_{\ell}$ for some index $0 \leq \ell \leq \Lambda$ where initially we have $\ell = \Lambda$, and thus initially $H = G_{\Lambda}$.

We then forward edge updates to $G$ to the graph $\hat{G}$ and thus to the data structure maintaining the hierarchy $G_0, G_1, \ldots, G_{\Lambda}$. Additionally, whenever the update operation $\textsc{AddTerminalVertex}(a)$ is invoked, we check whether after adding $a$ to $A$, we have that the size of $A$ is $m^{1 - (\ell-1)/\Lambda}/8$; and if so, we decrement $\ell$, i.e. subtract one from $\ell$.  Whenever the update operation  $\textsc{RemoveTerminalVertex}(a)$ is invoked, we remove the vertex $a$ from $A$, and then check whether the size of $A$ is $\lfloor m^{1- \ell/\Lambda}/16  \rfloor$; and if so we increment $\ell$, i.e. we add one to $\ell$.

Finally, for every vertex $a \in A$, we apply to $\hat{G}$ the following additional sequence of updates whenever $a$ is added to $A$ via $\textsc{AddTerminalVertex}(a)$, $G_{\ell}$ is re-initialized in the hierarchy, or if the index $\ell$ changes: we first add an isolated vertex $a'$ to $\hat{G}$, then add an edge $(a, a')$ of length $1/m$ to $\hat{G}$, and finally remove first the edge $(a, a')$ and the vertex copy $a'$ again from $\hat{G}$. Note that here we are inserting an edge of length $1/m$ which is much smaller than allowed, however, by multiplying all edge weights in $\hat{G}$ by $m$, we only increase the edge weights polynomially, and maintain the same metric information (just shifted by factor $m$). Thus, we can allow for such updates wlog. 

This concludes the description of the algorithm.

\paragraph{Analysis.} We establish \Cref{thm:mainTheoremVSExt} via the following series of claims.

\begin{claim}
After processing every update, we have that $A \subseteq V(H) \subseteq V(G)$.
\end{claim}
\begin{proof}
We maintain $A$ as described and maintain $H$ to be the graph $G_{\ell}$ for some $\ell$. Note that since $G_{\ell}$ was last re-initialized, $\hat{G}$ was updated such that for every vertex $a \in A$, we added the isolated vertex copy $a'$ of $a$, connected $a$ and $a'$ by a length $1/m$ edge $e$ and then removed both the edge and the vertex copy. 

But from \Cref{lma:workhorseSparsifier}, we have that this sequence of operations ensures that $a$ is present in $G_{\ell}$. To see this, observe that when we apply the described update sequence, we have that $a$ and $a'$ are added to the set $A_1$ maintained by the data structure \Cref{lma:workhorseSparsifier} on $G_0$ to obtain $G_1$. Further, we claim that the graph $G_1$ must have the edge $e$ between these vertices with length at most $(\gamma_{\ell})^{O(K)} \cdot 1/m$ to preserve the distance between vertices $a$ and $a'$ in $G_1$.

Assume for the sake of contradiction that this is not the case. Then, by the upper bound on distances between vertices in $A_1$ in $G_1$ from Property \ref{prop:approxVS} in \Cref{lma:workhorseSparsifier}, we would have that there is a path $P$ between $a$ and $a'$ containing another vertex $v$ where $l_H(P) \leq (\gamma_{\ell})^{O(K)} \cdot 1/m < 1$. But this implies that $v$ is at distance less than $1$ from $a$. But since $\hat{G} \setminus \{a'\} = G$ and all edges in $G$ are in $[1, L]$, we have that the distance between $a$ and $v$ in $\hat{G}$ is at least $1$. But this contradicts the lower bound on distances in $G_1$ from Property \ref{prop:approxVS} in \Cref{lma:workhorseSparsifier}). 

Thus, $G_1$ is undergoing an edge insertion $e$, and we can show via the same argument where we use that ${(\gamma_{\ell})^{O(K)}}^{O(\Lambda)} \cdot 1/m < 1$ that $e$ is in fact added to every graph $G_i$ for any $0 \leq i \leq \Lambda$ since the last re-initialization of graph $G_{\ell}$ in the hierarchy.

But by \Cref{lma:workhorseSparsifier}, this implies that for all graphs that have not been re-initialized since this update, we have that $a$ is in its vertex set and thus, in particular, $a$ is in the vertex set of $G_{\ell}$ by assumption.  
\end{proof}

\begin{claim}\label{clm:changeEll}
At every time $t$ where the value $\ell$ is changed while processing the update, we have that at least $\Omega(|A|)$ terminal vertex insertions/removals have been observed since the last time $t' < t$ that $\ell$ was changed/ initialized.
\end{claim}
\begin{proof}
We prove by case analysis:
\begin{itemize}   
    \item \underline{if $\ell$ was decremented:} then we have that $A$ is of size $m^{1 - \ell/\Lambda}/8$ after the change of $\ell$. 

    We claim that the size of $A$ at time $t'$ was at most $m^{1 - \ell/\Lambda}/16$ which implies that there are at least $m^{1 - \ell/\Lambda}/16$ many terminal vertex insertion operations invoked since time $t'$. 

    Assume first that $t' = 0$, then we have that $A$ was of size $0$ at time $t'$. The claim follows immediately. Otherwise, $t' > 0$, we have that at time $t'$, $\ell$ was either decremented or incremented. If it was decremented, it was of size $m^{1 - (\ell+1)/\Lambda}/8 < m^{1 - \ell/\Lambda}/16$; if it was incremented it was of size $\lfloor m^{1 - \ell/\Lambda}/16 \rfloor \leq m^{1 - \ell/\Lambda}/16$.  

    \item \underline{if $\ell$ was incremented:} we have that after the change of the value of $\ell$, $A$ is of size $\lfloor m^{1- (\ell-1)/\Lambda}/16  \rfloor$. 
    
    Note that $\ell$ cannot exceed $\Lambda$ as $A$ never drops below size $\lfloor m^{1- \Lambda/\Lambda}/16  \rfloor = 0$. Thus, $t' > 0$.

    Again, if at time $t'$, $\ell$ was incremented, it was of size $\lfloor m^{1- (\ell-2)/\Lambda}/16  \rfloor$; if it was decremented it was of size $m^{1 - (\ell-1)/\Lambda}/8$. Thus, again in either case, there were at least $\Omega(|A|)$ terminal vertex removals to $A$ since time $t'$ that can be charged.
\end{itemize}
\end{proof}

\begin{claim}\label{clm:recourseVSStatement}
After the first $t$ updates to the data structure, the number of updates to $\hat{G}$ is $O(t)$. The number of updates to $H$ is $e^{O(\log^{20/21} m \log\log m)} \cdot t$ and at any stage, we have that the size of $H$ is $|A| \cdot e^{O(\log^{20/21} m \log\log m)}$.
\end{claim}
\begin{proof}
We have that every update to $G$ generates a single update to $\hat{G}$. Every addition of a terminal vertex $a$ to the set $A$ immediately generates $4$ updates to $\hat{G}$. We call these updates to $\hat{G}$ the \emph{immediate updates to $\hat{G}$}.

It remains to analyze the number of \emph{non-immediate updates to $\hat{G}$}. These updates are caused by $G_{\ell}$ being re-initialized, or $\ell$ being changed. In either case, the algorithm then generates an additional number of $4 \cdot |A|$ updates to $\hat{G}$.

We charge these updates using the following case analysis:
\begin{itemize}
    \item \underline{if $G_{\ell}$ is re-initialized and $\ell$ has not been changed since the last re-initialization of $G_{\ell}$:} since $G_{\ell}$ is re-initialized only when a level $j < \ell$ is rebuilt, and level $j$ is rebuilt only every $u_j$ updates to $G$, we have that in between two such re-initialization, we observe at least $u_{\ell-1}$ updates to $\hat{G}$. 

    But since $\ell$ was not changed, the only non-immediate updates to $\hat{G}$ caused since resulted from the re-initialization. But since $\ell$ was not changed, $A$ was of size at most $m^{1-\ell/\Lambda}/8$ at the time of the last re-initialization. Thus, only $4 \cdot m^{1-\ell/\Lambda}/8   
    = u_{\ell -1} /2$ non-immediate updates to $\hat{G}$ were observed since. 
    
    Thus, half of the updates to $\hat{G}$ are immediate updates. Thus, we can bound the number of non-immediate updates generated by $\hat{G}$ at the current time to at least $u_i/2$ immediate updates to $\hat{G}$ and thus to $\Omega(u_i)$ adversarial data structure operations.

    \item \underline{if $G_{\ell}$ is re-initialized and $\ell$ has been changed since the last re-initialization of $G_{\ell}$:} We then have that all updates to $\hat{G}$ can be charged to invocations of the operations $\textsc{AddTerminalVertex}(a)$/ $\textsc{RemoveTerminalVertex}(a)$ as follows: let $t'$ be the last time that $\ell$ was changed. Let $x$ be the size of the set $A$ at time $t'$. Then, we have that from \Cref{clm:changeEll}, we can charge $\Omega(x)$ data structure operations uniquely to the re-initialization of $G_{\ell}$. Further, if $A$ now has current size $y$, we have that there are at least $|y - x|$ many additional such operations since time $t'$ that can be charged.

    Since $G_{\ell}$ then generates only $4y$ updates to $\hat{G}$, we have that each data structure operation can be charged with $O(1)$ many of such updates.

    \item \underline{if $\ell$ is updates:} then we can charge $O(1)$ updates to $\hat{G}$ uniquely to each data structure operation by \Cref{clm:changeEll}.
\end{itemize}

We have that every update to $\hat{G}$ generates at most $e^{O(\log^{20/21} m \log\log m)}$ changes to the graph $G_{\ell}$. Whenever $\ell$ changes we can charge $\Omega(|A|)$ many data structure operations for this change, and we have from \Cref{clm:runtimeFulylDynAPSP} that $G_{\ell}$ has size at most $|A| \cdot e^{O(\log^{20/21} m \log\log m)}$.
\end{proof}

\begin{claim}
The algorithm is deterministic, and initially takes time $m \cdot e^{O(\log^{20/21} m \log\log m})$. Every update is processed in amortized time $e^{O(\log^{20/21} m \log\log m})$.
\end{claim}
\begin{proof}
The runtime of the initialization procedure follows straightforwardly from \Cref{clm:runtimeFulylDynAPSP}.

The update time follows from combining \Cref{clm:recourseVSStatement} and the analysis from \Cref{clm:runtimeFulylDynAPSP}, along with the simple observation that given the hierarchy, it is simple to maintain index $\ell$ and thus locate $G_{\ell}$, and we can then report all changes to $H$ in time linear in the recourse of $H$ which is bounded again by \Cref{clm:recourseVSStatement}.
\end{proof}

All other properties then follow straightforwardly from \Cref{lma:extensionOfMainThm} which extends \Cref{lma:workhorseSparsifier} and by choosing an appropriate value $\gamma_{vertexSparsifier}$.

\paragraph{Mapping Hierarchical Forests.} Finally, since $H$ is at all times equal to a graph $G_{\ell}$ in the hierarchy for some index $\ell$, we can recursively apply \Cref{thm:mainTheoremVSExt} (analogously to how it is applied in \Cref{subsec:maintainLowDiam}), to derive the statement in \Cref{thm:mainTheoremVSExt} about mapping a hierarchical forest $F$.
\section{Decremental Single-Source Shortest Paths via the Dynamic Shortest Path Framework}
\label{sec:decrSSSP}

In this section, we show how to use the toolbox created in \Cref{sec:sparsifier_hierarchy} to implement \Cref{thm:mainSSSPGeneral}. We build on the framework developed in \cite{gutenberg2020deterministic} that was further refined in \cite{bernstein2022deterministic}. While we do not require the refinements of \cite{bernstein2022deterministic} over \cite{gutenberg2020deterministic} (beyond a simple scaling technique), we present our algorithm in the framework given in \cite{bernstein2022deterministic} since it provides interfaces that are easier to adapt. 

We note that the theorem below that summarizes the technical result obtained in this section works only for bounded-degree graphs with small polynomial edge lengths. A simple reduction however suffices to show that general $m$-edge graphs with lengths in $[1, L]$ can be handled by invoking the theorem below on $\log_2 L$ decremental bounded-degree graphs with $\Otil(m)$-vertices and $\Otil(m)$ edges. The interested reader is referred to Proposition II.1.2, in \cite{bernstein2022deterministic}, for a formal proof.

\begin{theorem}\label{thm:mainSSSP}
Given an $n$-vertex bounded-degree graph $G$ with lengths in $[1, n^4]$ that undergoes a sequence of edge deletions, a dedicated source vertex $s \in V$ and an accuracy parameter $\eps = \Omega(1/\polylog m)$. Then, there is an algorithm that maintains a flat hierarchical forest $F$ over $G$ along with vertex maps $\Pi_{V(G) \mapsto V(F)}, \Pi_{V(F) \mapsto V(G)}$ and embedding $\Pi_{F \mapsto G}$ such that, for some $\gamma_{SSSP} = e^{O(\log^{83/84} m \log\log m)}$, at any time:
\begin{enumerate}
    \item for every $v \in V$, if $\dist_G(s,v) \leq n^5$, then $\Pi_{F \mapsto G}(\pi_F(\Pi_{V(G) \mapsto V(F)}(s), \Pi_{V(G) \mapsto V(F)}(v))) \leq (1+\eps) \dist_G(s,v)$, i.e. the path between the two nodes in $F$ that vertices $s$ and $v$ are mapped to has length at most $(1+\eps)\dist_G(s,v)$, and
    \item $\econg(\Pi_{F \mapsto G}) \leq \gamma_{SSSP}$.
\end{enumerate}
The algorithm maintains $F$ and the associated maps $\Pi_{V(G) \mapsto V(F)}, \Pi_{V(F) \mapsto V(G)}$ and $\Pi_{F \mapsto G}$ explicitly and the total number of changes to $F$ and these maps is at most $m \cdot \gamma_{SSSP}$. The algorithm runs in time $m \cdot \gamma_{SSSP}$.
\end{theorem}

\paragraph{Roadmap.} In the following, we first show in 
\Cref{subsec:maintainCovering} that we can use our new algorithmic toolbox to maintain coverings in graphs under deletions. Obtaining an efficient algorithm for coverings was previously the key challenge addressed in \cite{bernstein2022deterministic}, here, we draw on their techniques (which have been heavily inspired by \cite{gutenberg2020deterministic}). In \Cref{subsec:maintainSSSPviaCover}, we then explain some of the components from \cite{bernstein2022deterministic} that we can use almost directly and therefore give various blackbox statements. This can in fact already be used to obtain an algorithm for decremental SSSP with significantly better runtime. Finally, in \Cref{subsec:maintainSSSPAndEmbedding}, we demonstrate how to use our new algorithm to maintain coverings, the blackbox components from \cite{bernstein2022deterministic}, and our toolbox to obtain the algorithm described in \Cref{thm:mainSSSP}.

\subsection{Maintaining a Covering}
\label{subsec:maintainCovering}

We start by defining a covering as defined in \cite{bernstein2022deterministic} which is also commonly referred to as a sparse neighborhood cover.

\begin{definition}[Covering, compare to Definition II.2.6. in \cite{bernstein2022deterministic}]
\label{def:Covering}
Let $G=(V,E,l)$ be a decremental graph. A $(d,K,\eps,\stretchSSSP, \Delta)$-covering $\cC$ of $G$ is a dynamic collection of vertex sets called \emph{covering sets} (or also covers) where each cover $C\in\cC$ is added to $\cC$ at some time during the algorithm and is then associated with a fixed radius $r(C) \in [d, \gamma \cdot d \cdot \left(\frac{\stretchSSSP}{\eps}\right)^{K-1}]$ such that 
\begin{enumerate} 
    \item after being added to the collection $\cC$, each cover $C$ is decremental over time, i.e. after first being added to $\cC$ the data structure only deletes vertices from $C$, and
    \item at any time, for any $C \in \cC$, we have that $\diam_{G[C]}(C) \leq r(C)$, and
    \item over all times, for every vertex $v \in V(G)$, there are at most $\Delta$ many covers $C \in \cC$ such that $v \in \bar{B}_G(C, \frac{\stretchSSSP}{8\eps} \cdot r(C))$.
\end{enumerate}
\end{definition}

The main result of this section is summarized by the following theorem.

\begin{theorem}[Covering, compare to Theorem II.4.1 in \cite{bernstein2022deterministic}]\label{thm:covering}
Let $G$ be an $n$-vertex bounded-degree
decremental graph. Given parameters $d,K,\eps, \gamma_{stretchCover}$ where $\eps\le0.1$ and $\gamma_{stretchCover} \geq 2 \gamma_{approxAPSP}^2$. There is an algorithm that maintains a $(d,K,\eps,\gamma_{stretchCover},O(Kn^{2/K}))$-covering of $G$
in total update time $\Otil(Kn^{1+2/K} \cdot\gamma_{timeAPSP})$.
\end{theorem}

\paragraph{Maintaining Clustering, Covers and Cores.} In our algorithm, we focus on maintaining the set $\textsc{Clustering}$ whose elements are (possibly non-disjoint) vertex sets.
When we add a new vertex set to $\textsc{Clustering}$, we never change this set in $\textsc{Clustering}$ again.
We also never remove an element from $\textsc{Clustering}$, so this set of vertex sets is purely incremental.

We associate with each cluster $C \in \textsc{Clustering}$, a core $\Core(C)$, initially, the core is equal to the cluster, however, cores are decremental sets where we ensure that $\Core(C)$ always has small diameter in $G$. Further, we maintain for each cluster $C$, an associated $\cover(C)$. The cover of a cluster $C$ consists of the vertices in $\Core(C)$ and all vertices that are reasonably close to vertices in $\Core(C)$ with respect to the current graph $G$. We also take $\cover(C)$ to be decremental sets. 

Finally, we can take the covering $\cC$ to be the collection of all sets $\cover(C)$ for $C \in \textsc{Clustering}$. It remains to describe how to maintain $\textsc{Clustering}$ and the associated sets $\Core(C)$ and $\cover(C)$ for each cluster $C$.

\paragraph{Initializing the Covering.} To initialize the algorithm, we set the collection $\textsc{Clustering} = \emptyset$ (and thus implicitly $\cC = \emptyset$) and then invoke procedure $\textsc{CoverAllVertices}()$ implemented by \Cref{alg:coverVertices}.

The algorithm checks whether there is a vertex that is not yet covered by any of the clusters. If such a vertex $v \in V$ exists, we search for a small integer $\ell$ such that the ball around $v$ to radius $d_{\ell+1}$ is not by much larger than the ball to radius $d_{\ell}$ where we define $d_{\ell} \defeq d \cdot (\frac{\gamma_{stretchCover}}{\eps})^{\ell}$. We take the ball to radius $d_{\ell}$ to be a new cluster $C$ that is then added to the clustering $\textsc{Clustering}$. 

We take the core of the cluster $C$ to be the initial cluster and the cover of $C$ to be the cluster with some additional padding. We take $\ell$ to be the level of $C$, henceforth denoted by $\ell(C)$. Finally, we initialize data structures $\mathcal{X}_C$ and $\mathcal{Y}_C$ to monitor whether the core and cover of $C$ need to be changed. 

\begin{algorithm}
\tcc{While there exists a vertex $v\in V$ not covered by any core in $\textsc{Clustering}$.}
\tcc{We use $d_{\ell} \defeq d \cdot (\frac{\gamma_{stretchCover}}{\eps})^{\ell}$.}
\While(\label{lne:coverAlgoWhileCond}){there exists a vertex $v$, such that for all  $C\in\textsc{Clustering}, v \notin \cover(C)$}{
    Let $v$ be such a vertex and $\ell$ be the smallest integer with $|\bar{B}_{G}(v,d_{\ell+1})|\le n^{(\ell+1)/K}$.\label{lne:shellissmall}\\
    $C \gets \bar{B}_{G}(v, d_{\ell})$.\label{lne:takeC}\\
    Add $C$ to $\textsc{Clustering}$.\label{lne:addToCovering}\\
    \tcc{Initialize data structures associated with $C$.}
    $\Core(C) \gets \bar{B}_{G}(v, d_{\ell})$.\label{lne:takeCore}\\
    $\cover(C) \gets \bar{B}_{G}(v, 4d_{\ell})$.\label{lne:takeCover}\\
    $\ell(C) \gets \ell$.\label{lne:takeLevel}\\
    Initialize a dynamic APSP data structure $\mathcal{X}_C$, as described in \Cref{thm:mainTheoremAPSPFormal}, initialized on graph $X_C$ that is initially equal to $G[\cover(C)]$ with diameter set initialized to $\Core(C)$. \\
    Initialize a dynamic APSP data structure $\mathcal{Y}_C$, as described in \Cref{thm:mainTheoremAPSPFormal}, initialized on graph $Y_C$ that is initially equal to $G[\cover(C)] / \Core(C)$, 
 that is the graph $G[\cover(C)]$ where the vertices in $\Core(C)$ are contracted into a single vertex with diameter set initialized to $\cover(C)$.
}
\caption{$\textsc{CoverAllVertices}()$}
\label{alg:coverVertices}
\end{algorithm}

Once such an integer $\ell$ is found, the radius of the cluster is fixed to $r(C) = d_{\ell}$ and the cluster is then initialized to be the ball to radius $4 \cdot r(C)$ from the vertex $v$ that was previously not clustered. The algorithm finally starts an All-Pairs Shortest-Paths data structure $\mathcal{X}_C$ to monitor distance in the graph $G[\cover(C)]$. The algorithm returns when all vertices are clustered which yields the initial covering.

\paragraph{Updating the Covering.} 
We next discuss how updates to $G$ affect the covering $\cC$ through changes to the $\textsc{Clustering}$, cores, and cover sets.
The algorithm to maintain the covering $\cC$ under such updates is given in \Cref{alg:Covering}. The algorithm first forwards the edge deletion update to $G$ to all data structures $\mathcal{X}_C$ that are affected, i.e. that work on graphs that contain the edge deleted from $G$. The deletion might then lead to some cluster $C$ having larger diameter than stipulated by \Cref{thm:covering} as distances might increase. Therefore the algorithm queries data structure $\mathcal{X}_C$ and while it finds a pair of vertices at large distance, it peels off one of the vertices along with its ball of radius roughly $d_{\ell(C)}$ from the cluster $C$. After the first while-loop (the one starting in \Cref{lne:whileInnerSNC}) terminates, we thus have that all clusters satisfy the radius constraint again, however, some vertices might no longer be clustered due to them being removed in the previous step from a cluster.

\begin{algorithm}[!ht]
Update all data structures $\mathcal{X}_C$ by removing the edge deleted from $G$ during the $t$-th update to $G$ from every graph $X_C$.\label{lne:updateDCviaG}\\
\While(\label{lne:whileInnerSNC}){$\exists C \in \textsc{Clustering}$ so that $\mathcal{X}_C.\textsc{QueryDist}(s,t) > 8 \log_2(m)  \cdot \gamma_{ApproxAPSP} \cdot d_{\ell(C)}$ for $s, t \gets \mathcal{X}_C.\textsc{QueryDiameterWitnessPair}()$}{
    Let $s,t$ be the witness pair and $C \in \textsc{Clustering}$ for which the while-condition was satisfied.\\ 
    Find a vertex $u \in \{s,t\}$ and an integer $i$ with $0 < i \leq 2  \log_2(m) $ such that 
    $|\bar{B}_{X_C}(u, 2d_{\ell(C)} \cdot i) \cap \Core(C)| \leq 2 \cdot |\bar{B}_{X_C}(u, 2d_{\ell(C)} \cdot (i-1)) \cap \Core(C)| \leq |\Core(C)|/2$ and $\deg_{X_C}(\bar{B}_{X_C}(u, 2d_{\ell(C)} \cdot i)) \leq 2 \cdot \deg_{X_C}(\bar{B}_{X_C}(u, 2d_{\ell(C)} \cdot (i-1)))$.\label{lne:findIntegerBallGrowing}\\
    Remove all vertices in $\bar{B}_{X_C}(u, 2d_{\ell(C)} \cdot i) \cap \Core(C)$ from $\Core(C)$ and the diameter set of $\mathcal{X}_C$.\label{lne:updateCore} \\
    Remove all edges and vertices incident to ball $\bar{B}_{X_C}(u, 2d_{\ell(C)} \cdot (i-1))$ from graph $X_C$ and thus from data structure $\mathcal{X}_C$. \label{lne:updateGraph}
}
Update all data structures $\mathcal{Y}_C$ such that the graph $Y_C$ is equal to the graph $G[\cover(C)] / \Core(C)$, i.e. reflects the changes to $G$ and $\Core(C)$.\\
\While(\label{lne:whileInnerSNC2}){$\exists C \in \textsc{Clustering}$ so that $\mathcal{Y}_C.\textsc{QueryDist}(x,y) > 2 \cdot \gamma_{ApproxAPSP}^2 \cdot d_{\ell(C)}$ 
for $x,y \gets \mathcal{Y}_C.\textsc{QueryDiameterWitnessPair}()$}{
    \If{$\mathcal{Y}_C.\textsc{QueryDist}(\Core(C), x) > \gamma_{ApproxAPSP} \cdot d_{\ell(C)}$
    }{
        Remove vertex $x$ from set $\cover(C)$ and the diameter set of $\mathcal{Y}_C$.\label{lne:removeClsuterVertexX}
    }\Else{
        Remove vertex $y$ from set $\cover(C)$ and the diameter set of $\mathcal{Y}_C$.\label{lne:removeClsuterVertexY}
    }
}

$\textsc{CoverAllVertices}()$.
\caption{$\textsc{UpdateCovering}()$ 
\label{alg:Covering}}
\end{algorithm}

To deal with this issue, the algorithm finally invokes procedure $\textsc{CoverAllVertices}()$ which ensures that all vertices are clustered properly after the algorithm terminates.

\paragraph{Analysis.} We establish \Cref{thm:covering} by proving the following series of claims.

\begin{claim}\label{clm:analyzeBallGrowing}
Whenever the algorithm enters \Cref{lne:findIntegerBallGrowing}, there is an algorithm that finds a vertex $u$ and an integer $0 < i < 2  \log_2(m)$ that satisfy the requirements given in  \Cref{lne:findIntegerBallGrowing} and returns $u$ and $i$ in time $\Otil(|E_{X_C}(\bar{B}_{X_C}(u, 2d_{\ell(C)} \cdot (i-1)))|)$.
\end{claim}
\begin{proof}
We first observe that there is a vertex $u \in \{s,t\}$ such that $|\bar{B}_{X_C}(u, 2d_{\ell(C)} \cdot 2  \log_2(m) ) \cap \Core(C)| \leq |\Core(C)|/2$. This observation follows since we have from the while-loop condition that $\widehat{\dist}(s,t) > 8  \log_2(m)  \cdot \gamma_{ApproxAPSP} \cdot d_{\ell(C)}$ which implies that $\dist_{X_C}(s,t) \geq \widehat{\dist}(s,t) / \gamma_{approxAPSP}$ by \Cref{thm:mainTheoremAPSPFormal} which implies that $\dist_{X_C}(s,t) > 4  \log_2(m)  \cdot 2d_{\ell(C)}$. But this implies that the balls $\bar{B}_{X_C}(s, 2  \log_2(m)  \cdot 2d_{\ell(C)})$ and $\bar{B}_{X_C}(t, 2  \log_2(m)  \cdot 2d_{\ell(C)})$ are vertex-disjoint, and thus either the ball of $s$ or $t$ contains at most half the vertices in $\Core(C)$.

Next, let $u$ be as above, we show that there exists an index $i$ that satisfies the requirements. Assume for the sake of contradiction that no such index $i$ exists. If $u$ is an isolated vertex, then the proof is trivial since for $i = 1$, we certainly have this property. Otherwise we have $\deg(\bar{B}_{X_C}(u, 0)) \geq 1$. Further, we have that $u \in \Core(C)$ since we maintain the diameter set of $\mathcal{X}_C$ to be the set of vertices in $\Core(C)$. Thus, $|\bar{B}_{X_C}(u, 0) \cap \Core(C)| \geq 1$. But since no index $i$ satisfies the requirements, we have
\begin{equation}\label{eq:ballGrowCore}
|\bar{B}_{X_C}(u, 2d_{\ell(C)}  i) \cap \Core(C)| > 2 \cdot |\bar{B}_{X_C}(u, 2d_{\ell(C)}  (i-1)) \cap \Core(C)| 
\end{equation}
or 
\begin{equation}\label{eq:ballGrowEdges}
\deg_{X_C}(\bar{B}_{X_C}(u, 2d_{\ell(C)} \cdot i)) > 2 \cdot \deg_{X_C}(\bar{B}_{X_C}(u, 2d_{\ell(C)} \cdot (i-1))). 
\end{equation}

Thus, there are at least $ \log_2(m) $ indices $i$ for which either \eqref{eq:ballGrowCore} or \eqref{eq:ballGrowEdges} holds. But in the former case, it is not hard to show by induction on indices $i$ that this implies $|\bar{B}_{X_C}(u, 2d_{\ell(C)} \cdot 2  \log_2(m) ) \cap \Core(C)| > 2^{ \log_2(m) } \geq m$ which yields a contradiction since there are at most $m$ vertices in $X_C \subseteq G$ by assumption; and otherwise, we have that $\deg_{X_C}(\bar{B}_{X_C}(u, 2d_{\ell(C)} \cdot 2  \log_2(m) )) > m$ which contradicts the bound on the sums of degrees that we derived earlier for $u$.

We have now established existence of $u$ and $i$ that satisfy the requirements. To compute such $u$ and $i$, one can run Dijkstra's algorithm from $s$ and $t$ in parallel where after relaxing all vertices at distance at most $2d_{\ell(C)} \cdot i$ one can evaluate whether the index $i$ satisfies the requirement. Once the first Dijkstra algorithm finds such an index $i$, with respect to the ball of $u \in \{s,t\}$, both procedures are aborted and the algorithm returns $u$ and $i$. The runtime analysis is straightforward from the fact that Dijkstra's algorithm run from vertex $u$ relaxes all vertices at distance at most $2d_{\ell(C)} \cdot i$ in time $\Otil(|E_{X_C}(\bar{B}_{X_C}(u, 2d_{\ell(C)} \cdot i)|)$ where we have from the fact that $i$ stipulates \eqref{eq:ballGrowEdges} that $\deg_{X_C}(\bar{B}_{X_C}(u, 2d_{\ell(C)} \cdot i)) > 2 \cdot \deg_{X_C}(\bar{B}_{X_C}(u, 2d_{\ell(C)} \cdot (i-1)))$ and we have that since the procedures are run in parallel that both have spent the same amount of time up until this point of the algorithm.
\end{proof}

\begin{claim}\label{clm:verticesAreNeverInTheSameCoreAgain}
For any cluster $C \in \textsc{Clustering}$, consider a time when the set $\Core(C)$ is updated in \Cref{lne:updateCore} in \Cref{alg:Covering}. We denote by $u$ the vertex that was chosen by the algorithm in \Cref{lne:findIntegerBallGrowing} while updating $\Core(C)$, by $\Core^{OLD}(C)$ the core of $C$ before the update, by $\Core^{NEW}(C)$ the core thereafter, and by $X_C^{OLD}$ and $X_C^{NEW}$ the graph $X_C$ before and after the while-loop iteration. Then, we have for every vertex $v \in \bar{B}_{X^{OLD}_C}(u, 2d_{\ell(C)} \cdot (i-1)) \cap \Core^{OLD}(C)$ that $\dist_G(v, \Core^{NEW}(C)) > 2 \cdot d_{\ell(C)}$.
\end{claim}
\begin{proof}Let $v \in \bar{B}_{X^{OLD}_C}(u, 2d_{\ell(C)} \cdot (i-1)) \cap \Core^{OLD}(C)$ as defined above, and let $a$ be any vertex in $\Core^{NEW}(C)$. We prove that there is no $va$-path $P$ in the current graph $G$ of length at most $2 \cdot d_{\ell}(C)$. This then establishes the claim since the choice of $a$ is arbitrary among all vertices in $\Core^{NEW}(C)$.

We prove the claim by contradiction. Assume there is such a path $P$ of length at most $2d_{\ell}$. Note that every vertex on $P$ is at distance at most $d_{\ell}$ from either $v$ or $a$. Since both $a$ and $v$ are in the initial core $\Core(C)$ when $C$ was added to $\textsc{Clustering}$, we have that the initial $\cover(C)$ must have contained all vertices on $P$ since it consists of all vertices at distance at most $3d_{\ell}$ from at least one vertex in the initial core. Since initially $X_C = G[\cover(C)]$, we thus have that $P$ was contained in the initial graph $X_C$.

Further, we have that whenever we remove edges from $X_C$ this occurs either in \Cref{lne:updateDCviaG} of \Cref{alg:Covering} when we delete edges that are deleted from $G$, but this cannot delete an edge from $P$ since we claimed that it is still in $G$. Or it occurs in \Cref{lne:updateGraph} of \Cref{alg:Covering}. But note that whenever we remove edges from $X_C$ in this scenario, before we remove any edge from $X_C$, we remove every vertex that is at a distance less than $2d_{\ell}$ from such an edge from the core $\Core(C)$ (see \Cref{lne:updateCore} and \Cref{lne:updateGraph}). Thus, if the path $P$ is not in $X^{OLD}_C$, then we must have that either $v$ or $a$ where already removed from $\Core(C)$ at an earlier iteration of the while-loop thus contradicting that $v,a \in \Core^{OLD}(C)$.
\end{proof}

A similar argument establishes the following claim about the maintenance of the cluster set.

\begin{claim}\label{clm:whenRemovedFromClusterItIsFar}
For any cluster $C$, we have that whenever we remove a vertex $z$ from $\cover(C)$ in \Cref{lne:removeClsuterVertexX} or \Cref{lne:removeClsuterVertexY}, we have that $\dist_G(z, \Core(C)) > d_{\ell(C)}$.
\end{claim}
\begin{proof}
We have from the while-loop condition that in every iteration that picks such $x$ and $y$ in cluster $C$, that $\mathcal{Y}_C.\textsc{QueryDist}(x,y) = \widehat{\dist}(x,y) > 2 \cdot \gamma_{ApproxAPSP}^2 \cdot d_{\ell(C)}$. From \Cref{thm:mainTheoremAPSPFormal}, we have that this implies that $\dist_{Y_C}(x,y) \geq \widehat{\dist}_{Y_C}(x,y)/\gamma_{approxAPSP} > 2 \cdot \gamma_{approxAPSP} \cdot d_{\ell(C)}$. We have by the triangle inequality that $\dist_{Y_C}(x, \Core(C)) + \dist_{Y_C}(\Core(C), y) \geq \dist_{Y_C}(x,y)$ and thus $\min\{\dist_{Y_C}(x, \Core(C)), \dist_{Y_C}(\Core(C), y)\} \geq \frac{1}{2} \cdot \dist_{Y_C}(x,y) > \gamma_{approxAPSP} \cdot d_{\ell(C)}$. 

Now, if the if-condition holds, we have that $\dist_{Y_C}(\Core(C), x) \geq \widehat{\dist}(\Core(C), x) / \gamma_{ApproxAPSP} \geq d_{\ell(C)}$ by \Cref{thm:mainTheoremAPSPFormal}. In the else-case, we have that $\dist_{Y_C}(\Core(C), x) \leq \widehat{\dist}(\Core(C), x) \leq \gamma_{approxAPSP} \cdot d_{\ell(C)}$, and thus $\dist_{Y_C}(\Core(C), y) > \gamma_{approxAPSP} \cdot d_{\ell(C)}$.

Using these lower bounds on the distance in either case and arguing along the same line of reasoning as in the proof of \Cref{clm:verticesAreNeverInTheSameCoreAgain}, we can thus establish that in either case, the vertex removed from $C$ is at distance at least $d_{\ell(C)}$ not only in $Y_C$ but, in fact, also in $G$.
\end{proof}

The above claim can next be used to derive the following useful claim on the disjointness of cores of clusters.

\begin{claim}\label{clm:disjointCores}
At any time, for any vertex $w \in V$ and $\ell \in [0, K-1]$, there is at most one cluster $C \in \textsc{Clustering}$  with level $\ell(C) = \ell$ and  $w \in \Core(C)$.
\end{claim}
\begin{proof}
Assume for the sake of contradiction that there are two distinct clusters $C, C'$ are in $\textsc{Clustering}$ with $\ell = \ell(C) = \ell(C')$ and $w \in \Core(C), \Core(C')$. Let $C$ be the cluster that was first added to $\textsc{Clustering}$, i.e. that was added before $C'$ was added and let us focus on the time when cluster $C'$ was added to $\textsc{Clustering}$. Let $v$ be the vertex that is chosen in \Cref{lne:shellissmall} of \Cref{alg:coverVertices} when $C'$ was created, i.e. initially, $\Core(C') = \bar{B}_G(v, d_{\ell})$. Since each core set is a monotonically decreasing set we also have that $w$ is in the initial set $\Core(C')$. Thus, we have $\dist_G(v, w) \leq d_{\ell}$. 

But we show that this yields a contradiction because we can also derive from these facts that $w \in \cover(C)$ at the time that $C'$ was created which contradicts the while-loop condition in \Cref{lne:coverAlgoWhileCond} in \Cref{alg:coverVertices}.

To see this last claim, observe that since the distance from $v$ to $w$ is at most $d_{\ell}$ at the time that $C'$ is added, it must have been at most $d_{\ell}$ at all previous times. But since $v$ was added to $\Core(C)$ when $C$ was added to $\cC$, we had $w$ added to the initial set $\cover(C)$, but then by  \Cref{clm:whenRemovedFromClusterItIsFar} it has remained in the cover ever since as its distance to $v$ and thus $\Core(C)$ did not increase to more than $d_{\ell}$. 
\end{proof}

We next prove that every vertex $v$ is only ever in the proximity of a few clusters. This claim is the main technical claim of this section and almost immediately yields the proof of \Cref{thm:covering}.

\begin{claim}\label{clm:fewVerticesInPromity}
Over all times, for every vertex $w \in V(G)$, there are at most $O(K n^{2/K})$ many clusters $C \in \textsc{Clustering}$ such that at the time that $C$ is added to $\textsc{Clustering}$, we have $w \in \bar{B}_G(C, d_{\ell(C) + 1}/8)$.
\end{claim}
\begin{proof}
Let us fix any level $\ell \in [0, K-1]$. For vertex $w$, we have that while $|\bar{B}_G(w, d_{\ell+1}/2)| > n^{(\ell+1)/K}$, we have that every cluster $C$ that is added to $\textsc{Clustering}$ in a while-loop iteration starting in \Cref{lne:coverAlgoWhileCond} in \Cref{alg:coverVertices} that picks vertex $v$ and level $\ell(C)$ in \Cref{lne:shellissmall} and has $w \in \bar{B}_G(C, d_{\ell(C) + 1}/8)$ has $\ell(C) \neq \ell$ since $\bar{B}_G(w, d_{\ell+1}/2) \subseteq \bar{B}_G(v, d_{\ell+1})$ and thus $|\bar{B}_G(v, d_{\ell +1})| > n^{(\ell+1)/K}$. 

Let us therefore focus on the clusters added to $\textsc{Clustering}$ from the first time where 
\[
|\bar{B}_G(w, d_{\ell+1}/2)| \leq n^{(\ell+1)/K}.
\]
Let us define the following potential function 
\[
\Phi_{\ell}(w) = \sum_{z \in \bar{B}_G(w, d_{\ell+1}/4)} \min\{| \bar{B}_G(z, 2 \cdot d_{\ell(C)})|, n^{(\ell + 1)/K}\}
\] 
which is monotonically decreasing over time since $G$ is decremental and where we have that initially $\Phi_{\ell}(w) \leq |\bar{B}_G(w, d_{\ell+1}/2)|^2 \leq n^{2(\ell+1)/K}$.

We show that for every cluster $C$ added to $\textsc{Clustering}$ with $\ell(C) = \ell$ and $w \in \bar{B}_G(C, d_{\ell(C) + 1}/8)$, it either has at the current time still a core $\Core(C)$ of size at least $\frac{1}{2}n^{\ell/K}$ or we can uniquely charge a drop by $\Omega(n^{2\ell/K})$ units of the potential $\Phi_{\ell}(w)$ to cluster $C$. Since each core $\Core(C)$ of a cluster $C$ under consideration has $\Core(C) \subseteq \bar{B}_G(w, d_{\ell+1}/2)$, and since each vertex can be in at most one core at level $\ell$ by  \Cref{clm:disjointCores} at any time, we can bound the number of such clusters with cores of size at least $\frac{1}{2}n^{\ell/K}$ by $2n^{1/K}$. Further, since $\Phi_{\ell}(w)$ is initially of size at most $n^{2(\ell+1)/K}$ and we charge for each other cluster $C$ at least $\Omega(n^{2\ell/K})$ units, and since $\Phi_{\ell}(w)$ remains non-negative, we can also bound the number of such clusters of the second type by $O(n^{2/K})$. This yields the claim.

It remains to show that we can charge each cluster $C \in \textsc{Clustering}$ with $\ell(C) = \ell$ that on initialization had $w$ contained in $\bar{B}_G(C, d_{\ell(C) + 1}/8)$ and that at the current time has $\Core(C)$ of size less than $\frac{1}{2} n^{\ell/K}$. To this end, observe that every core set $\Core(C)$ is initially of size at least $n^{\ell/K}$ by minimality of $\ell$ in \Cref{lne:shellissmall}. Next, consider the times while $\Core(C)$ was still of size at least $\frac{1}{2}n^{\ell/K}$ and some vertex set was removed from $\Core(C)$ in  \Cref{lne:updateCore}. Let $u$ and $i$ be the vertex and index selected in \Cref{lne:findIntegerBallGrowing} and let $\bar{B}_{X_C}(u, d_{\ell(C)} \cdot i) \cap \Core(C)$ be the set of vertices removed from $\Core(C)$. Then, we have that by \Cref{clm:verticesAreNeverInTheSameCoreAgain} that all vertices in $\bar{B}_{X_C}(u, d_{\ell(C)} \cdot (i-1)) \cap \Core(C)$ are at distance at least $2d_{\ell}$ from the vertices that remain in $\Core(C)$ and thus from at least $\frac{1}{4} n^{\ell/K}$ vertices in $\Core(C)$.  And since for each vertex $z \in \Core(C)$, we have that when $C$ was added that $\Core(C) \subseteq \bar{B}_G(z, 2 \cdot d_{\ell(C)})$, we have for each such vertex $z \in \bar{B}_{X_C}(u, d_{\ell(C)} \cdot (i-1)) \cap \Core(C)$ that while it was in the core of $C$, its ball $\bar{B}_G(z, 2 \cdot d_{\ell(C)})$ decreased in size by at least $\frac{1}{4} n^{\ell/K}$. 

But since by choice of $u$ and $i$, we have $|\bar{B}_{X_C}(u, d_{\ell(C)} \cdot i) \cap \Core(C)| \leq 2|\bar{B}_{X_C}(u, d_{\ell(C)} \cdot (i-1)) \cap \Core(C)|$, we have that at least half the vertices that leave $\Core(C)$ have this property. And thus, we have that once $\Core(C)$ has size less than $\frac{1}{2}n^{\ell/K}$, we can charge at least $\frac{1}{4}n^{\ell/K}$ vertices $z$ whose balls decreased by size at least $\frac{1}{4}n^{\ell/K}$ while being in the core of $C$, and we can thus uniquely charge $C$ with $\frac{1}{16}n^{2\ell/K}$ units of the potential, as desired.
\end{proof}

We can now establish that the covering $\cC$ is maintained correctly.

\begin{claim}\label{clm:correctnessCovering}
The algorithm correctly maintains a $(d,K,\eps,\gamma_{stretchCover},O(Kn^{2/K}))$-covering $\cC$.
\end{claim}
\begin{proof}
For each cluster in $\textsc{Clustering}$, we let $r(C) = 2\gamma_{approxAPSP}^2 \cdot d_{\ell(C)}$ whenever we initialize  cluster $C$ with associated cover set $\cover(C)$ and let $r(\cover(C)) = r(C)$.

It is straightforward to verify that the first property in \Cref{def:Covering} is satisfied by the algorithm. To see that for every cover $\cover(C)$ where $C$ is the cluster in $\textsc{Clustering}$, we have $\diam_{G[\cover(C)]}(C) \leq r(C)$, it suffices to inspect the intialization procedure $\textsc{CoverAllVertices}()$ and the while-loop starting in \Cref{lne:whileInnerSNC2} in procedure $\textsc{UpdateCovering}(\cdot)$ which removes vertices from $\cover(C)$ after each update to $G$ until $Y_C = G[\cover(C)] / \Core(C)$ has strong diameter at most $2\gamma_{approxAPSP}^2 \cdot d_{\ell}(C) \leq r(\cover(C))$.  The final property follows immediately from \Cref{clm:fewVerticesInPromity}.
\end{proof}

It remains to establish the following claim to bound the runtime.

\begin{claim}
The runtime of the algorithm to maintain covering $\cC$ is at most $O(K \cdot n^{1+2/K} \cdot \gamma_{timeAPSP})$. 
\end{claim}
\begin{proof}
We have that every vertex $v \in V$, by \Cref{clm:fewVerticesInPromity}, only participates in at most $O(K \cdot n^{2/K})$ many graphs $X_C$ and $Y_C$. Further for each data structure $\mathcal{X}_C$, the underlying graph is decremental and so it the diameter set. For each data structure $\mathcal{Y}_C$, we additionally have that the vertices in $\Core(C)$ are contracted and vertices are leaving the core over time and are then added to the graph with their incident edges. But since $G$ is a constant-degree graph, only a constant number of operations to $\mathcal{Y}_C$ suffices to remove a vertex from $\Core(C)$. 

Thus, the total runtime required by all data structures $\mathcal{X}_C$ and $\mathcal{Y}_C$ is $O(K \cdot n^{1+2/K} \cdot \gamma_{timeAPSP})$ by \Cref{thm:mainTheoremAPSPFormal}. 

All other operations of the algorithm can be subsumed by the runtime of these data structures (for the ball growing procedure in \Cref{lne:findIntegerBallGrowing}
 of \Cref{alg:Covering}, we spend by \Cref{clm:analyzeBallGrowing} time almost-linear in the number of edges deleted from $X_C$). 
\end{proof}

\subsection{Maintaining Shortest Paths via Coverings}
\label{subsec:maintainSSSPviaCover}

Given a covering $\cC$, it is rather straightforward to find a hopset for graph $G$. We will not be concerned with the exact guarantees of the hopset that we are creating and rather just create a hopset as it was used in \cite{bernstein2022deterministic}. 

\begin{definition}[Hopsets, see Definitions II.2.8 and  II.5.2 in \cite{bernstein2022deterministic}]\label{def:Htil}
Given a decremental graph $G = (V,E,l)$, depth parameters $d \geq 1$ and approximation parameter $1/\polylog(n) \leq \epsilon < 1$, and let there be a $(d,K,\frac{\eps}{2000 \log(n)},\gamma,\Delta)$-covering $\cC$ of $G$ being explicitly maintained. 

We say that the hopset $\Htil$ induced by covering $\cC$ is the fully-dynamic graph $\Htil = (\tilde{V}, \tilde{E}, \tilde{l})$ where $\tilde{V} = (V \cup \cC)$ and the edge set $\tilde{E}$ consists of the following edges:
\begin{itemize}
    \item for every $C \in \cC$ with vertex set $V^{init}_C$ being equal to the vertices in the ball $B_G(C, \frac{\gamma_{stretchCover} \cdot 250 \log n}{\eps} \cdot r(C))$ at the time when $C$ is added to $\cC_i$, we have for every vertex $v \in V^{init}_C \cup C$, an edge $e =(v,C)\in \tilde{E}$ of length $
    \tilde{l}(e) \defeq \gamma\cdot r(C) + \dist_G(v,C)$.
\end{itemize}
We say that $\hat{H}$ is a $(1+\eps/250)$-approximate hopset induced by $\cC$ if it is a graph over the same vertex and edge set as $\Htil$ but has lengths that are up to a factor $(1+\eps/250)$ larger than in $\Htil$ on every edge.
\end{definition}

The main data structure in \cite{bernstein2022deterministic} then internally maintains such $(1+\eps)$-approximate hopsets. 

\begin{theorem}[Hopset Maintenance, compare to Proofs of Proposition II.2.3 and Theorem II.5.1 in \cite{bernstein2022deterministic}]
\label{thm:hopsetMaintenance}
Given an $n$-vertex bounded-degree decremental graph $G=(V,E,l)$ and accuracy parameter $\eps\ge0$. And let there be coverings $\cC_0, \cC_1, \ldots, \cC_{5 \Lambda}$ maintained explicitly such that for every $0 \leq i \leq 5\Lambda$, $\cC_i$ is a $(n^{i/\Lambda},K,\frac{\eps}{2000 \log(n)},\gamma,\Delta)$-covering of $G$ where we require that $n^{1/\Lambda} \geq \left(\frac{\stretchSSSP}{\eps}\right)^K$.

Then, there is an algorithm that maintains an $(1+\eps/250)$-approximate hopset $\hat{H}_i$ induced by covering $\cC_i$ for every $0 \leq i \leq 5 \Lambda$ with total update time $\Otil(m \Delta n^{O(1/\Lambda)})$.
\end{theorem}

Next, we define an approximate ball data structure. Our definition is extremely close to the corresponding definition in \cite{bernstein2022deterministic}. However, \cite{bernstein2022deterministic} tailored their definitions to distance maintenance only, and then in a later section showed how to maintain shortest paths. By focusing on distances, \cite{bernstein2022deterministic} obtained a slightly slicker interface that omits the existence of the hopset and a forest certifying the distances. But both hopset and forest are crucial components of their internal data structures and here we make them explicit for latter purposes.

\begin{definition}[Approximate Ball, see Definitions II.2.1 in \cite{bernstein2022deterministic}]
\label{def:Apxball}
An \emph{approximate ball} data structure $\Apxball(G,S,D, \Lambda ,K,\eps,\gamma,\Delta, \{\cC_{0}, \cC_{1}, \ldots, \cC_{5\Lambda}\}, \{ \hat{H}_0, \hat{H}_1, \ldots, \hat{H}_{5\Lambda}\})$
is given a decremental graph $G=(V,E,l)$, a decremental source set $S\subseteq V$, a depth parameter $D$, an accuracy parameter $\eps\ge0$, and a $(n^{i/\Lambda},K,\frac{\eps}{2000 \log(n)},\gamma,\Delta)$-covering $\cC_i$ of $G$ along with an $(1+\eps/250)$-approximate hopset $\hat{H}_i$ induced by covering $\cC_i$ for every $0 \leq i \leq 5 \Lambda$. We define $\hat{H}_{\leq i} \defeq \cup_{j \leq i} \hat{H}_j$ and the (static) vertex set $V^{init}$ to be the set of vertices in the initial ball $\ball_{G}(S,D)$. We let $E_S$ be the dynamic edge set consisting of edges in $G[V^{init}]$ and edges from $\hat{H}_{\leq \lfloor \log_{n^{1/\Lambda}}(D) \rfloor}$ that are incident to a vertex in $V^{init}$. We let $G_S$ denote the graph $G \cup \hat{H}_{\leq \lfloor \log_{n^{1/\Lambda}}(D) \rfloor}$ induced by the edge set $E_S$.

Then, the data structure explicitly maintains a forest $F$ on graph $G_S$  such that, for every vertex $v \in V(G)$, we have
\begin{enumerate}
\item \label{enu:Apxball:overestimate} $\dist_G(S, v) \leq\dist_F(S, v) $,
\item \label{enu:Apxball:approx} if $v\in \bar{B}_{G}(S,D)$, then $\dist_F(S, v) \leq (1+\epsilon)\dist_{G}(S,v)$.
\end{enumerate}
\end{definition}

Finally, we observe the following Theorem from \cite{bernstein2022deterministic} that can be derived rather straightforwardly by inspecting the proof of Theorem II.5.1 in \cite{bernstein2022deterministic}. In fact, proving the theorem below is easier since \cite{bernstein2022deterministic} used a delicate inductive proof where coverings and approximate balls up to certain depths are used to build one another while for us, this is not necessary since we can maintain the coverings from our new APSP data structure (see \Cref{thm:covering}) which streamlines the proof.

\begin{theorem}[compare to Proposition II.2.3 and Theorem II.5.1 in \cite{bernstein2022deterministic}]
\label{thm:Apxball}
Given an $n$-vertex bounded-degree decremental graph $G=(V,E,l)$ with lengths in $[1, n^4]$, a decremental source set $S\subseteq V$, a depth parameter $D$, and accuracy parameter $\eps\ge0$. And let there be coverings $\cC_0, \cC_1, \ldots, \cC_{5 \Lambda}$ along with an $(1+\eps)$-approximate hopsets $\{ \hat{H}_0, \hat{H}_1, \ldots, \hat{H}_{5\Lambda}\}$ where $\hat{H}_i$ is induced by covering $\cC_i$ maintained explicitly such that for every $0 \leq i \leq 5\Lambda$, $\cC_i$ is a $(n^{i/\Lambda},K,\frac{\eps}{2000 \log(n)},\gamma,\Delta)$-covering of $G$. Here, we require that $n^{1/\Lambda} > \left(\frac{\stretchSSSP \cdot 2000 \log(n)}{\eps}\right)^K$.

Then, there is an algorithm that implements an approximate ball data structure as defined in \Cref{def:Apxball} denoted
$\Apxball(G, S,D,\Lambda,K,\eps,\gamma,\Delta, \{\cC_{0}, \cC_{1}, \ldots, \cC_{5\Lambda}\})$ with total update time $\Otil(\left|\ball_{G}(S,D)\right|\Delta n^{O(1/\Lambda)})$.
\end{theorem}

We point out that the reason that we consider $5\Lambda + 1$ coverings is that we have that $n^{5/5\Lambda} = n^5$ is the largest distance in $G$ since we have by assumption that the largest edge length is $n^4$ and each path consist of at most $n$ edges. 

While we could use the above Theorem directly to obtain a $(1+\eps)$-approximate SSSP data structure, we defer the proof to the next section where we also show how to maintain a approximate SSSP tree $T$ that certifies these distance and has a low congestion embedding into $G$.

\subsection{Maintaining a Single-Source Shortest Path Tree with Embeddings of Low Congestion}
\label{subsec:maintainSSSPAndEmbedding}

\paragraph{Maintaining Coverings and Approximate Ball Data Structures.} We define $\gamma_{stretchCover} = \max\{ \gamma_{approxAPSP}, \gamma_{lowDiamTree}, e^{\log^{20/21}m\log\log m}\}$ which implies $\gamma_{stretchCover} = e^{\Theta(\log^{20/21}m\log\log m)}$, and define $K \defeq \log^{1/42}(n)$ and $\Lambda \defeq \log^{1/84}(n)$. By this choice, we have that $\left(\frac{\gamma_{stretchCover} \cdot 2000 \log n}{\eps}\right)^K = e^{O(\log^{41/42} m \log\log m)}$ and thus for reasonably large $n$, we have $n^{1/\Lambda} > \left(\frac{\gamma_{stretchCover} \cdot 2000 \log n}{\eps}\right)^K$.

We maintain for every $0 \leq i \leq 5 \Lambda$, a $(n^{i/\Lambda},K,\frac{\eps}{2000 \log(n)},\gamma_{stretchCover},\Delta = O(Kn^{2/K}))$-covering $\cC_i$ of $G$ using the algorithm from \Cref{thm:covering} and a $(1+\eps/250)$-approximate hopset $\hat{H}_i$ induced by covering $\cC_i$ for every $0 \leq i \leq 5 \Lambda$ via the data structure from \Cref{thm:hopsetMaintenance}.

We further associate with each cover set $C$ in $\cC_i$, an approximate ball data structure \\ $\Apxball(G, C, \frac{\gamma_{stretchCover} 250 \log n}{\eps} \cdot r(C),\Lambda,K,\eps,\gamma_{stretchCover},\Delta, \{\cC_{0}, \cC_{1}, \ldots, \cC_{5\Lambda}\}, \{\hat{H}_0, \hat{H}_1, \ldots, \hat{H}_{5\Lambda}\})$  as described in \Cref{thm:Apxball} and let the forest maintained by this data structure be denoted by $Y_C$. Note that by our choice of $K$ and $\Lambda$, we have from \Cref{thm:Apxball} that $Y_C$ is a forest with edges in the $G \cup \hat{H}_{<i}$.

We further maintain a flat hierarchical forest $F'_C$ that preserves the diameter on the graph $G[C]$ as described in \Cref{thm:mainTheoremLowDiamTree}. We denote by $\Pi_{V(C) \mapsto V(F'_C)}, \Pi_{V(F'_C) \mapsto V(C)}$ and $\Pi_{F'_C \mapsto G[C]}$ the corresponding vertex maps and graph embeddings.

Finally, for our dedicated source vertex $s \in V$, we run an approximate ball data structure $\Apxball(G, \{s\},\infty,\Lambda,K,\eps,\gamma_{stretchCover},\Delta, \{\cC_{0}, \cC_{1}, \ldots, \cC_{5\Lambda}\},  \{\hat{H}_0, \hat{H}_1, \ldots, \hat{H}_{5\Lambda}\})$  as described in \Cref{thm:Apxball} and let the forest maintained by this data structure be denoted by $F_s$.

\paragraph{Approximate Shortest Path Forest via a Forest Hierarchy.} Finally, we describe how to maintain the flat hierarchical forest $F$ over $G$ that preserves distances from our dedicated source vertex. Our construction is fairly similar to the construction in \Cref{subsec:mappingHierachVS} except that we additionally need to map in multiple smaller steps instead of a single mapping operation (this is somewhat similar to the algorithm in \Cref{subsec:maintainLowDiam}).

To obtain $F$, we define a hierarchy of flat hierarchical forests $F_{5 \Lambda+1}, F_{5 \Lambda}, \ldots, F_0$ where each forest $F_i$ is a flat hierarchical forest over $V(G)$ in graph $G \cup \hat{H}_{< i}$, i.e. each edge in $F_i$ is a copy of an edge in graph $G$ or in $\hat{H}_{< i}$, i.e. an edge of a hopset induced by a covering at level $j < i$. Since $\hat{H}_{<0}$ is an empty graph, we have that $F_0$ consists of edges that are copies of edges in $G$, and thus $F_0$ is a flat hierarchical forest over $G$. We take $F = F_0$.

\begin{algorithm}
$F_{5\Lambda+1} \gets F_s$.\\
Define $\Delta_i \defeq (1 + \Delta^2 + \Delta \cdot \gamma_{lowDiamTree})^{5\Lambda + 1 - i}$ for all $0 \leq i \leq 5\Lambda + 1$.\\
\For{$i = 5\Lambda, 5\Lambda-1, \ldots, 0$}{
    \ForEach{$C \in \cC_i, j \in [0, \Delta_i)$}{
        Add to $F_i$ a copy $D_{C, j}$ of the graph $X_C$ that is obtained as the direct sum of forests $Y_C$ and $F'_C$ where we merge the nodes in $\im(\Pi_{V(C) \mapsto V(F'_C}))$ from $F'_C$ each of which is identified with a vertex in $G \cup H_{<i}$ with the corresponding vertices in $Y_C$ if they exist.
    }
    Construct a map $\Lambda_i = \Pi_{V(F_{i+1}) \mapsto [0, \Delta_i)}$ that maps each node $x$ from $V(F_{i+1})$ identified with a cover set $C \in \cC_i$ to a number $\Lambda_i(x)$ in $[0, \Delta_i)$ such that any two nodes $x,y \in V(F_{i+1})$ that are identified with the same cover set $C \in \cC_i$ in $G \cup \hat{H}_{< i+1}$ have $\Lambda(x) \neq \Lambda(y)$. \label{lne:constructMapSSSP}\\
    Add to $F_i$ all nodes from forest $F_{i+1}$ that are not identified with cover sets $C \in \cC_i$.\\
    \ForEach{$e = (x,y) \in E(F_{i+1})$}{
        \If{neither of the endpoints of $e$ is identified with a cover set $C \in \cC_i$}{
            Add $e$ directly to $F_i$.
        }\Else{
            Let w.l.o.g. $x$ be the node that is identified with a cover set $C \in \cC_i$.\\
            Merge the node $y$ in $F_i$ with the node in $D_{C,\Lambda(x)}$ that is identified with $y$ (if there are multiple such vertices, pick an arbitrary one).
        }
    }
    
}

\Return $F_0$.
\caption{$\textsc{MapApproximateShortestPathForest}()$}
\label{alg:mapShorestPathForest}
\end{algorithm}

\paragraph{Initializing the Approximate Shortest Path Forest.} In \Cref{alg:mapShorestPathForest}, we describe how to initialize the hierarchy of forests. Here, we define each forest $F_i$ recursively: we take $F_{5 \Lambda +1}$ to be equal to the shortest path forest $F_s$. Then for $i \leq 5 \Lambda$, the algorithm constructs $F_i$ from $F_{i+1}$. To this end, it first constructs a large number of graphs $X_C$ for every cover set $C \in \cC_i$. The $j$-th copy of this graph $X_C$ is denoted by $D_{C, j}$. Each graph $X_C$ is obtained from stitching together the low diameter forest $F'_C$ over the vertices in $C$ and the shortest path forest $Y_C$ that preserves distances between the set $C$ and the vertices in the ball of radius $\Theta(\stretchSSSP_{stretchCover}/\eps) \cdot r(C)$ around $C$. 

The algorithm then constructs a map $\Lambda_i$ that implicitly maps each node $x$ in $F_{i+1}$ that is identified with a cover set $C \in \cC_i$ to a unique copy $D_{C, \Lambda_i(x)}$ of graph $X_C$. Finally, the algorithm adds all nodes not identified with such cover sets from $F_{i+1}$ to $F_i$. It then maps the edge set of $F_{i+1}$ to $F_i$ by either adding the same edge if it was not incident to any node identified with a cover set $C \in \cC_i$, or otherwise, it simply merges the non-cover set endpoint into a node identified with the same vertex as itself in a copy of graph $X_C$. 

\paragraph{Maintaining the Approximate Shortest Path Forest.} The maintenance of the hierarchy of forest $F_{5 \Lambda+1}, F_{5 \Lambda}, \ldots, F_0$ and thus of $F = F_0$ is rather straightforward: every update to $F_s = F_{5 \Lambda + 1}$ can be forwarded directly recursively to every forest in the hierarchy, and only affects how a single edge is projected in the case of an edge insertion/deletion or affects a single entry of each map $\Lambda_i$ in case of an isolated vertex insertion/deletion.

Every update to a graph $Y_C$ or $F'_C$ leads only to a single change in the graph $X_C$. That is because the vertex map $\Pi_{V(C) \mapsto V(Y_C')}$ has the image of each element in the image fixed throughout the existence of the element, thus we have that edge insertions and deletion are processed with a single edge recourse, and each vertex insertion/deletion with at most recourse of one.

It is further not hard to see that the number of copies of such a graph $X_C$ for cover set $C \in \cC_i$ is exactly $\Delta_i$ which is then an upper bound on the number of changes this causes in the forests $F_{i}, F_{i-1}, F_{i-2}, \ldots, F_0$ because to these forests a change to each copy in $X_C$ can be processed just like the changes to $F_s$ by the entire hierarchy.

Finally, it is straightforward to maintain the vertex map $\Pi_{V(F_i) \mapsto V(G) \cup V(H_{< i})}$ for every $0 \leq i \leq 5\Lambda+1$ and thus, in particular, the vertex map $\Pi_{V(F) \mapsto V(G)}$ since $F = F_0$. We let $\Pi_{V(G) \mapsto V(F_i)}$ be defined as follows: we let $\Pi_{V(G) \mapsto V(F_{5\Lambda + 1})}$ simply be the identity map since $F_{5\Lambda + 1}$ is a real forest (not only a hierarchical forest) and thus only has each vertex in $V$ once (plus the vertices corresponding to cover sets). For $i \leq 5 \Lambda$, we have that $F_i$ is created from copies of graphs $X_C$ and the set of nodes in $F_{i+1}$ that are not identified with a cover set $C \in \cC_i$ and thus all nodes in $\im(\Pi_{V(G) \mapsto V(F_{i+1})})$ are present in $F_i$. Thus, we can obtain $\Pi_{V(G) \mapsto V(F_i)}$ from $\Pi_{V(G) \mapsto V(F_{i+1})}$ by using the identity map.

\paragraph{Analysis.} We now show that the algorithm described is a valid implementation of \Cref{thm:mainSSSP}. We start by proving correctness of the algorithm.

\begin{claim}
For every $0 \leq i \leq 5 \Lambda + 1$, we have that the vertex congestion of map $\Pi_{V(F_i) \mapsto V(G) \cup V(H_{<i})}$ is at most $\Delta_{i-1}$. We further show that we can maintain a map $\Lambda_i$ as described in \Cref{lne:constructMapSSSP}, i.e. we can map each node $x$ in $F_i$ identified with a cover set $C \in \cC_i$ to a unique graph $D_{C, \Lambda(x)}$. 
\end{claim}
\begin{proof}
We prove the claim by induction. We have that $F_s$ is a simple forest in the graph $G \cup H_{< 5\Lambda +1}$ by \Cref{thm:Apxball}, and thus so is $F_{5\Lambda +1}$ which we maintain equal to $F_s$ and thus, the vertex map $\Pi_{V(F_{5\Lambda +1}) \mapsto G \cup H_{< 5\Lambda +1}}$ is simply the identity map on $F_s$ and thus the vertex congestion of this map is $1 = (1+\Delta)^0 = \Delta_{5\Lambda + 1}$.

For $i \leq 5 \Lambda$, we have that by the inductive hypothesis, we have that the vertices in $F_{i+1}$ are mapped to vertices in $G \cup \hat{H}_{< i +1}$ with vertex congestion at most $\Delta_{i+1}$ and thus the map $\Lambda_i$ exists (and can be maintained rather straightforwardly). 

We then have that $F_i$ is obtained from the vertices in $F_{i+1}$ that are not identified with a cover set $C \in \cC_i$, and the $\Delta_i$ copies of graphs $Y_C$ and $Y_C'$ for every $C \in \cC_i$ and it then merges certain nodes (recall that $X_C$ is obtained from $Y_C$ and $Y_C'$ by merging vertices). However, merging nodes identified with the same vertex or cover set can only reduce the vertex congestion of the graph embedding.

We can thus derive an upper bound on the congestion by separately upper bounding the congestion induced by $F_{i+1}$, all the graphs $Y_C$, and all the graphs $Y_C'$ added to $F_i$. The congestion induced by $F_{i+1}$ is trivially bound by $\Delta_i$ by the induction hypothesis.

For the graphs $Y_C$, we note that each such forest is a simple forest in the graph $G_C \subseteq G \cup H_{<i}$ and thus contains each vertex in this graph at most once. In fact, by \Cref{def:Htil} and \Cref{def:Apxball}, we have that $G_C$ consists of vertices in $V(G)$ that are in the initial ball $B_G(C, \frac{\gamma_{stretchCover} \cdot 250 \log n}{\eps})$ but from \Cref{thm:covering}, we have that each vertex $v \in V(G)$ appears in at most $\Delta$ such balls. Further, $G_C$ contains vertices from the graph $H_{< i}$ that are incident to a vertex in the initial ball $B_G(C, \frac{\gamma_{stretchCover} \cdot 250 \log n}{\eps})$. This implies that we have a node in $G_C$ identified with a cover set $C'$ from covering $\cC_j$ for some $j < i$ if $C'$ contains a vertex $w$ that also appears in $B_G(C, \frac{\gamma_{stretchCover} \cdot 250 \log n}{\eps})$. But since each vertex $w \in V(G)$ appears in at most $\Delta$ such cover sets $C'$ for covering $\cC_j$, we have that each such cover set $C' \in \cC_j$ appears in at most $\Delta^2$ graphs $G_C$ overall. Thus, the congestion of all copies of graphs $Y_C$ in $F_i$ is at most $\Delta_i \cdot \Delta^2$. 

Finally, we have that each forest $F'_C$ is obtained on graph $G[C]$, and we have again from \Cref{thm:covering} that every vertex $w \in V(G)$ appears in at most $\Delta$ such cover sets $C$ and thus graphs $G[C]$ and thus the total congestion of all forests $F'_C$ is at most $\Delta_i \cdot \Delta \cdot \gamma_{lowDiamTree}$ by \Cref{thm:mainTheoremLowDiamTree}. 

We can thus conclude that the congestion of the vertex map from $F_i$ into $G \cup H_{<i}$ is at most $(1 + \Delta^2 + \Delta \cdot \gamma_{lowDiamTree}) \Delta_i = \Delta_{i-1}$. 
\end{proof}

Next, we provide a stretch analysis for forest $F$ (we obtain a slightly weaker guarantee on the $\epsilon$ here, however, by standard rescaling techniques this is w.l.o.g.).  

\begin{claim}\label{clm:recStretchAnalysisSSSP}
At any time, for every $v \in V$, if $\dist_G(s,v) < \infty$, and any $0 \leq i \leq 5 \Lambda + 1$, then 
\[
\dist_G(s,v) \leq \Pi_{{F_i} \mapsto G \cup \hat{H}_{< i}}(\pi_{F_i}(\Pi_{V(G) \mapsto V({F_i})}(s), \Pi_{V(G) \mapsto V({F_i})}(v))) \leq (1+\eps)^{5\Lambda + 2 - i} \dist_G(s,v)
\]
where the graph embedding of $F_i$ denoted by $\Pi_{{F_i} \mapsto G \cup \hat{H}_{< i}}$ is defined straightforwardly as every edge in $F_i$ is a copy of an edge in $G \cup \hat{H}_{< i}$.
\end{claim}
\begin{proof}
We prove the claim by induction. For $F_{5\Lambda + 1}$ the claim follows immediately from the guarantees from \Cref{def:Apxball} and \Cref{thm:Apxball} on the forest $F_s$ where we recall that we maintain $F_{5\Lambda+1} = F_s$.

Next, consider the forest $F_i$. We have by the induction hypothesis that the claim holds for the forests $F_{i+1}$. But note that by construction of $F_i$, we can map the path 
\[
\pi_{F_{i+1}}(\Pi_{V(G) \mapsto V({F_{i+1}})}(s), \Pi_{V(G) \mapsto V({F_{i+1}})}(v))
\]
in $F_{i+1}$ straightforwardly to a path in $F_i$ where each vertex on the path that is identified with a cover set $C \in \cC_i$ is removed with both its incident edges and the endpoints of both endpoints are instead merged into the corresponding vertices in $X_C$. Let $x$ and $y$ these endpoints and note that these endpoints are identified with vertices $u$ and $v$ in $V$ (by definition \Cref{def:Htil}). Then, we have, again from \Cref{def:Htil}, that the two edges incident to $C$ that were removed from the path had combined length of at least $2 \cdot \gamma_{stretchCover} \cdot r(C) + \dist_G(u, C) + \dist_G(v,C)$. But the path between the vertices $u$ and $v$ in graph $X_C$ is clearly of length at most $\gamma_{lowDiamTree} \cdot r(C) + (1+\epsilon) \cdot (\dist_G(u, C) + \dist_G(v, C))$ by the guarantees of the data structures from \Cref{thm:mainTheoremLowDiamTree} and \Cref{thm:Apxball} that maintain $Y_C$ and $Y_C'$, and the way that we merge vertices to obtain $X_C$. Thus, the path $\pi_{F_i}(\Pi_{V(G) \mapsto V({F_i})}(s), \Pi_{V(G) \mapsto V({F_i})}(v))$ has length (when mapped to $G \cup H_{<i}$) at most $(1+\eps)$ times the length of path $\pi_{F_{i+1}}(\Pi_{V(G) \mapsto V({F_{i+1}})}(s), \Pi_{V(G) \mapsto V({F_{i+1}})}(v))$ in $F_{i+1}$, as desired.
\end{proof}

Finally, we analyze the runtime of our algorithm.

\begin{claim}\label{clm:runtiemSSSP}
The algorithm runs with initialization and total update time $n \cdot e^{O(\log^{83/84} n \log\log n)}$.
\end{claim}
\begin{proof}
For our algorithm, we maintain for every $0 \leq i \leq 5 \Lambda$, a $(n^{i/\Lambda},K,\frac{\eps}{2000 \log(n)},\gamma_{stretchCover},\Delta = O(Kn^{2/K}))$-covering $\cC_i$ of $G$ and the approximate hopset $\hat{H}_i$ induced by this covering. These coverings and hopsets can be maintained by \Cref{thm:covering} and \Cref{thm:hopsetMaintenance} in time $n \cdot e^{O(\log^{83/84} n \log\log n)}$. 

Then, for every $0 \leq i \leq 5\Lambda$, we maintain data structures for every $C \in \cC_i$ to maintain the forests $Y_C$ and $Y_C'$. Let us first analyze the time to maintain all such forests $Y_C$. We have that for such a cluster $C$, the data structure to maintain $Y_C$ requires time $\Otil(|\bar{B}(C, \frac{\gamma_{stretchCover} 250 \log n}{\eps} \cdot r(C))$ where the ball is taken at the time that $C$ was first added to $\cC_i$. But we have that every vertex $v \in V(G)$ is in at most $\Delta$ cover sets in $\cC_i$ over all times by \Cref{thm:covering}. Thus, the total time required by all such data structures is at most $n \cdot \Delta \cdot e^{O(\log^{83/84} n \log\log n)} = n \cdot e^{O(\log^{83/84} n \log\log n)}$. The time to maintain the low-diameter forests $F'_C$ can be bound similarly where we use that the time to maintain such a forest on a particular cover set $C$ is at most $3C \cdot \gamma_{lowDiamTree}$ by \Cref{thm:mainTheoremLowDiamTree} where we use the fact that $G$ has maximum degree $3$.

It is not hard to see that we can maintain each graph $X_C$ from $Y_C$ and $F'_C$ in time linear in the number of changes to the two forests. 

Finally, we have that to maintain the forest $F_{5\Lambda + 1} = F_s$, we run an additional data structure from \Cref{thm:Apxball} that might explore the entire graph $G$ and thus take time $m \cdot e^{O(\log^{83/84} n \log\log n)}$. The forests $F_i$ can then be maintained rather straightforwardly in time $n \cdot \Delta_{i} \cdot e^{O(\log^{83/84} n \log\log n)}$ using our analysis of graphs $X_C$ above. Thus the total time to maintain all forests $F_{5\Lambda + 1}, F_{5\Lambda}, \ldots, F_0$ is at most $n \cdot (\Delta \gamma_{lowDiamTree})^{O(\Lambda)} \cdot e^{O(\log^{83/84} n \log\log n)} = n \cdot e^{O(\log^{83/84} n \log\log n)}$. The corresponding vertex maps and graph embeddings with forests $F_{5\Lambda + 1}, F_{5\Lambda}, \ldots, F_0$ can be maintained with only a constant number of additional operations.

This bounds the runtime required by all components of the algorithm by $n \cdot e^{O(\log^{83/84} n \log\log n)}$, as desired.
\end{proof}

Finally, we obtain the proof of \Cref{thm:mainSSSP} from the fact that $F = F_0$, \Cref{clm:recStretchAnalysisSSSP}, and that we can rescale $\eps$ by a logarithmic factor in $n$ without increasing the runtime significantly; from \Cref{clm:recStretchAnalysisSSSP} combined with the fact that the vertex congestion trivially upper bounds the edge congestion; and \Cref{clm:runtiemSSSP} which not only provides an upper bound on the runtime of the algorithm, but also the recourse of $F$ and the vertex maps and graph embedding, since these are maintained explicitly.

\pagebreak
\bibliographystyle{alpha}
\bibliography{refs}

\appendix 

\pagebreak

\section{Missing Proofs}

\subsection{Bunches and Clusters: Proof of \Cref{thm:TZschemes}}
\label{sec:bunchesAndClusters}
In this section, we prove \Cref{thm:TZschemes}, which we restate here for convenience, as \Cref{thm:TZschemesAppendix}.
\begin{theorem}[see \cite{thorup2001compact}, Theorem 3.1]\label{thm:TZschemesAppendix}
Given a graph constant degree graph $G=(V,E)$ with edge lengths $l$ and a size reduction parameter $b$, there is an algorithm $\textsc{Center}(G, b)$ that in time $\tilde{O}(mb)$ computes a set $A \subseteq V$ of size at most $n/b$ such that for every vertex $v \in V$, we have $|B_G(v, A)| \leq 2b \log n$ and $|C_G(v,A)| \leq 2b \log n$.
\end{theorem}

Consider a set $A$, and a collection of subsets $\mathcal{B} \subseteq 2^{A}$, where each $B \in \mathcal{B}$ satisfies $\abs{B} \geq d$.
We say that $H \subseteq A$ is a hitting set for $\mathcal{B}$ if for all $B \in \mathcal{B}$ there exists some $a \in B \cap H$.
We will frequently use a standard result on deterministically computing hitting sets.

\begin{lemma}[Deterministic Hitting Set]
\label{lem:detHS}
Using a deterministic algorithm, in time $\O(|A| + \sum_{B \in \mathcal{B}} |B|)$, we can compute a set hitting $H$ of $\mathcal{B}$ of size $\frac{|A|}{d} \log |\mathcal{B}|$.
\end{lemma}

\begin{proof}
Initially, we set $H_0 \gets \emptyset$.
We now consider the following procedure, starting with $i = 0$.
Consider a bipartite graph $G$ with vertex set $A_i =  A \setminus H_i$ on the left,
and $\mathcal{B}_i = \setof{B \in \mathcal{B} : B \cap H \neq \emptyset}$, and an edge $(a,B)$ iff $a \in B$.
Each vertex $B \in \mathcal{B}_i$ has degree $\geq d$, and hence the sum of degrees of vertices on the left is $\sum_{a \in A_i} \deg_G(a) \geq d |\mathcal{B}_i|$.
Thus, we must have a vertex $a \in A$ with degree
$\deg_G(a) \geq d |\mathcal{B}_i| / |A_i|$.
We define $H_{i + 1} = H_i \union \setof{a}$, 
and repeat the the procedure with $i \gets i+1$, until $\mathcal{B}_{i} = \emptyset$, at which point we conclude that $H = H_i$ is a hitting set, and $|H| = i$.
Furthermore, we have
\[
|\mathcal{B}_{i+1}|
\leq 
|\mathcal{B}_{i}|
(1 - d |\mathcal{B}_i| / |A_i|)
\leq 
|\mathcal{B}_{i}|
(1 - d / |A|)
\]
Thus,
$|\mathcal{B}_{i}| \leq (1 - d / |A|)^i |\mathcal{B}| $, and we must terminate with 
$\abs{H} = i \leq \frac{|A|}{d} \log |\mathcal{B}|$.

We now discuss how to implement the procedure in the stated time.
We can construct the initial bipartite graph in time
$O(|A| + \sum_{B \in \mathcal{B}} |B|)$,
and use a (Fibonacci) max-heap to maintain the vertex degrees for $a \in A_i$.
We then repeatedly extract a maximum degree vertex, delete its neighbors, and update degrees of vertices adjacent to the deleted neighbors, with total running time $\O(|A| + \sum_{B \in \mathcal{B}} |B|)$,. 
\end{proof}

\begin{proof}[Proof of \Cref{thm:TZschemesAppendix}]
First, we grow around each vertex $u$ a shortest distance ball 
$B(u) = \setof{ v \in V : d(u,v) < r_v }$, 
where we keep increasing $r_v$ until the ball contains exactly $b$ vertices, or
if equidistant vertices to $u$ make this impossible, we add an arbitrary subset of the vertices at distance $r_v$ to $B(u)$ until it ball has size exactly $b$.
We can do this in total time $\O(n b)$ by Dijkstra's algorithm.

Now, using \Cref{lem:detHS}, we compute a hitting set $W_0 \subseteq V$ for the set of balls $\mathcal{B} = \setof{B(v) : v \in V}$
with size $|W_0| \leq \frac{n}{b}\log(n) $.
By the definition of a bunch,
we must have that for all $u$, $B_G(u,W_0) \subseteq B(u)$,
and hence every bunch w.r.t. $W_0$ has size at most $b$.

We also want to bound the size of clusters.
We have a bound on the total cluster size of 
$\sum_{v \in V} 
|C_G(v,W_0)|
=
\sum_{v \in V} 
|B_G(v,W_0)|
\leq 
nb$.
Thus, by Markov's inequality, at least half of the $n$ clusters satisfy $|C_G(v,W_0)| \leq 2 b$.

Observe that for any sets $W, X \subseteq V$,
\[
\sum_{x \in X} 
|C_G(x,W)|
= 
\sum_{x \in X} 
\sum_{v \in V}
\mathbbm{1}_{[x \in B_G(v,W)]}
=
\sum_{v \in V} 
|B_G(v,W) \cap X|
\]

Furthermore, for any sets $W, W'$,
\begin{align}
\label{eq:monotoneBunchAndCluster}
C_G(v,W) \subseteq C_G(v,W \union W')
\text{ and }
B_G(v,W) \subseteq B_G(v,W \union W')
.
\end{align}

Let $V_1$ be the set of vertices whose clusters are strictly larger than $2b$ w.r.t. $W_{0}$.

We now repeat the following procedure, starting with $i = 1$, unless $|V_i| \leq n/b$. If $|V_i| \leq n/b$, we instead return $W = W_i \union V_i$ as our pivot set.
For each vertex $v$, grow a set
$B_i(u) = \setof{ v \in B(u) \cap V_i : d(u,v) < r_{v,i} }$,
where we keep increasing $r_{v,i}$ until the set contains exactly $k_i = \lfloor \frac{|V_i| b}{n} \rfloor$ vertices,
or 
if equidistant vertices to $u$ make this impossible, we add an arbitrary subset of $V_i$ at distance $r_{v,i}$ to $B_i(u)$ until the ball has size exactly $k_i$.
If $|B(u) \cap V_i| < k_i$, we instead define 
$B_i(u) = \emptyset$.

Now, using \Cref{lem:detHS}, we compute a hitting set $H_i \subseteq V_i$ for the collection of sets 
\[
\mathcal{B}_i = \setof{B_i(v) : v \in V \text{ and } B_i(v) \neq \emptyset}.
\]
As $B_i(v) \in \mathcal{B}_i$ has size $k_i$,
the hitting set has size $|H_i| \leq \frac{|V_i|}{k_i} \log n \leq 2 \frac{n}{b} \log n$.
We then define $W_i = W_{i-1} \union H_i$.
By \eqref{eq:monotoneBunchAndCluster}, we have 
$(B_G(v,W_i) \cap V_i) \subseteq B_i(v)$.

We conclude that 
\[
\sum_{v \in V_i} 
|C_G(v,W_i )|
=
\sum_{v \in V} 
|B_G(v,W_i) \cap V_i|
\leq 
n k_i 
\]

Thus, by Markov's Inequality, at least half of the $|V_i|$ vertices satisfy $|C_G(v,W_i )| \leq 2 n k_i/|V_i| \leq 2b$.
We now choose $V_{i+1}$ to be the subset of $V_i$, for which this cluster size upper bound fails to hold.
We then repeat the procedure above with $i \gets i+1$.

Finally, we obtain a pivot set $W$.
As $W$ includes all our hitting sets computed above, using \eqref{eq:monotoneBunchAndCluster}, we have 
for all $v \in V$, 
$|B_G(x,W)| \leq b$ and $|C_G(x,W)| \leq 2b$.
Each round above ensures $|V_{i+1}| \leq |V_{i}|/2$, and thus we finish in at most $\log n$ rounds.
Thus, overall $|W| \leq O(\frac{n}{b} \log^2n).$
Each round can be implemented by running Dijkstra in a ball of size $b$ around each vertex, followed by a hitting set computation, leading to an overall running time of $\O(nb)$.
\end{proof}

\end{document}